\documentclass[11pt,reqno,english]{amsart}
\PassOptionsToPackage{natbib=true}{biblatex}
\usepackage{lmodern}
\usepackage{lmodern}
\usepackage[T1]{fontenc}
\usepackage[utf8]{inputenc}
\usepackage{babel}
\usepackage{amsbsy}
\usepackage{amstext}
\usepackage{amsthm}
\usepackage{amssymb}
\usepackage{graphicx}
\usepackage{geometry}
\usepackage{setspace}
\usepackage[pdfusetitle,
 bookmarks=false,
 breaklinks=false,pdfborder={0 0 1},backref=false,colorlinks=false]
 {hyperref}

\makeatletter
\numberwithin{equation}{section}
\numberwithin{figure}{section}

\usepackage{babel}
\usepackage{csquotes}   

\usepackage{amsfonts}

\newcommand{\indep}{\mathop{\perp \! \! \! \perp}}

\theoremstyle{plain}
	\newtheorem{thm}{Theorem}[section]
	\newtheorem{lem}[thm]{Lemma}\newtheorem{prop}[thm]{Proposition}

\theoremstyle{definition}

\theoremstyle{remark}

\AddToHook{package/biblatex/after}{

\renewbibmacro*{name:andothers}{
  \ifboolexpr{
    test {\ifnumequal{\value{listcount}}{\value{liststop}}}
    and
    test \ifmorenames
  }
    {\ifnumgreater{\value{liststop}}{1}
       {\finalandcomma}
       {}%
     \andothersdelim\bibstring[\emph]{andothers}}
    {}}
  
  \DeclareFieldFormat{pages}{#1} 
  \renewbibmacro{in:}{\ifentrytype{article}{}{\printtext{\bibstring{in}\intitlepunct}}} 
}

\hypersetup{colorlinks,
            bookmarksdepth=3,
            citecolor=[RGB]{46,126,42},
            linkcolor=[RGB]{128,0,6},
            urlcolor=[RGB]{138,0,135}}
\usepackage{bookmark}

\usepackage[hyphenbreaks]{breakurl}

\makeatother

\theoremstyle{plain}
\newtheorem{assumption}[thm]{\protect\assumptionname}
\theoremstyle{remark}
\newtheorem{rem}[thm]{\protect\remarkname}
\usepackage[style=authoryear,sorting=nyt,dashed=false,maxbibnames=999,maxcitenames=2,uniquename=init,uniquelist=false,giveninits=true,date=year]{biblatex}
\providecommand{\assumptionname}{Assumption}

\providecommand{\remarkname}{Remark}

\begin{document}
\title[Nonseparable first stage]{Closed-form estimation and uniform inference in additively separable
triangular models with a nonseparable first stage}
\author[K.~Sunada]{Keita Sunada}
\address{Aarhus Center for Econometrics, Aarhus University, Universitetsbyen
51, 8000 Aarhus C, Denmark.}
\email{\href{mailto:ksunada@econ.au.dk}{ksunada@econ.au.dk}}
\thanks{I am grateful to Nese Yildiz and Bin Chen for their helpful comments.
This research was supported in part by the Aarhus Center for Econometrics
(ACE) funded by the Danish National Research Foundation grant number
DNRF186.}
\date{August 23, 2026}
\begin{abstract}
This paper studies the nonparametric identification and estimation
of additively separable triangular models with continuous endogenous
and instrumental variables, allowing for a nonseparable first-stage
equation. Under the independence of instrumental variables and unobservables,
we show that the outcome function possesses a closed-form expression
as a functional of conditional cumulative distribution functions.
The resulting plug-in estimators require no regularization and converge
at the rate $n^{-m/(2m+1)}$, where $m$ is the order of smoothness
the model imposes. Also, no estimator of the outcome function converges
faster. We use the empirical bootstrap to construct a uniform confidence
band that covers the outcome function at every point of a compact
set simultaneously.
\end{abstract}

\maketitle

\section{Introduction}

This paper studies the identification and estimation of a system of
structural equations that takes the following form: 
\begin{align}
Y & =g\left(X\right)+\varepsilon,\label{eq:model}\\
X & =h\left(Z,\eta\right),\nonumber 
\end{align}
where $X$ is a continuous scalar explanatory variable (treatment),
$\varepsilon$ and $\eta$ are scalar unobservables, $Z$ is a continuous
scalar instrumental variable, and $g$ and $h$ are unknown functions.
We assume that instrument $Z$ and unobservables $\left(\varepsilon,\eta\right)$
are independent.

A key feature of (\ref{eq:model}) is that the first-stage equation
$X=h(Z,\eta)$ allows for nonseparable interactions between the instrument
and the unobservable. This is empirically relevant because in many
applications, the validity of an instrument is derived from economic
theory, and the resulting first-stage relationship is inherently nonseparable.
For instance, in demand estimation, researchers commonly use demand
shifters or product characteristics as instruments for price. These
instruments affect price through the markup implied by the firm's
first-order condition, and without strong functional-form assumptions,
the instrument and the unobservable can interact in a nonseparable
way. More generally, when instruments arise from equilibrium conditions
of structural models, separability of the first stage is the exception
rather than the rule. 

Our contributions are threefold. First, we show that $g$ has a closed-form
expression: it is identified as a functional of the conditional distribution
functions $F_{Y|X,Z}$ and $F_{X|Z}$. Second, replacing those two
functions by kernel estimators gives a plug-in estimator of $g$ that
needs no regularization and converges at the rate $n^{-m/(2m+1)}$
at a point, and is asymptotically normal at a slightly smaller bandwidth.
We also show that no estimator of $g$ converges faster, so the nonseparable
first stage does not lower the rate: it is the rate that would be
available if the endogeneity were absent. Third, we use the empirical
bootstrap to construct a uniform confidence band that covers $g$
at every point of a compact subset of the interior of the support
of $X$ simultaneously, the subset omitting the normalization point.

In the literature on identification and estimation of nonparametric
structural models, a closely related work is presented by \citet{NPV99-ECTA}.
Their estimation method relies on the additive separability of first-stage
equation $X=h\left(Z\right)+\eta$ under the mean independence condition
$\mathbb{E}\left[\varepsilon|\eta,Z\right]=\mathbb{E}\left[\varepsilon|\eta\right]$.
Thus their results are not applicable to (\ref{eq:model}). Their
estimator attains the slower of the two rates that \citet{Sto82-AoS}
shows to be optimal for the second stage and for the reduced form
separately. In model (\ref{eq:model}) both stages are univariate
and the same $m$ derivatives are assumed of each, so those two rates
coincide: theirs is $n^{-2m/(2m+1)}$ in mean square, whose square
root is the rate our estimator attains at a point. Separability of
the first stage therefore does not improve the rate at which $g$
is estimated.\footnote{Assumption \ref{assu:rate} allows a nonempty bandwidth window only
when the model has $m\ge4$ derivatives and the kernel is of order
$m$. The comparison therefore holds at $m\ge4$. Below that no bandwidth
sequence meets the assumption.} Fully nonseparable systems have also been investigated \citep{Che03-ECTA,IN09-ECTA,DF15-ECTA,Tor15-ECTA,Ish21-ET}.
Notably, \citet{IN09-ECTA} focus on quantile, average, and policy
effects rather than the direct estimation of $g$. \citet{DF15-ECTA}
and \citet{Tor15-ECTA} provide sufficient conditions for point-identification
of the nonseparable function $g(x,w,e)$ using binary instruments.
However, they do not offer clear guidance for estimation. Additionally,
\citet[Theorem S1]{Tor15-ECTA} establishes point-identification of
the function $g$ using continuous $Z$ in a non-constructive manner.
Since the model (\ref{eq:model}) emerges as a special case of the
model in \citet[Theorem S1]{Tor15-ECTA} under a level-normalization,
our first contribution lies not in establishing the identification
of $g$ but in providing the closed-form expression for $g$. \citet{torgovitsky2017minimum}
provides an estimation method for nonseparable triangular systems,
but he restricts his attention to the case that the function $g(x,w,e)$
is unknown up to a finite dimensional parameter.

Our identification result is close in spirit to those obtained by
\citet{Che03-ECTA} and \citet{CKK15-JoE}. \citet{Che03-ECTA} provides
the identification of the derivative of $g$ under a nonseparable
triangular system, which shares similarities with the form presented
in Proposition \ref{prop:identification} below.

Our estimator is more directly indebted to \citet{CKK15-JoE}. They
study nonparametric transformation models, a different model identified
under different conditions, but the estimator they build from their
identification result is one we can reuse. The difference is that
their identification result gives $g$ itself while ours gives $g^{\prime}$,
so that $g(x)$ is recovered by integrating from 0 to $x$. The endpoint
of that integral is not averaged away, and Proposition \ref{prop:lower_bound}
shows that no estimator of $g(x)$ reaches the parametric rate. Their
estimator is asymptotically linear, and ours cannot be, since asymptotic
linearity would give the parametric rate. The two asymptotic arguments
therefore share little beyond the averaging step. What is new here
is the identification result that supplies the expression to average,
and the bootstrap and the band of Section \ref{sec:Estimation-and-bootstrap}.

In the absence of assistance from the first-stage equations, one can
estimate $g$ by solving an integral equation $\mathbb{E}\left[Y|Z\right]=\mathbb{E}\left[g(X)|Z\right]$
with respect to $g$ under the mean independence condition $\mathbb{E}\left[\varepsilon|Z\right]=0$.
Given a well-known completeness condition, there exists a unique solution
to this problem. However, this approach necessitates a regularization
scheme to handle the ill-posed inverse problem, which leads to a slower
convergence rate of estimators in general \citep{NP03-ECTA,BCK07-ECTA,CFR07-Handbook,DFFR11-ECTA,Hor14-AnnRevEcon}.
Uniform confidence bands are available in that literature: \citet{HL12-JoE}
construct them from bootstrap critical values, \citet{CC18-QE} obtain
them for nonlinear functionals of a sieve estimator, and \citet{CCK25-ReStud}
choose the sieve dimension adaptively; see also \citet{Babii20-ET}.
All of them inherit the ill-posedness of the inverse problem, so the
bands rest on estimators whose rate it degrades; the band of Section
\ref{sec:Estimation-and-bootstrap} does not.

The remainder of the paper is organized as follows. Section \ref{sec:Identification}
provides the identification result. Section \ref{sec:Estimation-and-bootstrap}
proposes nonparametric estimators based on this identification result,
and studies their asymptotic behaviors. Section \ref{sec:Monte-carlo}
reports a Monte Carlo study. All proofs are deferred to the Appendix.

\section{Identification }\label{sec:Identification}

Let $\mathcal{A}$ be the support of random variable $A$, and $\mathcal{A}_{b}$
the conditional support of $A$ given $B=b$. Let $\mathcal{A}^{\circ}$
be the interior of $\mathcal{A}$. We assume that $\mathcal{Y},\mathcal{X},\mathcal{Z}\subset\mathbb{R}$.
We consider the model in (\ref{eq:model}) and impose the following
assumptions:
\begin{assumption}
\label{assu:rv}(i) The random variables $Y|X=x,Z=z$, $X|Z=z$, $Z$
are absolutely continuously distributed for each $x$, $z$. (ii)
$\mathcal{X}$ is convex. 
\end{assumption}

This assumption is imposed by \citet[Assumption C]{Tor15-ECTA}. The
first condition requires $Y$, $X$ and $Z$ to be continuously distributed.
It covers $\eta$ as well, though it does not name it: $h(z,\cdot)$
is one-to-one under Assumption \ref{assu:IV}(ii), so an atom of $\eta$
at $a$ would put one of $X$ given $Z=z$ at $h(z,a)$. The proof
of Proposition \ref{prop:identification} uses $V(\cdot|z)$, which
is continuous for that reason.
\begin{assumption}
\label{assu: func_g} (i) $g$ is continuously differentiable on $\mathcal{X}$,
the derivative being one-sided at a boundary point. (ii) $0\in\mathcal{X}^{\circ}$
such that $g(0)=0$. 
\end{assumption}

The first condition is imposed by \citet[Assumption G]{Tor15-ECTA}.
The second condition is a level normalization, which requires that
the constant be zero when the model is linear. This type of normalization
is also necessary for identification in nonseparable models.
\begin{assumption}
\label{assu:IV} (i) The instrument is independent of the unobservables:
$\left(\varepsilon,\eta\right)\indep Z$. (ii) The function $h(z,\cdot)$
is strictly increasing for each $z$. 
\end{assumption}

The first condition requires a full independence between instruments
and error terms, which is imposed by \citet{IN09-ECTA}, \citet{Tor15-ECTA},
and \citet{DF15-ECTA}. \citet{NPV99-ECTA} impose a weaker mean independence
condition $\mathbb{E}\left[\varepsilon|Z,\eta\right]=\mathbb{E}\left[\varepsilon|\eta\right]$.
Our identification result exploits the full independence, and cannot
be weakened to this mean independence. The monotonicity imposed in
the second condition is standard in the literature.
\begin{assumption}
\label{assu: local} Let $V(x|z)=F_{X|Z}(x|z)$. For every $x$ in
a dense subset $\mathcal{X}_{d}$ of $\mathcal{X}$ and $y\in\mathcal{Y}^{\circ}_{x}$,
there exists $z\in\mathcal{Z}^{\circ}$ such that (i) $y\in\mathcal{Y}^{\circ}_{x,z}$
and $x\in\mathcal{X}^{\circ}_{z^{\prime}}$ for every $z^{\prime}$
in a neighborhood of $z$, (ii) $V(x|\cdot)$ is continuously differentiable
on a neighborhood of $z$ with $\nabla_{z}V(x|z)\neq0$, and $\nabla_{x}V(x|z)$
exists, (iii) $\nabla_{y}F_{Y|X,Z}(\cdot|x,\cdot)$ and $\nabla_{z}F_{Y|X,Z}(\cdot|x,\cdot)$
exist and are continuous on a neighborhood of $(y,z)$, and $\nabla_{x}F_{Y|X,Z}(y|x,z)$
exists; write $f_{Y|X,Z}(y|x,z):=\nabla_{y}F_{Y|X,Z}(y|x,z)$. (iv)
$f_{Y|X,Z}(y|x,z)>0$. 
\end{assumption}

The first and second conditions strengthen those imposed by \citet[Theorem S1]{Tor15-ECTA}.
The third and fourth are additionally needed. Conditional distributions
are determined only up to null sets, while the four conditions above
and the conclusion of Proposition \ref{prop:identification} are statements
at a point. Throughout, $F_{Y|X,Z}$ and $V$ denote the versions
continuous in their conditioning arguments, which is what parts (ii)
and (iii) are conditions on. Assumption \ref{assu: density} asks
the same on the region where estimation takes place. The conditions
are placed on a neighborhood of $z$ rather than at $z$ alone because
the proof inverts $z\mapsto V(x|z)$ near $z$ and then differentiates
a composition in which $x$ enters through both arguments; the two
partial derivatives at the point would not suffice for that. 
\begin{assumption}
\label{assu:support} The supports satisfy $\mathcal{Y}_{x,z}=\mathcal{Y}$
and $\mathcal{X}_{z}=\mathcal{X}$ for each $x$ and $z$.
\end{assumption}

Proposition \ref{prop:identification} itself does not need this assumption.
Under it the support conditions of Assumption \ref{assu: local} (i)
hold at every $(x,y,z)$ with $x\in\mathcal{X}^{\circ}$ and $y\in\mathcal{Y}^{\circ}$,
so the identified expression holds at every $(y,z)$ at which the
derivatives in parts (ii) and (iii) of that assumption exist and $\nabla_{z}V(x|z)$
and $f_{Y|X,Z}(y|x,z)$ are nonzero. Assumption \ref{assu: density}
below imposes those conditions on the region over which the estimators
of Section \ref{sec:Estimation-and-bootstrap} average. Averaging
over $\left(y,z\right)$ reduces the asymptotic variance. 

Our first result is stated as follows:
\begin{prop}
\label{prop:identification}Suppose Assumptions \ref{assu:rv}--\ref{assu:IV}
hold. Let $(x,y,z)\in\mathcal{X}\times\mathcal{Y}\times\mathcal{Z}$
satisfy conditions (i) to (iii) of Assumption \ref{assu: local}.
Then
\begin{equation}
\nabla_{x}g(x)f_{Y|X,Z}(y|x,z)=\frac{\nabla_{x}V(x|z)}{\nabla_{z}V(x|z)}\nabla_{z}F_{Y|X,Z}(y|x,z)-\nabla_{x}F_{Y|X,Z}(y|x,z).\label{eq:identification_prod}
\end{equation}
If condition (iv) also holds, the derivative of the outcome function
$g$ is identified at $x$ as 
\begin{equation}
\nabla_{x}g(x)=\frac{\frac{\nabla_{x}V(x|z)}{\nabla_{z}V(x|z)}\nabla_{z}F_{Y|X,Z}(y|x,z)-\nabla_{x}F_{Y|X,Z}(y|x,z)}{f_{Y|X,Z}(y|x,z)}.\label{eq:identification}
\end{equation}
\end{prop}

\begin{proof}
See Appendix \ref{sec:lemma1}.
\end{proof}

\begin{rem}
If Assumption \ref{assu: local} holds, such a $(y,z)$ exists for
every $x\in\mathcal{X}_{d}$. Since $\nabla_{x}g$ is continuous and
$\mathcal{X}_{d}$ is dense in $\mathcal{X}$, the derivative is determined
on all of $\mathcal{X}$ by its values on $\mathcal{X}_{d}$. Determination
is not the same as an expression: (\ref{eq:identification}) holds
only where conditions (i) to (iv) hold, and condition (i) puts those
$x$ in $\mathcal{X}^{\circ}$, so it is never available at the edge
of the support. $\blacksquare$
\end{rem}

\begin{rem}
Under Assumptions \ref{assu:rv}--\ref{assu:IV} and parts (i) and
(ii) of Assumption \ref{assu: local}, the additively separable model
considered in this setting is a special case of the nonseparable models
studied by \citet{Tor15-ECTA}. Since \citet[Theorem S1]{Tor15-ECTA}
provides the identification of $g(x,e)$, the identification of $\nabla_{x}g(x)$
is also known where his conditions hold. However, his identification
of $g(x,e)$ is not constructive. Thus the contribution of Proposition
\ref{prop:identification} lies in providing the closed-form expression.
$\blacksquare$
\end{rem}

\begin{rem}
Model (\ref{eq:model}) has no covariates. Discrete covariates $W$
need no change: everything below applies within each cell $\{W=w\}$,
so $g(\cdot,w)$ is identified and estimated cell by cell, with the
rate governed by the cell sizes. Continuous covariates are a different
matter, since they enter the kernel estimators and raise their dimension;
the rate falls with each one added. $\blacksquare$
\end{rem}

As a corollary of Proposition \ref{prop:identification}, we have
$g(x)=S(y,x,z)$ where 
\begin{equation}
S(y,x,z)=\int^{x}_{0}\frac{\frac{\nabla_{x}V(u|z)}{\nabla_{z}V(u|z)}\nabla_{z}F_{Y|X,Z}(y|u,z)-\nabla_{x}F_{Y|X,Z}(y|u,z)}{f_{Y|X,Z}(y|u,z)}\mathrm{d}u\label{eq:g_functional}
\end{equation}
at every $(y,z)$ such that conditions (i) to (iv) of Assumption \ref{assu: local}
hold at $(u,y,z)$ for every $u$ between $0$ and $x$. Assumption
\ref{assu:support} makes the supports invariant, and Assumption \ref{assu: density}
of Section \ref{sec:Estimation-and-bootstrap} places the remaining
conditions on the region over which the estimators average. The step
from $\nabla_{u}g$ to $g$ uses Assumption \ref{assu:rv}(ii), so
that the segment from $0$ to $x$ lies in $\mathcal{X}$, and Assumption
\ref{assu: func_g}, which makes the integrand continuous and fixes
the level. Note that the closed form holds at any triple satisfying
conditions (i) to (iv), whether or not $u$ lies in $\mathcal{X}_{d}$;
the dense subset enters Assumption \ref{assu: local} only through
the existence of a suitable $z$. We shall propose nonparametric estimators
based on this expression.

\section{Estimation and bootstrap inference }\label{sec:Estimation-and-bootstrap}

\subsection{Estimation}

Suppose we have a random sample $\left\{ Y_{i},X_{i},Z_{i}\right\} ^{n}_{i=1}$
drawn from the model. Let $\widehat{S}(y,x,z)$ be a preliminary nonparametric
estimator of $S(y,x,z)$. We follow \citet{CKK15-JoE} and consider
two estimators. The first one is a least-squares type estimator: 
\[
\widehat{g}_{\mathrm{LS}}(x):=\int\int\mu(y,z)\widehat{S}(y,x,z)\mathrm{d}y\mathrm{d}z,
\]
where $\mu(y,z)$ is a weighting function. The second one is a smoothed
version of the least absolute deviation (LAD) type estimator by \citet{horowitz1998bootstrap}:
\begin{align*}
\hat{g}_{\mathrm{LAD}}(x) & =\arg\min_{q}Q_{b}(q|\hat{S}(\cdot,x,\cdot)),\\
Q_{b}(q|\hat{S}(\cdot,x,\cdot)) & =\int\int\mu(y,z)\left\{ \hat{S}(y,x,z)-q\right\} \left\{ 2\Lambda_{b}\left[\hat{S}(y,x,z)-q\right]-1\right\} \mathrm{d}y\mathrm{d}z,
\end{align*}
with $\Lambda_{b}:=\Lambda(\cdot/b)$ for a known distribution function
$\Lambda$ with median zero and some bandwidth $b>0$. $\hat{g}_{\mathrm{LAD}}$
is expected to be more robust to outliers in the sample than $\widehat{g}_{\mathrm{LS}}$.
\begin{rem}
\label{rmk:mu} Under Assumptions \ref{assu:support} and \ref{assu: density},
conditions (i) to (iv) of Assumption \ref{assu: local} hold at $(u,y,z)$
for every $y\in\mathcal{Y}_{\mu}$, $z\in\mathcal{Z}_{\mu}$ and $u\in\widetilde{\mathcal{X}}_{0}$,
so $S(y,x,z)=g(x)$ at every $(y,z)\in\mathcal{Y}_{\mu}\times\mathcal{Z}_{\mu}$
and every $x\in\mathcal{X}_{0}$. With $\int\int\mu=1$ from Assumption
\ref{assu: density}(ii) this gives $\int\int\mu(y,z)S(y,x,z)\mathrm{d}y\mathrm{d}z=g(x)$.
The estimand does not depend on $\mu$, which enters only through
the variance, much as the weight matrix does in generalized method
of moments. The same holds for $\hat{g}_{\mathrm{LAD}}$: the population
objective is $\varrho(g(x)-q)$ for $\varrho(t)=t\{2\Lambda_{b}(t)-1\}$,
which Assumption \ref{assu: kernel}(ii) makes uniquely minimized
at $q=g(x)$, whatever the weight. $\blacksquare$
\end{rem}

To estimate $S(y,x,z)$ nonparametrically, we use a kernel-based estimator.
First, observe that 
\[
F_{Y\mid X,Z}(y|x,z)=\frac{\int^{y}_{-\infty}f_{Y,X,Z}(u,x,z)du}{f_{X,Z}(x,z)}=:\frac{\Phi(y,x,z)}{f(x,z)}.
\]
Then a natural estimator of the conditional distribution $F_{Y\mid X,Z}(y|x,z)$
is $\widehat{F}(y|x,z)=\widehat{\Phi}(y,x,z)/\widehat{f}(x,z)$ where
\begin{align*}
\widehat{\Phi}(y,x,z) & =\frac{1}{n}\sum^{n}_{i=1}\frac{1}{h_{x}h_{z}}\boldsymbol{1}\left\{ Y_{i}\le y\right\} K\left(\frac{X_{i}-x}{h_{x}}\right)K\left(\frac{Z_{i}-z}{h_{z}}\right),\\
\widehat{f}(x,z) & =\frac{1}{n}\sum^{n}_{i=1}\frac{1}{h_{x}h_{z}}K\left(\frac{X_{i}-x}{h_{x}}\right)K\left(\frac{Z_{i}-z}{h_{z}}\right).
\end{align*}
As an estimator for the density $\frac{\partial}{\partial y}F_{Y|X,Z}=f_{Y|X,Z}$,
we consider $\widehat{\Phi}_{y}(y,x,z)/\widehat{f}(x,z)$ where 
\[
\widehat{\Phi}_{y}(y,x,z)=\frac{1}{n}\sum^{n}_{i=1}\frac{1}{h_{y}h_{x}h_{z}}K\left(\frac{Y_{i}-y}{h_{y}}\right)K\left(\frac{X_{i}-x}{h_{x}}\right)K\left(\frac{Z_{i}-z}{h_{z}}\right).
\]
A kernel-based estimator of the conditional distribution $V(x|z)$
is obtained by $\widehat{V}(x|z)=\widehat{\Psi}(x,z)/\widehat{p}(z)$
where 
\[
\widehat{\Psi}(x,z)=\frac{1}{n}\sum^{n}_{i=1}\frac{1}{h_{z}}\boldsymbol{1}\left\{ X_{i}\le x\right\} K\left(\frac{Z_{i}-z}{h_{z}}\right),\qquad\widehat{p}(z)=\frac{1}{n}\sum^{n}_{i=1}\frac{1}{h_{z}}K\left(\frac{Z_{i}-z}{h_{z}}\right).
\]

Finally a kernel based estimator of $S(y,x,z)$ is obtained by 
\begin{align*}
\widehat{S}(y,x,z) & =\int^{x}_{0}\frac{\frac{\nabla_{x}\widehat{V}(u|z)}{\nabla_{z}\widehat{V}(u|z)}\nabla_{z}\widehat{F}(y|u,z)-\nabla_{x}\widehat{F}(y|u,z)}{\nabla_{y}\widehat{F}(y|u,z)}\mathrm{d}u.
\end{align*}

We employ the following assumptions on the model, the kernel function,
and the weight function.
\begin{assumption}
\label{assu: kernel} (i) The univariate kernel $K$ is differentiable,
and there exist constants $C>0$ and $\nu>m+1$ such that $\left|K^{(i)}(z)\right|\le C\left|z\right|^{-\nu}$,
$\left|K^{(i)}(z)-K^{(i)}(z^{\prime})\right|\le C\left|z-z^{\prime}\right|$,
and $K^{(i)}$ is of bounded variation, for $i=0,1$, where $K^{(i)}(z)$
denotes $i$-th derivative of $K$. Furthermore, there exists $m\ge2$
such that $\int_{\mathbb{R}}K(z)\mathrm{d}z=1$, $\int_{\mathbb{R}}z^{j}K(z)\mathrm{d}z=0$
for $1\le j\le m-1$, and $\int_{\mathbb{R}}\left|z\right|^{m}\left|K(z)\right|\mathrm{d}z<\infty$.
(ii) The distribution function $\Lambda$ is strictly increasing and
three times continuously differentiable with bounded derivatives,
its density satisfies $\Lambda^{\prime}(0)>0$, and the bandwidth
$b>0$ for $\Lambda_{b}=\Lambda(\cdot/b)$ is held fixed. 
\end{assumption}

The differentiability condition is required to ensure the existence
of derivatives of $\widehat{F}(y|x,z)$ and $\widehat{V}(x|z)$. Hence
we rule out uniform and Epanechnikov kernels. Part (ii) concerns $\hat{g}_{\mathrm{LAD}}$
alone. The condition $\Lambda^{\prime}(0)>0$ makes the smoothed objective
behave like an absolute-deviation criterion at its minimum, and it
is the only feature of $\Lambda$ that survives into the first-order
asymptotics: as the end of Appendix \ref{sec:thm1} shows, $\Lambda^{\prime}(0)/b$
enters the score and the Hessian in the same way and cancels between
them, so $b$ does not appear in the first-order asymptotics and need
not shrink with $n$. \citet{CKK15-JoE} hold $b$ fixed for the same
reason. Strict monotonicity is used elsewhere: with median zero it
makes the population objective in Appendix \ref{sec:thm1} uniquely
minimized at $g(x)$. The standard normal distribution used in Section
\ref{sec:Monte-carlo} satisfies both conditions.
\begin{assumption}
\label{assu: density} Let $\mathcal{X}_{0}\subset\mathcal{X}^{\circ}$
be a compact set with $0\notin\mathcal{X}_{0}$, let $\widetilde{\mathcal{X}}_{0}$
be the smallest interval containing $\mathcal{X}_{0}\cup\{0\}$, and
write $\mathcal{Y}_{\mu}\times\mathcal{Z}_{\mu}$ for the support
of $\mu$. (i) The joint density $f_{Y,X,Z}(y,x,z)$ is bounded and
$m$-times differentiable with bounded derivatives. There are open
sets $\mathcal{U}\supset\widetilde{\mathcal{X}}_{0}$ and $\mathcal{W}\supset\mathcal{Z}_{\mu}$
whose closures are compact subsets of $\mathcal{X}^{\circ}$ and $\mathcal{Z}^{\circ}$
such that $\int\sup_{(x,z)\in\mathcal{U}\times\mathcal{W}}\left|\partial^{j_{1}}_{x}\partial^{j_{2}}_{z}f_{Y,X,Z}(y,x,z)\right|\mathrm{d}y<\infty$
for $j_{1}+j_{2}\le m$. (ii) The weight function $\mu(y,z)$ is $m$-times
continuously differentiable for both arguments with compact support
$\mathcal{Y}_{\mu}\times\mathcal{Z}_{\mu}$ contained in $\mathcal{Y}^{\circ}\times\mathcal{Z}^{\circ}$
and with nonempty interior, and $\int\int\mu(y,z)\mathrm{d}y\mathrm{d}z=1$.
(iii) On a neighborhood of $\mathcal{Y}_{\mu}\times\widetilde{\mathcal{X}}_{0}\times\mathcal{Z}_{\mu}$
the derivatives $\nabla_{x}V$, $\nabla_{z}V$, $\nabla_{x}F_{Y|X,Z}$,
$\nabla_{y}F_{Y|X,Z}$ and $\nabla_{z}F_{Y|X,Z}$ exist and are continuous,
and the denominators are bounded away from zero: $\inf f(x,z)>0$,
$\inf f_{Y|X,Z}(y|x,z)>0$, $\inf\left|\nabla_{z}V(x|z)\right|>0$
and $\inf p(z)>0$, the infima taken over $(y,x,z)\in\mathcal{Y}_{\mu}\times\widetilde{\mathcal{X}}_{0}\times\mathcal{Z}_{\mu}$. 
\end{assumption}

The integrability in (i) allows differentiating under the integral
sign, which the proofs of Lemma \ref{lem:boundary} and Lemma \ref{lem:plugin}
both do. Part (iii) makes conditions (ii) to (iv) of Assumption \ref{assu: local}
hold at every point of the region, which the discussion of Assumption
\ref{assu:support} defers to it. The tail exponent is taken above
$m+1$. It makes $\int K^{2}$, $\int\left|v\right|K(v)^{2}\mathrm{d}v$,
$\int\left|K\right|^{3}$ and $\int\left|K\right|^{4}$ finite, the
first two appearing in the asymptotic variance and in its distance
from the variance at a finite $n$, and it bounds the kernel away
from its centre wherever Appendix \ref{sec:thm1} needs that. The
kernel is not assumed to have compact support, so a remainder is left
from outside the region on which the density is smooth, and $\nu\ge m+1$
keeps it within the $O(h^{m}_{x})$ of the smoothing bias. Note that
the number of derivatives $m$ must align with the order of the kernel
function. Part (iii) collects the denominators. The estimators below
are ratios: $\widehat{F}=\widehat{\Phi}/\widehat{f}$ and $\widehat{V}=\widehat{\Psi}/\widehat{p}$
are ratios of kernel averages, and the integrand of (\ref{eq:g_functional})
divides in turn by $f_{Y|X,Z}$ and by $\nabla_{z}V$. A lower bound
on each turns uniform convergence of the kernel estimators into uniform
convergence of the ratios, on $\mathcal{Y}_{\mu}\times\widetilde{\mathcal{X}}_{0}\times\mathcal{Z}_{\mu}$.
The domain is $\widetilde{\mathcal{X}}_{0}$ rather than $\mathcal{X}_{0}$
because (\ref{eq:g_functional}) integrates from the normalization
point to $x$: estimating $g$ at a point of $\mathcal{X}_{0}$ uses
the integrand at every $u$ between $0$ and $x$, and those $u$
need not lie in $\mathcal{X}_{0}$. Assumption \ref{assu:rv} (ii)
makes $\mathcal{X}$ convex and Assumption \ref{assu: func_g} (ii)
places the normalization point in its interior, so $\widetilde{\mathcal{X}}_{0}$
is again a compact subset of $\mathcal{X}^{\circ}$. The restriction
to a compact $\mathcal{X}_{0}$ strictly inside $\mathcal{X}$ is
not a technicality: at the edge of the support of $X$ the bounds
on $\left|\nabla_{z}V\right|$ and on $f$ both fail.
\begin{assumption}
\label{assu:rate}$h_{v}=c_{v}n^{-a_{v}}\ (a_{v}>0)\text{ for }v\in\{y,x,z\}$,
$\sqrt{nh_{x}}h^{m}_{y}\to0$, $\sqrt{nh_{x}}h^{m}_{z}\to0$, $\log n/(\sqrt{n}h_{y}h^{1/2}_{x}h_{z})\to0$,
$\log n/(\sqrt{n}h^{5/2}_{x}h_{z})\to0$, and $\log n/(\sqrt{n}h^{1/2}_{x}h^{3}_{z})\to0$. 
\end{assumption}

Write $\bar{h}:=\max\{h_{y},h_{x},h_{z}\}$ for the largest of the
three. The first restriction fixes the bandwidths at powers of $n$.
Section \ref{sec:Monte-carlo} uses that, and it makes $\log(1/h_{x})$
of order $\log n$, which the entropy bounds of Appendix \ref{sec:proof_band}
need. The next two make the smoothing bias of $(\hat{\Phi},\hat{\Psi},\hat{f},\hat{p})$
and of their derivatives $o((nh_{x})^{-1/2})$. Nothing is imposed
on $h^{m}_{x}$: the bias in the direction of $x$ is the estimator's
own rather than a nuisance error, and Theorems \ref{thm:pointwise}
and \ref{thm:band} control it by undersmoothing. The last three conditions
make the squares and products of the estimation errors $\widehat{\Phi}-\Phi,\widehat{\Psi}-\Psi,\widehat{f}-f,\widehat{p}-p$
and of the errors in their derivatives $o((nh_{x})^{-1/2})$. Squares
and products are the form in which these errors enter: $\widehat{S}$
is built from the ratios $\widehat{\Phi}/\widehat{f}$ and $\widehat{\Psi}/\widehat{p}$,
and expanding a ratio leaves a remainder quadratic in them. No condition
is placed on the errors themselves, which are of order $\sqrt{\log n/(nh_{y}h_{x}h_{z})}$
or larger and so are not $o_{p}((nh_{x})^{-1/2})$. They also enter
linearly, and the linear terms produce the limit distribution of Theorem
\ref{thm:pointwise}. Write $a_{n}\asymp b_{n}$ when $a_{n}/b_{n}$
is bounded above and away from zero. With $h_{y}\asymp h_{x}\asymp h_{z}\asymp n^{-a}$
these conditions, together with the undersmoothing that Theorems \ref{thm:pointwise}
and \ref{thm:band} add, leave $a\in(1/(2m+1),1/7)$, which is nonempty
once $m\ge4$.\footnote{With $h_{y}\asymp h_{x}\asymp h_{z}\asymp n^{-a}$ the last three
conditions are $n^{-1/2+5a/2}$, $n^{-1/2+7a/2}$ and $n^{-1/2+7a/2}$
up to the logarithm, so the last two bind and give $a<1/7$; $\sqrt{nh_{x}}h^{m}_{x}\rightarrow0$
gives $a>1/(2m+1)$.} \footnote{The three exponents need not agree, and the conditions bind on them
differently: $h_{x}$ enters the second with a high power because
two of the smoothed quantities, $\nabla_{x}\widehat{\Phi}$ and $\nabla_{x}\widehat{f}$,
are derivatives in $x$, while $h_{y}$ and $h_{z}$ are constrained
mainly through the bias. With $a_{x}=0.15$, for instance, the conditions
hold for $a_{y},a_{z}\in(0.106,0.125)$ at $m=4$. The lower endpoint
is the exponent at which the estimator attains the optimal rate, and
it is excluded: inference requires undersmoothing, and with it a rate
arbitrarily close to the optimum rather than equal to it. Section
\ref{sec:Monte-carlo} uses a fourth-order kernel.} 

The following gives the limit distribution at a point for the nonparametric
estimators $\hat{g}_{\mathrm{LS}}$ and $\hat{g}_{\mathrm{LAD}}$.
\begin{thm}
\label{thm:pointwise} Suppose that Assumptions \ref{assu:rv}--\ref{assu:support}
and \ref{assu: kernel}--\ref{assu:rate} hold, that $\sigma^{2}(0)>0$,
and that the bandwidth undersmooths in the sense that $\sqrt{nh_{x}}h^{m}_{x}\rightarrow0$.
Then for each $x\in\mathcal{X}_{0}$
\[
\sqrt{nh_{x}}\left(\widehat{g}_{\mathrm{LS}}(x)-g(x)\right)\rightsquigarrow N\left(0,\sigma^{2}(x)+\sigma^{2}(0)\right),
\]
where $\sigma^{2}(a)$ is defined below. The same limit holds for
$\widehat{g}_{\mathrm{LAD}}(x)$.
\end{thm}

\begin{proof}
See Appendix \ref{sec:thm1}.
\end{proof}

Lemma \ref{lem:boundary} decomposes the deviation as
\[
\widehat{g}_{\mathrm{LS}}(x)-g(x)=T_{n}(x)-T_{n}(0)+o_{p}((nh_{x})^{-1/2}),
\]
where
\begin{align*}
T_{n}(a) & =\frac{1}{nh_{x}}\sum^{n}_{i=1}K\left(\frac{X_{i}-a}{h_{x}}\right)c_{a}(W_{i}),\\
c_{a}(W_{i}) & =\int\frac{\mu(y,Z_{i})\left[F_{Y|X,Z}(y|a,Z_{i})-\boldsymbol{1}\left\{ Y_{i}\le y\right\} \right]}{f_{Y|X,Z}(y|a,Z_{i})f(a,Z_{i})}\mathrm{d}y,\\
\sigma^{2}(a) & =\left[\int K(v)^{2}\mathrm{d}v\right]f_{X}(a)\mathbb{E}\left[c_{a}(W_{i})^{2}|X_{i}=a\right].
\end{align*}
Each $T_{n}(a)$ is a kernel average in the single argument $x$,
centered at $a$, and $\sigma^{2}(a)$ is the limit of $nh_{x}$ times
its variance. The term at $a=0$ is present because the estimator
integrates from the normalization point. Assumption \ref{assu: func_g}(ii)
sets $g(0)=0$, $\widehat{S}$ integrates from $0$ to $x$, and the
error made at the lower endpoint is the same at every $x$. For $x$
bounded away from the origin the two averages are asymptotically uncorrelated,
so their variances add. Section \ref{sec:Monte-carlo} measures the
component that the evaluation points share.

We need to exclude the normalization point from the evaluation set,
though not because anything degenerates there. Assumption \ref{assu: func_g}
(ii) sets $g(0)=0$, so $\hat{g}(0)=0$ by construction and the deviation
at that point is identically zero. Away from it, the two terms Lemma
\ref{lem:boundary} leaves are kernel averages centered at $x$ and
at $0$, and their covariance is $O\big((h_{x}/\left|x\right|)^{\nu}\big)$
after rescaling, which makes their variances add to $\sigma^{2}(x)+\sigma^{2}(0)$.
For that to hold uniformly, $\mathcal{X}_{0}$ has to be bounded away
from the origin. Theorem \ref{thm:pointwise} and the band of Theorem
\ref{thm:band} are therefore available on any compact set omitting
a neighborhood of the normalization point, but not on an interval
straddling it; in Section \ref{sec:Monte-carlo}, $\mathcal{X}_{0}$
is a union of two intervals bounded away from the origin.

The rate is $(nh_{x})^{-1/2}$ rather than $n^{-1/2}$ because $\widehat{S}$
differentiates in the direction it integrates over. A kernel in $x$
that stays inside $\int^{x}_{0}\mathrm{d}u$ is averaged by it and
contributes at $n^{-1/2}$; but $\nabla_{x}\widehat{V}$ and $\nabla_{x}\widehat{F}$
are derivatives in $u$, and integrating them by parts returns the
kernel at $u=0$ and $u=x$, where nothing averages it. $T_{n}(x)$
and $T_{n}(0)$ are those two terms, and a kernel at a point converges
at $(nh_{x})^{-1/2}$. Only $h_{x}$ appears in the limit, although
$\widehat{F}$ smooths in three arguments and $\widehat{V}$ in two.
The estimator integrates $\widehat{S}$ against $\mu$ in $y$ and
$z$, so those two directions are averaged rather than evaluated at
a point, and $h_{y}$ and $h_{z}$ enter through the bias and the
remainder alone. The leading term is therefore a kernel average in
$x$ alone, and the rate is the one-dimensional rate.

The rate $n^{-m/(2m+1)}$ and the limit distribution of Theorem \ref{thm:pointwise}
are attained at different bandwidths. At $h_{x}\asymp n^{-1/(2m+1)}$
the two terms in the bound $|\widehat{g}_{\mathrm{LS}}(x)-g(x)|=O_{p}(h^{m}_{x}+(nh_{x})^{-1/2})$
of Lemma \ref{lem:boundary} are of the same order, and the bound
is $O_{p}(n^{-m/(2m+1)})$. At that bandwidth the estimator attains
the rate Proposition \ref{prop:lower_bound} shows below cannot be
improved. Note that Theorem \ref{thm:pointwise} holds at smaller
bandwidths, those satisfying $\sqrt{nh_{x}}h^{m}_{x}\rightarrow0$.
That condition removes the bias from the limit distribution, and it
excludes $h_{x}\asymp n^{-1/(2m+1)}$. Appendix \ref{sec:thm1} is
written under the conditions of Theorem \ref{thm:pointwise}, so Lemma
\ref{lem:boundary} carries that condition too. The bound does not
use it. Parts (ii) to (iv) retain the bias terms rather than discarding
them, and the undersmoothing condition enters their proofs only in
remarks that a retained term is negligible, so they hold under Assumptions
\ref{assu:rv}--\ref{assu:support} and \ref{assu: kernel}--\ref{assu:rate}
alone. Thus, the bandwidths it allows give a rate slower than $n^{-m/(2m+1)}$.
Taking $a_{x}$ arbitrarily close to $1/(2m+1)$ makes the difference
as small as desired.

The next result bounds what any estimator of $g(x)$ can achieve,
using a result of \citet{Sto80-AoS}.
\begin{prop}
\label{prop:lower_bound} Let $\mathcal{P}$ be the set of distributions
of $\left(Y,X,Z\right)$ that (\ref{eq:model}) generates under Assumptions
\ref{assu:rv}--\ref{assu:support}, with the joint density bounded
and $m$ times differentiable with bounded derivatives, that bound
being held fixed across the set. Fix $x\in\mathcal{X}^{\circ}$ with
$x\neq0$. Then there is a $c>0$ such that
\[
\liminf_{n\rightarrow\infty}\inf_{\tilde{g}_{n}}\sup_{P\in\mathcal{P}}P^{n}\left(\left|\tilde{g}_{n}(x)-g_{P}(x)\right|\ge cn^{-m/(2m+1)}\right)>0,
\]
where the infimum runs over all estimators of $g(x)$ based on a sample
of size $n$. 
\end{prop}

\begin{proof}
See Appendix \ref{sec:lower}.
\end{proof}

Imposing that the endogeneity is absent makes the problem easier,
so the best rate available under that restriction limits what any
estimator can achieve without it. \citet[p.~576]{NPV99-ECTA} argue
in the same way for their own estimator, comparing instead with the
problem in which the control function is known rather than estimated.
The proof restricts to the subset $\mathcal{P}_{0}$ of distributions
on which $\varepsilon$ is independent of $\left(\eta,Z\right)$.
The step that is not immediate is that $\mathcal{P}_{0}$ is still
inside $\mathcal{P}$. The restriction makes $F_{Y|X,Z}(y|x,z)$ free
of $z$, so $\nabla_{z}F_{Y|X,Z}(y|x,z)$ vanishes. The proof checks
Assumption \ref{assu: local} to confirm that a model with that property
still satisfies the assumptions that define $\mathcal{P}$. The proof
fixes the first stage at one admissible choice, but uses no property
of that choice beyond what Assumptions \ref{assu:rv}--\ref{assu:support}
require. Letting $g$ vary leaves the first stage untouched, and with
it $V(x|z)$, so every member of $\mathcal{P}_{0}$ has the first
stage that was fixed. Write $\mathcal{P}(\lambda)$ for the distributions
in $\mathcal{P}$ whose first stage satisfies $\inf\left|\nabla_{z}V(x|z)\right|\ge\lambda$.
Fixing a first stage that meets the bound puts $\mathcal{P}_{0}$
inside $\mathcal{P}(\lambda)$, and the argument above gives the same
lower bound there. The rate is therefore $n^{-m/(2m+1)}$ over every
nonempty $\mathcal{P}(\lambda)$, so a stronger instrument does not
raise the rate at which $g(x)$ can be estimated. Nothing in the construction
is particular to this way of indexing the first stage, and the same
holds for any restriction imposed on the first stage alone.

\subsection{Bootstrap inference}

An interval that covers $g(x)$ at each $x$ separately need not cover
at every $x$ at once. Only simultaneous coverage makes a band informative
about the shape of $g$: whether it is monotone, whether it is linear,
whether it departs from a candidate specification anywhere on $\mathcal{X}_{0}$.
We study a sup-$t$ band, with the critical value taken from the empirical
bootstrap.

Let $\left\{ W^{*}_{i}\right\} ^{n}_{i=1}$ be drawn i.i.d.\  with
replacement from the original sample $\left\{ W_{i}\right\} ^{n}_{i=1}$,
and let $\hat{F}^{*}$ and $\hat{V}^{*}$ be the kernel estimators
constructed from $\left\{ W^{*}_{i}\right\} ^{n}_{i=1}$ with the
same bandwidths. Let $\widehat{S}^{*}$ be obtained by substituting
$\hat{F}^{*}$ and $\hat{V}^{*}$ for $\hat{F}$ and $\hat{V}$ in
the definition of $\widehat{S}$ above, and let $\hat{g}^{*}(x)$
be $\hat{g}_{\mathrm{LS}}(x)$ or $\hat{g}_{\mathrm{LAD}}(x)$ computed
from $\widehat{S}^{*}$ in place of $\widehat{S}$.

Write $\bar{\sigma}(x):=\left\{ \sigma^{2}(x)+\sigma^{2}(0)\right\} ^{1/2}$
for the standard deviation in Theorem \ref{thm:pointwise}. Let $\hat{s}_{n}(x)>0$
estimate $\bar{\sigma}(x)$ from the sample and $\hat{s}^{*}_{n}(x)>0$
estimate it from the bootstrap sample. They need not differ: Section
\ref{sec:Monte-carlo} uses the same estimator for both. For $\alpha\in(0,1)$
define the conditional quantile of the studentized supremum, 
\begin{equation}
c_{n}(1-\alpha):=\inf\left\{ c\ge0:P^{*}\left(\sup_{x\in\mathcal{X}_{0}}\frac{\sqrt{nh_{x}}\left|\hat{g}^{*}(x)-\hat{g}(x)\right|}{\hat{s}^{*}_{n}(x)}\le c\right)\ge1-\alpha\right\} ,\label{eq:supt_crit}
\end{equation}
 and the band 
\begin{equation}
\mathcal{C}_{n}(x):=\left[\hat{g}(x)-c_{n}(1-\alpha)\frac{\hat{s}_{n}(x)}{\sqrt{nh_{x}}},\;\hat{g}(x)+c_{n}(1-\alpha)\frac{\hat{s}_{n}(x)}{\sqrt{nh_{x}}}\right].\label{eq:supt_band}
\end{equation}

Since $\hat{s}_{n}$ is positive, $g(x)\in\mathcal{C}_{n}(x)$ for
every $x\in\mathcal{X}_{0}$ exactly when 
\[
\sup_{x\in\mathcal{X}_{0}}\frac{\sqrt{nh_{x}}\left|\hat{g}(x)-g(x)\right|}{\hat{s}_{n}(x)}\le c_{n}(1-\alpha).
\]
 The studentizer that sets the half-width is therefore the one that
appears in the coverage event, while the one used in (\ref{eq:supt_crit})
enters only through the critical value. The two can be chosen separately,
and Assumption \ref{assu: band} asks the same thing of each. Studentizing
at all lets the width track the local precision of $\hat{g}$, which
in this model varies by a factor of six across $\mathcal{X}_{0}$,
since $\hat{g}(x)$ is $x$ times the average of $\widehat{\nabla_{u}g}$
over $[0,x]$.

Two choices of the pair are natural. The first takes both from the
bootstrap, 
\begin{equation}
\hat{s}_{n}(x)=\hat{s}^{*}_{n}(x)=\hat{\sigma}(x),\qquad\hat{\sigma}^{2}(x):=nh_{x}\mathbb{E}^{*}\left[\left(\hat{g}^{*}(x)-\mathbb{E}^{*}\hat{g}^{*}(x)\right)^{2}\right],\label{eq:boot_sd}
\end{equation}
 where $\mathbb{E}^{*}$ is expectation with respect to the resampling
distribution conditional on the data. Section \ref{sec:Monte-carlo}
implements this. Since $\hat{\sigma}$ is a function of the sample
alone it is the same in every resample, and multiplying it by a constant
leaves the band unchanged, the critical value absorbing the factor;
what the band responds to is the shape of $\hat{\sigma}$ across $\mathcal{X}_{0}$
and not its level.

The second studentizes each resample separately: form 
\begin{equation}
\hat{\sigma}^{2}_{\mathrm{pl}}(a):=\frac{1}{n}\sum^{n}_{i=1}\frac{1}{h_{x}}K\left(\frac{X_{i}-a}{h_{x}}\right)^{2}\hat{c}_{a}(W_{i})^{2},\label{eq:plugin}
\end{equation}
 with $\hat{c}_{a}$ formed from $\widehat{F}$, $\nabla_{y}\widehat{F}$
and $\hat{f}$. This estimates $\sigma^{2}(a)$: it reuses $h_{x}$
and replaces $K$ by $K^{2}$. Assumption \ref{assu: band} constrains
the studentizer against $\bar{\sigma}(x)$, which combines $\sigma^{2}(x)$
and $\sigma^{2}(0)$, so the plug-in is taken at both points and summed,
\[
\hat{s}_{n}(x)^{2}:=\hat{\sigma}^{2}_{\mathrm{pl}}(x)+\hat{\sigma}^{2}_{\mathrm{pl}}(0).
\]
Evaluating at $0$ is allowed although $0$ is excluded from $\mathcal{X}_{0}$:
it lies in $\widetilde{\mathcal{X}}_{0}$, where Assumption \ref{assu: density}(iii)
keeps the estimated nuisance functions away from zero.

Taking $\hat{s}_{n}$ from the sample and $\hat{s}^{*}_{n}$ from
each resample makes the band a percentile-$t$ band. The denominator
of the bootstrap statistic then comes from the same resample as its
numerator, so a resample that is dispersed inflates both. The plug-in
costs one further pass over each resample. Studentizing each resample
by a bootstrap standard deviation instead would require a bootstrap
inside each resample.

Drawing the critical value from a bootstrap rather than from an extreme-value
limit is standard in nonparametric problems \citep[see, e.g.,][]{CCK14-AoS}.
In NPIV models, \citet{HL12-JoE} and \citet{CC18-QE} provide valid
uniform confidence bands for ill-posed inverse problems. The current
setup avoids ill-posedness by using a (nonseparable) control function.
\citet{HL12-JoE} employ the empirical bootstrap as we do and studentize
with an explicit estimator of the influence-function variance; (\ref{eq:plugin})
is the counterpart here. Their band rests on finitely many points:
they form joint intervals there and interpolate between them, which
requires $g$ to be piecewise monotone or Lipschitz and leaves a term
that vanishes only as the grid gets finer. Theorem \ref{thm:band}
is stated on $\mathcal{X}_{0}$ directly, so no interpolation step
is needed. In practice $\hat{s}_{n}$ and $c_{n}$ are computed from
$B$ bootstrap draws on a finite grid; see Section \ref{sec:Monte-carlo}.
\begin{assumption}
\label{assu: band} (i) $\sup_{x\in\mathcal{X}_{0}}\left|\hat{s}_{n}(x)/\bar{\sigma}(x)-1\right|=o_{p}(1/\log n)$;
and (ii) $\sup_{x\in\mathcal{X}_{0}}\left|\hat{s}^{*}_{n}(x)/\bar{\sigma}(x)-1\right|=o_{p}(1/\log n)$
conditionally on the data, in probability.
\end{assumption}

Assumption \ref{assu: band} needs $\inf_{x\in\mathcal{X}_{0}}\bar{\sigma}(x)$,
which is ensured because Theorem \ref{thm:pointwise} already requires
$\sigma^{2}(0)>0$, and since $\sigma^{2}(x)\ge0$ that gives $\inf_{x\in\mathcal{X}_{0}}\bar{\sigma}(x)\ge\sigma(0)>0$.
Parts (i) and (ii) are the price of studentizing, and they ask for
a rate rather than for consistency alone. The reason is that the class
over which the supremum is taken is not Donsker: its envelope is of
order $h^{-1/2}_{x}$ and dominates its members (Lemma \ref{lem:vc}),
so the supremum grows like $\sqrt{\log n}$ rather than being bounded
in probability. Write $\varepsilon$ for the supremum in parts (i)
and (ii), the largest relative error the studentizer makes on $\mathcal{X}_{0}$.
Dividing by the studentizer rather than by $\bar{\sigma}$ shifts
the statistic by $\varepsilon\sqrt{\log n}$. The anti-concentration
inequality of \citet{CCK14-AoS} turns that shift into a change in
probability, again with a factor of order $\sqrt{\log n}$. The two
factors compound, so $\varepsilon$ contributes a term of order $\varepsilon\log n$
to the coverage error. Parts (i) and (ii) require $\varepsilon=o_{p}(1/\log n)$,
which makes that term vanish. With $\varepsilon=o_{p}(1)$ alone,
it need not.

The plug-in (\ref{eq:plugin}) and the bootstrap standard deviation
(\ref{eq:boot_sd}) differ in what parts (i) and (ii) then require.
For the plug-in they can be verified, which is Lemma \ref{lem:plugin}.
For $\hat{\sigma}$ taken from the bootstrap, (i) asks for the convergence
of a conditional second moment, and a distributional approximation
does not deliver one on its own. It would follow from uniform square-integrability
of $\sqrt{nh_{x}}(\hat{g}^{*}-\hat{g})$ with a rate attached, which
is not established here. Assumption \ref{assu: band} is imposed on
the pair itself rather than on either construction, so the theorem
below holds for whichever of the two is used.
\begin{thm}
\label{thm:band} Suppose the conditions of Theorem \ref{thm:pointwise}
and Assumption \ref{assu: band} hold. Then the band (\ref{eq:supt_band})
covers $g$ at every point of $\mathcal{X}_{0}$ simultaneously with
asymptotic probability $1-\alpha$, 
\[
P\left(g(x)\in\mathcal{C}_{n}(x)\;\text{for all }x\in\mathcal{X}_{0}\right)\rightarrow1-\alpha.
\]
 The difference between the two sides is at most a constant multiple
of 
\[
(h_{x}\log n)^{1/2}+\sqrt{nh_{x}}R_{n}(\log n)^{1/2}+\sqrt{nh_{x}}\bar{h}^{m}(\log n)^{1/2}+(nh_{x})^{-1/8+\upsilon}+\left(\Delta_{n}+\Delta^{*}_{n}\right)\log n,
\]
 for every $\upsilon\in(0,1/8)$, the constant depending on $\upsilon$,
where $\Delta_{n}$ and $\Delta^{*}_{n}$ are deterministic sequences
dominating the two suprema in Assumption \ref{assu: band} (i) and
(ii) with probability tending to one. The sequence $R_{n}$ is the
square of the largest uniform rate the smoothed quantities attain,
given at (\ref{eq:Rn}).
\end{thm}

\begin{proof}
See Appendix \ref{sec:proof_band}.
\end{proof}

The rescaled process $\sqrt{nh_{x}}(\hat{g}(x)-g(x))$, indexed by
$x\in\mathcal{X}_{0}$, is not tight, because the class the supremum
runs over changes with $n$. The proof therefore cannot draw the critical
value from a tight limit process, or from an extreme-value law for
its supremum, as a band on a continuum usually does. It relies instead
on two results, which do different work. \citet{CCK16-SPA} supplies
couplings: on the sample side and on the bootstrap side, each supremum
is brought close to the supremum of one Gaussian process, indexed
by the same class and at the same $n$, and the two are compared through
that common intermediary. A coupling bounds the difference between
two suprema, and does not by itself bound the difference between their
distribution functions. The anti-concentration inequality of \citet{CCK14-AoS}
supplies that step, because it limits the mass the supremum of a Gaussian
process can place in a short interval, and the same inequality sets
the rate asked of the studentizers in Assumption \ref{assu: band}.
The coverage error is bounded at each $n$ rather than in the limit,
which is why the second display carries a rate.
\begin{rem}
A band is honest over a class of distributions if its coverage error
tends to zero uniformly over the class \citep[p.~1788]{CCK14-AoS}.
Theorem \ref{thm:band} bounds the coverage error at a single distribution,
and the constant in that bound depends on it, so it does not establish
honesty. The obstacle is not the construction of the band. Their Corollary
3.1 treats a band built by undersmoothing rather than by an explicit
bias correction, and when the smoothing parameter is deterministic
its coverage is asymptotically exact uniformly over the class. The
band here is of that construction, with deterministic bandwidths by
Assumption \ref{assu:rate}. What is missing is the class: the conditions
the theorem uses are imposed at one distribution, and honesty asks
them to hold uniformly over a set of them. Three of them would have
to change. First, Assumptions \ref{assu: kernel} and \ref{assu: density}
bound densities and derivatives without naming the bounds, so there
is no class to quantify over. Proposition \ref{prop:lower_bound}
defines one, $\mathcal{P}$, by holding those bounds fixed across
the set of distributions. Second, the condition $\sigma^{2}(0)>0$
is imposed at a single distribution, and would become an infimum over
the class. Third, parts (i) and (ii) of Assumption \ref{assu: band}
would have to hold uniformly over the class and at a polynomial rate,
which is what their Condition H4 asks of the scale estimator, in place
of the $o_{p}(1/\log n)$ in probability asked here. $\blacksquare$
\end{rem}

\section{Monte Carlo }\label{sec:Monte-carlo}

We study the finite-sample behavior of the estimators and of the band
in the design 
\begin{align*}
Y_{i} & =\sin(X_{i})+\varepsilon_{i},\\
X_{i} & =6L(-Z_{i}+\eta_{i})-3,
\end{align*}
 where $L(a)=(1+\exp(-a))^{-1}$. The support of $X_{i}$ is $(-3,3)$
and $g(0)=0$, as Assumption \ref{assu: func_g} requires. The marginal
distributions of $\varepsilon_{i}$ and $\eta_{i}$ are both $N(0,1)$
and their joint distribution is a Frank copula 
\[
F_{\varepsilon,\eta}(e,n)=-\frac{1}{\lambda}\log\left(1+\frac{(\exp(-\lambda\Phi(e))-1)(\exp(-\lambda\Phi(n))-1)}{\exp(-\lambda)-1}\right),
\]
 with $\lambda=2$, following \citet{torgovitsky2017minimum}. The
instrument $Z_{i}$ is drawn independently from $N(0,1)$, so that
$(\varepsilon,\eta)$ is independent of $Z$ as Assumption \ref{assu:IV}
requires. Each experiment uses 100 Monte Carlo replications.

The first stage is nonseparable in $(Z,\eta)$, which is the case
this paper is about. It is also a case in which the control function
used by \citet{NPV99-ECTA} is not valid, so their estimator gives
a natural comparison.

The support of the weight function is chosen in a data-driven way
to avoid the denominator problem: 
\[
\mathcal{Y}_{\mu}\times\mathcal{Z}_{\mu}:=\left\{ (y,z):\begin{array}{c}
\nabla_{y}\widehat{F}(y|\overline{X},z)>\tau\\
\nabla_{y}\widehat{F}(y|\overline{X},z)\nabla_{z}\widehat{V}(\overline{X}|z)>\tau
\end{array}\text{ and }\begin{array}{c}
\hat{q}_{Y}(2.5)\le y\le\hat{q}_{Y}(97.5)\\
\hat{q}_{Z}(2.5)\le z\le\hat{q}_{Z}(97.5)
\end{array}\right\} ,
\]
 where $\overline{X}$ is the sample mean of $X$ and $\hat{q}_{R}(\cdot)$
is the sample quantile function for $R\in\left\{ Y,Z\right\} $. Both
bounds are on denominators of the integrand in (\ref{eq:g_functional}),
which divides by $f_{Y|X,Z}$ and by $\nabla_{z}Vf_{Y|X,Z}$. We take
$\tau=10^{-4}$ in every experiment. The estimators are evaluated
as follows: (i) draw $\{(y_{j},z_{j})\}^{M}_{j=1}$ from $\mathcal{Y}_{\mu}\times\mathcal{Z}_{\mu}$
uniformly at random, (ii) for each $j$, compute 
\[
\widehat{S}^{(j)}(y_{j},x,z_{j})=\int^{x}_{0}\frac{\frac{\nabla_{x}\widehat{V}(u|z_{j})}{\nabla_{z}\widehat{V}(u|z_{j})}\nabla_{z}\widehat{F}(y_{j}|u,z_{j})-\nabla_{x}\widehat{F}(y_{j}|u,z_{j})}{\nabla_{y}\widehat{F}(y_{j}|u,z_{j})}\mathrm{d}u,
\]
 where the integral over $u$ is evaluated on a grid of $M_{u}$ midpoints,
and (iii) compute $\hat{g}_{\mathrm{LS}}(x)=M^{-1}\sum^{M}_{j=1}\widehat{S}^{(j)}(y_{j},x,z_{j})$
or $\hat{g}_{\mathrm{LAD}}(x)=\arg\min_{q}Q_{b}(q|\{\hat{S}^{(j)}\}^{M}_{j=1})$
with 
\[
Q_{b}(q|\{\hat{S}^{(j)}\}^{M}_{j=1})=M^{-1}\sum^{M}_{j=1}\left\{ \hat{S}^{(j)}(y_{j},x,z_{j})-q\right\} \left\{ 2\Lambda_{b}\left[\hat{S}^{(j)}(y_{j},x,z_{j})-q\right]-1\right\} ,
\]
taking $b=0.01$ and $\Lambda$ the standard normal distribution function.
We use $M=300$ and $M_{u}=100$ in the point-estimation experiments
and $M=200$, $M_{u}=50$ in the bootstrap experiment, which costs
$B=500$ times as much per replication. The support of $\mu$ and
the draws $(y_{j},z_{j})$ are selected once from the original sample
and held fixed across bootstrap replications, so that $\hat{g}^{*}-\hat{g}$
does not carry integration noise that $\hat{g}-g$ does not.\footnote{Two features of the implementation lie outside the theory as stated;
a third, the bandwidth sequence, is taken up below. The first is the
data-driven support. Assumption \ref{assu: density} asks that $\mu$
be fixed and $m$-times continuously differentiable, which an indicator
of estimated derivatives crossing a threshold is not, so the theorems
do not cover the rule used here. Appendix \ref{sec:thm1} identifies
what the smoothness provides: it makes the boundary terms of the four
integrations by parts in $z$ vanish. With a weight that is not continuous
those terms survive at order $(nh_{z})^{-1/2}$, the order of the
leading term itself, so this departure is not one that merely moves
a constant. Remark \ref{rmk:mu} gives that the estimand is the same
for every $\mu$, so a data-driven choice cannot move the target;
it does not follow that plugging in a random $\hat{\mu}$ is first-order
negligible. The second feature is the finite $M$. There is no approximation
bias: $S(y,x,z)=g(x)$ at every $(y,z)$ in the support, so $M^{-1}\sum^{M}_{j=1}S(y_{j},x,z_{j})$
is exactly $g(x)$ whatever the draws. The first-order variance is
another matter: the average over $M$ draws replaces the integral
against $\mu$ by one against their empirical measure, changing the
function $c_{a}$ and with it the variance of Theorem \ref{thm:pointwise}.
The simulations condition on the draws.}

The estimators are evaluated at 21 equally spaced points on $[-2.5,2.5]$.
We do not evaluate at $\pm3$. At the edges of the support of $X$
one has $V(x|z)\to1$ for every $z$, so $\nabla_{z}V\to0$ and the
ratio that identifies $\nabla_{x}g$ has a vanishing denominator;
Assumption \ref{assu: density} and Theorem \ref{thm:pointwise} are
stated on a compact $\mathcal{X}_{0}$ strictly inside $\mathcal{X}$
for this reason. The point $x=0$ is the normalization point, where
both estimators are zero by construction; the band excludes it, for
the reason given after Theorem \ref{thm:pointwise}, so the band is
reported on the remaining 20 points.

We set $h_{v}(n)=c\,\hat{\sigma}_{v}\,(n/1000)^{-0.15}$ for $v\in\{y,x,z\}$,
where $\hat{\sigma}^{2}_{v}$ is the sample variance. The exponent
0.15 that the simulations use for all three bandwidths departs deliberately
from what Assumption \ref{assu:rate} allows, and it is worth saying
at once in what respect. With a kernel of order $m$ the common exponent
must lie in $(1/(2m+1),1/7)$, which is $(0.111,0.143)$ at $m=4$,
and that window belongs to the case of three equal exponents. The
value 0.15 meets every condition in which $h_{x}$ appears, including
the undersmoothing one. It fails the conditions on the first-stage
bandwidths, which at $a_{x}=0.15$ ask for $a_{y},a_{z}\in(0.106,0.125)$
rather than 0.15. Using 0.15 for all three is therefore a further
respect in which the implementation departs from the theory. The first-stage
bandwidths it gives differ by a few per cent from those an exponent
inside the window would give, and at $n=1000$ not at all, since the
factor $(n/1000)^{-0.15}$ equals one there. Assumption \ref{assu:rate}
restricts the exponents and leaves the constants free, so $c$ can
be chosen separately for point estimation and for the band. We choose
it by comparing values of $c$ on a grid at $n=1000$. Table \ref{tab:bw}
reports that grid for the LAD estimator. Two criteria pick different
constants. Root mean squared error is smallest at $c=0.60$. The ratio
$\max_{x}|\text{bias}|/\text{sd}$, which governs how far a confidence
interval centered at $\hat{g}$ is displaced relative to its own width,
is smallest at $c=0.40$. We use $c=0.60$ for point estimation and
$c=0.40$ for the band.

\begin{table}[t]
\centering
\begin{tabular}{lcccccccc}
\hline
$c$ & $0.20$ & $0.30$ & $0.40$ & $0.50$ & $0.60$ & $0.75$ & $0.90$ & $1.10$\\
\hline
RMSE & $0.336$ & $0.226$ & $0.175$ & $0.148$ & $\mathbf{0.135}$ & $0.146$ & $0.190$ & $0.248$\\
$\max_x|\text{bias}|/\text{sd}$ & $0.869$ & $0.544$ & $\mathbf{0.461}$ & $0.578$ & $0.869$ & $1.484$ & $3.143$ & $6.025$\\
\hline
\end{tabular}
\caption{Bandwidth selection for the LAD estimator, $n=1000$, 21 evaluation points, 100 replications. Smallest value in each row in bold.}
\label{tab:bw}
\end{table}

The bias is not monotone in $c$. Below $c=0.40$ the trimming in
the display above becomes active: the surviving $(y,z)$ set falls
from about 9,200 of the 10,201 grid points at $c=0.40$ to 6,400 at
$c=0.20$, and the ratio rises again. The choice of the bandwidth
shrinks $\mathcal{Y}_{\mu}\times\mathcal{Z}_{\mu}$ since a smaller
one makes the estimated denominators smaller and noisier.

The kernel is the fourth-order Gaussian $K(u)=\tfrac{1}{2}(3-u^{2})\varphi(u)$,
where $\varphi$ is the standard normal density. A fourth-order kernel
can take negative values, so the estimated densities can go non-positive
and the ratio defining $\nabla_{x}g$ can blow up. In this simulation,
the estimated densities stayed positive in all 100 replications at
every bandwidth in Table \ref{tab:bw}.

\subsection{Point estimation}

Table \ref{tab:point} reports bias, standard deviation and root mean
squared error, averaged over the 21 evaluation points, for the two
estimators of this paper and for the series estimator of \citet{NPV99-ECTA}
at series lengths $p=2,\dots,6$. The series estimator uses a power
series basis with the same $p$ at both stages: $\mathbb{E}[X|Z]$
is fitted by a polynomial of degree $p$ in $Z$, and $Y$ is then
fitted by a polynomial of total degree $p$ in $(X,\hat{\eta})$,
with the estimate of $g$ taken from the terms in $X$ alone. Both
of our estimators use $c=0.60$. For reference, the range of $g$
over the grid is $1.995$.

\begin{table}[t]
\centering
\begin{tabular}{lcccccc}
\hline
 & \multicolumn{3}{c}{$n=500$} & \multicolumn{3}{c}{$n=1000$}\\
 & $|\text{bias}|$ & sd & RMSE & $|\text{bias}|$ & sd & RMSE\\
\hline
LAD & $0.047$ & $0.152$ & $0.179$ & $0.048$ & $0.104$ & $\mathbf{0.133}$\\
LS & $0.075$ & $0.241$ & $0.294$ & $0.070$ & $0.185$ & $0.237$\\
\hline
NPV, $p=2$ & $0.309$ & $0.086$ & $0.352$ & $0.310$ & $0.068$ & $0.348$\\
NPV, $p=3$ & $0.056$ & $0.126$ & $\mathbf{0.162}$ & $0.065$ & $0.093$ & $\mathbf{0.133}$\\
NPV, $p=4$ & $0.052$ & $0.171$ & $0.214$ & $0.065$ & $0.123$ & $0.162$\\
NPV, $p=5$ & $0.073$ & $0.207$ & $0.265$ & $0.065$ & $0.143$ & $0.180$\\
NPV, $p=6$ & $0.077$ & $0.242$ & $0.315$ & $0.067$ & $0.169$ & $0.213$\\
\hline
\end{tabular}
\caption{Bias, standard deviation and RMSE, averaged over 21 points on $[-2.5,2.5]$, 100 replications. Smallest RMSE in each column in bold.}
\label{tab:point}
\end{table}

Three things are visible. First, LAD dominates LS by a wide margin
at the same bandwidth, by a factor of 1.6 in RMSE at $n=500$ and
1.8 at $n=1000$. The gap is variance, not bias. It narrows if LS
is given its own bandwidth --- at $n=1000$ the RMSE-minimizing value
for LS is $c=0.75$, giving $0.245$ --- but LAD remains ahead. We
report LAD as the main estimator from here on.

Second, the series estimator's RMSE varies by a factor of 2.6 across
$p=2,\dots,6$ at $n=1000$. However, one can use a cross-validation
to choose $p$, and it picks $p=3$ in 89 of 100 samples at $n=500$
and 90 at $n=1000$. No equally natural criterion is available for
the bandwidth. Held-out prediction of $Y$ is the closest analogue,
but it targets $\mathbb{E}[Y\mid X]=g(X)+\mathbb{E}[\varepsilon\mid X]$,
which differs from $g$ precisely because $X$ is endogenous. The
difficulty is not specific to this estimator: \citet[Section 2.2]{CCK25-ReStud}
argue that under endogeneity the cross-validation criterion need not
estimate the mean squared error even asymptotically, so it is not
a meaningful criterion for choosing the sieve dimension, and they
replace it by a data-driven rule. Nothing of the kind is available
for a bandwidth here, and its selection remains an open problem.

Third, we do not claim an advantage in root mean squared error. LAD
beats the series estimator at every $p$ except $p=3$, and loses
to $p=3$ by 11 per cent at $n=500$. At $n=1000$ the two are level:
the difference is $0.0008$.

It should be noted that a misspecified estimator often achieves lower
mean squared error than a consistent one at a given sample size \citep[see, e.g.,][]{AKS25-ECTA}.
Lower error at one sample size is therefore not a reason to use the
series estimator here. What separates the two is whether the error
vanishes as the sample grows. Table \ref{tab:cons} in Appendix \ref{sec:Consistency}
follows both estimators from $n=250$ to $n=2000$. The bias of our
estimators falls with the sample: $n^{-0.19}$ for LS and $n^{-0.28}$
for LAD. The bias of the series estimator does not fall, at $n^{-0.05}$
with $p$ chosen by cross-validation. Its standard deviation falls
fastest of all, at $n^{-0.56}$, and by $n=2000$ its bias is larger
than its standard deviation.

\subsection{The band}

We now turn to the band (\ref{eq:supt_band}), computed for the LAD
estimator at $n=1000$ with $c=0.40$ and $B=500$ bootstrap draws.
The critical value $c_{n}(1-\alpha)$ and the scale $\hat{\sigma}$
are computed on the 20 evaluation points other than $x=0$; $\hat{\sigma}$
is the bootstrap standard deviation at each point, not smoothed across
points. This is the first of the two studentizers of Section \ref{sec:Estimation-and-bootstrap},
that is, $\hat{s}_{n}=\hat{s}^{*}_{n}=\hat{\sigma}$.

\begin{figure}[t]
\begin{centering}
\includegraphics[width=0.8\textwidth]{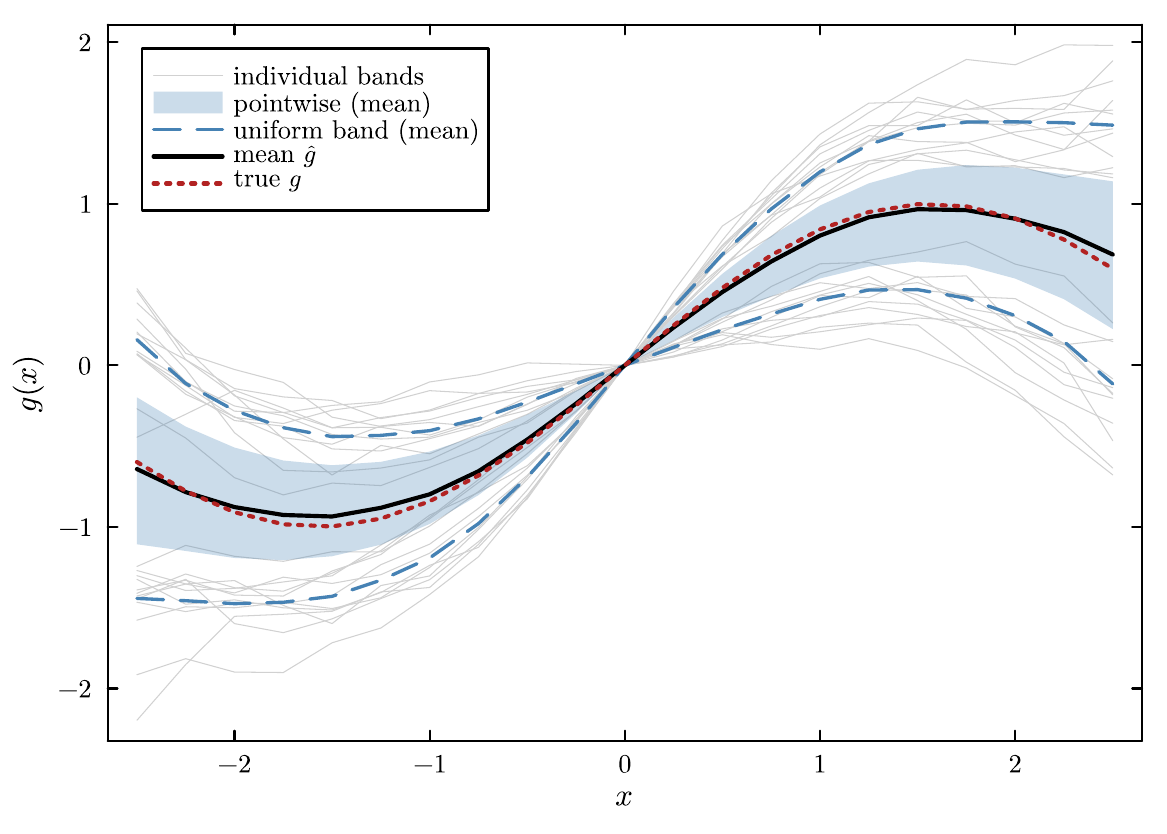}
\par\end{centering}
\caption{Uniform band and pointwise intervals at $\alpha=0.10$, LAD, $n=1000$,
$c=0.40$, $B=500$. Averages over 100 replications, with the bands
from individual replications in grey.}\label{fig:band}
\end{figure}

Figure \ref{fig:band} shows the band averaged over replications,
with the pointwise intervals and the individual replications behind
it. The width varies across $\mathcal{X}_{0}$, because $\hat{g}(x)$
carries the factor $x$: the band is about six times wider at the
outermost evaluation points $\pm2.5$ than at the innermost $\pm0.25$,
and narrows toward the normalization point, which is not itself one
of the evaluation points. That factor multiplies the bias as well,
but the smoothing bias is much smaller, so the mean of $\hat{g}$
still tracks $g$ where the band is widest. The individual replications
drawn behind the average share a common vertical displacement, which
$T_{n}(0)$ of Lemma \ref{lem:boundary} leaves: the deviation at
every $x$ carries the same $-T_{n}(0)$. The shared term makes the
deviations covary: for any two evaluation points, however far apart,
the covariance of $\hat{g}-g$ contains the variance of $T_{n}(0)$.
Without that term the covariance would fall to zero once the kernels
at the two points stop overlapping, at a separation of about two bandwidths.
Averaged over the 97 pairs of points more than 1.5 apart, it is 0.005,
against a variance of 0.029 at each point. About a fifth of the variance
at a point is therefore shared with every other point. The covariance
is positive in 95 of those 97 pairs, so the value is not noise around
zero. The comparison is unaffected by the smoothing bias, which is
common across replications and so drops out of a covariance.

\subsection{Coverage}

Table \ref{tab:cov} reports coverage over 100 replications. With
100 replications a coverage rate has a Monte Carlo standard error
of about $0.03$ at $0.90$. The row marked simultaneous asks how
often the pointwise intervals cover $g$ at all 20 points at once,
which is the quantity the band is designed to control. The two pointwise
rows differ only in centering. Write $q^{*}_{\tau}(x)$ for the $\tau$
quantile of the bootstrap replicates $\hat{g}^{*}(x)$. The row marked
Efron is $[q^{*}_{\alpha/2}(x),q^{*}_{1-\alpha/2}(x)]$. The row marked
percentile is its reflection about $\hat{g}(x)$, $[2\hat{g}(x)-q^{*}_{1-\alpha/2}(x),2\hat{g}(x)-q^{*}_{\alpha/2}(x)]$.
The two have the same width at every point, which is why the width
columns agree.

\begin{table}[t]
\centering
\begin{tabular}{lcccc}
\hline
 & \multicolumn{2}{c}{$\alpha=0.10$} & \multicolumn{2}{c}{$\alpha=0.05$}\\
 & coverage & width & coverage & width\\
\hline
uniform band & $0.950$ & $0.927$ & $0.980$ & $1.018$\\
\hline
pointwise, percentile & $0.808$ & $0.529$ & $0.888$ & $0.632$\\
pointwise, Efron & $0.915$ & $0.529$ & $0.964$ & $0.632$\\
pointwise, Efron, simultaneous & $0.550$ & --- & $0.790$ & ---\\
\hline
\end{tabular}
\caption{Coverage and average width, LAD, $n=1000$, $c=0.40$, $B=500$, 100 replications. Pointwise entries are averaged over the 20 points; the simultaneous row is the fraction of replications in which the pointwise intervals cover at all 20 points.}
\label{tab:cov}
\end{table}

There is no evidence of undercoverage: the band covers in $95$ replications
at $\alpha=0.1$ and 98 replications at $\alpha=0.05$. The pointwise
intervals cover at every point at once in 55 replications. The band
achieves this by being wider, by a factor of 1.75, the ratio of the
two critical values. That width is still useful because the band separates
the peak of the sine from its trough with room to spare. The band
also over-covers, at 0.950 against a nominal 0.90 and 0.980 against
0.95, and Theorem \ref{thm:band} does not predict that. Its error
bound is explicit but not numerically small here: at $n=1000$ with
$c=0.40$, so that $h_{x}\approx0.63$ and $\log n\approx6.9$, the
leading term $(h_{x}\log n)^{1/2}$ is about $2.1$ and the coupling
term about $2.4$, before any constant. As a numerical statement about
the coverage error the bound therefore says nothing at this sample
size; what the theorem gives is that the error vanishes, not that
it is small here.

The bootstrap supremum is larger than the sampling supremum it approximates.
Over 40 replications with 60 resamples, at the same $n$ and bandwidth,
the supremum of $\sqrt{nh_{x}}|\hat{g}^{*}-\hat{g}|$ over the resamples
exceeds the supremum of $\sqrt{nh_{x}}|\hat{g}-g|$ over the replications
by 12 per cent at the $0.95$ quantile (the common $\hat{\sigma}$
divides both suprema and so does not enter the comparison). Thus the
critical value made from the bootstrap distribution becomes larger,
and the band also becomes wider. Recomputing the plug-in standard
error inside each resample, which is the second studentizer of Section
\ref{sec:Estimation-and-bootstrap}, widens the gap to 23 per cent
instead of closing it.

\appendix

\section{Proof of Proposition \ref{prop:identification} }\label{sec:lemma1}

Take any $(x,y,z)$ satisfying conditions (i)--(iv) of Assumption
\ref{assu: local}, and write $e:=y-g(x)$. Note that 
\begin{align*}
F_{Y|X,Z}(y|x,z) & =P\left(g(X)+\varepsilon\le y|X=x,Z=z\right)\\
 & =P\left(\varepsilon\le y-g(x)|V(X|Z)=V(x|z),Z=z\right)\\
 & =F_{\varepsilon|V(X|Z)}\left(y-g(x)|V(x|z)\right)
\end{align*}
The second equality holds because $V(\cdot|z)$ is one-to-one at $x$:
if $V(x^{\prime}|z)=V(x|z)$ for some $x^{\prime}\neq x$, the interval
between the two would carry no mass of $X$ given $Z=z$, which $x\in\mathcal{X}^{\circ}_{z}$
forbids. So $\sigma(X,Z)=\sigma(V(X|Z),Z)$ up to null sets. For the
third equality, let $\tau(z,\cdot)$ be the inverse of $h(z,\cdot)$,
which exists by Assumption \ref{assu:IV}(ii). Assumption \ref{assu:IV}(i)
gives $V(x|z)=P\left(h(z,\eta)\le x\right)=F_{\eta}(\tau(z,x))$,
and $\tau(Z,X)=\eta$, so $V(X|Z)=F_{\eta}(\eta)$ almost surely.
The pair $\left(\varepsilon,V(X|Z)\right)$ is therefore a function
of $\left(\varepsilon,\eta\right)$, which Assumption \ref{assu:IV}(i)
makes independent of $Z$, and the conditioning on it may be dropped.
Now hold $x$ fixed. By Assumption \ref{assu: local}(ii) the map
$z\mapsto V(x|z)$ is continuously differentiable near $z$ with a
derivative that does not vanish there, so the inverse function theorem
supplies a continuously differentiable local inverse $\zeta$. The
display above holds at $(y^{\prime},x,z^{\prime})$ for every $y^{\prime}$
and every $z^{\prime}$ with $x\in\mathcal{X}^{\circ}_{z^{\prime}}$,
which Assumption \ref{assu: local}(i) supplies on a neighborhood
of $z$, so $F_{\varepsilon|V(X|Z)}(e|v)=F_{Y|X,Z}\left(e+g(x)|x,\zeta(v)\right)$
for $(e,v)$ in a neighborhood of $\left(y-g(x),V(x|z)\right)$. Assumption
\ref{assu: local}(iii) makes $\nabla_{y}F_{Y|X,Z}$ and $\nabla_{z}F_{Y|X,Z}$
continuous there, so both partial derivatives of $(e,v)\mapsto F_{\varepsilon|V(X|Z)}(e|v)$
are continuous on that neighborhood and the map is continuously differentiable
in the pair. That allows the chain rule below, in which $x$ moves
both arguments.

Taking derivatives of the previous display evaluating at $(y,x,z)$
gives 
\begin{align}
f_{Y|X,Z}(y|x,z) & =f_{\varepsilon|V(X|Z)}\left(y-g(x)|V(x|z)\right),\label{eq:diff1}\\
\nabla_{x}F_{Y|X,Z}(y|x,z) & =-\nabla_{x}g(x)f_{\varepsilon|V(X|Z)}\left(y-g(x)|V(x|z)\right)\label{eq:diff2}\\
 & \quad+\nabla_{x}V(x|z)\nabla_{v}F_{\varepsilon|V(X|Z)}\left(y-g(x)|V(x|z)\right),\nonumber \\
\nabla_{z}F_{Y|X,Z}(y|x,z) & =\nabla_{z}V(x|z)\nabla_{v}F_{\varepsilon|V(X|Z)}\left(y-g(x)|V(x|z)\right).\label{eq:diff3}
\end{align}
Rearranging (\ref{eq:diff3}) gives 
\[
\nabla_{v}F_{\varepsilon|V(X|Z)}\left(y-g(x)|V(x|z)\right)=\nabla_{z}V(x|z)^{-1}\nabla_{z}F_{Y|X,Z}(y|x,z).
\]
Inserting this expression into (\ref{eq:diff2}) yields 
\begin{align*}
\nabla_{x}g(x)f_{\varepsilon|V(X|Z)}\left(y-g(x)|V(x|z)\right) & =\frac{\nabla_{x}V(x|z)}{\nabla_{z}V(x|z)}\nabla_{z}F_{Y|X,Z}(y|x,z)-\nabla_{x}F_{Y|X,Z}(y|x,z).
\end{align*}
Since $f_{\varepsilon|V(X|Z)}\left(y-g(x)|V(x|z)\right)$ is identified
from (\ref{eq:diff1}), we obtain 
\[
\nabla_{x}g(x)=\frac{\frac{\nabla_{x}V(x|z)}{\nabla_{z}V(x|z)}\nabla_{z}F_{Y|X,Z}(y|x,z)-\nabla_{x}F_{Y|X,Z}(y|x,z)}{f_{Y|X,Z}(y|x,z)}.
\]
The identity before the division is the first claim of the proposition.
Condition (iv) makes the division possible and gives the second. By
Assumption \ref{assu: local} such a $(y,z)$ exists at every $x\in\mathcal{X}_{d}$.
Since $\nabla_{x}g$ is continuous from Assumption \ref{assu: func_g}
and $\mathcal{X}_{d}$ is dense in $\mathcal{X}$, its values on $\mathcal{X}_{d}$
determine it on all of $\mathcal{X}$, which is the claim of the remark
that follows it.

\section{Proof of Theorem \ref{thm:pointwise}}\label{sec:thm1}

Throughout this appendix the conditions of Theorem \ref{thm:pointwise}
are in force. Write $R_{n}$ for the square of the largest uniform
rate the smoothed quantities attain,
\begin{equation}
R_{n}:=\bar{h}^{2m}+\frac{\log n}{nh_{y}h_{x}h_{z}}+\frac{\log n}{nh^{3}_{x}h_{z}}+\frac{\log n}{nh_{x}h^{3}_{z}}.\label{eq:Rn}
\end{equation}
The three variance conditions of Assumption \ref{assu:rate} say that
$\sqrt{nh_{x}}\,R_{n}\rightarrow0$, and the bias conditions give
the same for the first summand.

Lemmas \ref{lem:linear1} and \ref{lem:linear2} provide the linear
approximation of $\widehat{S}(w)-S(w)$ in terms of the nonparametric
estimation errors. Lemma \ref{lem:AL} expands $\hat{g}_{\mathrm{LS}}(x)-g(x)$
in the nuisance errors and splits it into a sample average and two
boundary terms, and Lemma \ref{lem:boundary} evaluates the boundary
terms, which dominate. The theorem then follows from a central limit
theorem for a triangular array at the scale $\sqrt{nh_{x}}$.

Let $W_{i}=(Y_{i},X_{i},Z_{i})^{\prime}$ and $w=(y,x,z)^{\prime}.$
From here on $F$ stands for $F_{Y|X,Z}$, and the arguments of $F(y|x,z)$,
$V(x|z)$, $\Phi(y,x,z)$, $\Psi(x,z)$, $f(x,z)$ and $p(z)$ are
suppressed. A subscript on any of the six denotes a partial derivative
in that argument: $F_{y}=\nabla_{y}F_{Y|X,Z}$, which is the density
$f_{Y|X,Z}$ of (\ref{eq:g_functional}); $V_{x}=\nabla_{x}V$; $\Phi_{x}=\nabla_{x}\Phi$;
and so on.

Let 
\begin{align*}
\widehat{S}_{1}(w) & :=\int^{x}_{0}\widehat{\Sigma}_{1}(w)\mathrm{d}u:=\int^{x}_{0}\frac{\nabla_{x}\widehat{V}(u|z)\nabla_{z}\widehat{F}(y|u,z)}{\nabla_{z}\widehat{V}(u|z)\nabla_{y}\widehat{F}(y|u,z)}\mathrm{d}u,\\
\widehat{S}_{2}(w) & :=\int^{x}_{0}\widehat{\Sigma}_{2}(w)\mathrm{d}u:=\int^{x}_{0}\frac{\nabla_{x}\widehat{F}(y|u,z)}{\nabla_{y}\widehat{F}(y|u,z)}\mathrm{d}u.
\end{align*}
Then $\widehat{S}(w)$ can be written as $\widehat{S}(w)=\widehat{S}_{1}(w)-\widehat{S}_{2}(w)$.
In the same manner, let us denote $S_{1}(w)$ and $S_{2}(w)$ by the
population counterparts for $\widehat{S}_{1}(w)$ and $\widehat{S}_{2}(w)$.

Define the functionals $\nabla_{\Phi}S_{2}(w)$ and $\nabla_{f}S_{2}(w)$
by 
\begin{align*}
\nabla_{\Phi}S_{2}(w)[\Phi] & :=\int^{x}_{0}D^{2}_{\Phi}(y,u,z)\Phi(y,u,z)\mathrm{d}u+\int^{x}_{0}D^{2}_{\Phi_{x}}(y,u,z)\Phi_{x}(y,u,z)\mathrm{d}u\\
 & +\int^{x}_{0}D^{2}_{\Phi_{y}}(y,u,z)\Phi_{y}(y,u,z)\mathrm{d}u,\\
\nabla_{f}S_{2}(w)[f] & :=\int^{x}_{0}D^{2}_{f}(y,u,z)f(u,z)\mathrm{d}u+\int^{x}_{0}D^{2}_{f_{x}}(y,u,z)f_{x}(u,z)\mathrm{d}u,
\end{align*}
where 
\begin{align*}
D^{2}_{\Phi} & :=-\frac{f_{x}}{F_{y}f^{2}},\quad D^{2}_{\Phi_{x}}:=\frac{1}{F_{y}f},\quad D^{2}_{\Phi_{y}}:=-\frac{F_{x}}{F^{2}_{y}f},\quad D^{2}_{f}:=\frac{\Phi f_{x}}{F_{y}f^{3}},\quad D^{2}_{f_{x}}:=-\frac{\Phi}{F_{y}f^{2}}.
\end{align*}
The superscript indexes the two parts of $S=S_{1}-S_{2}$; the subscript
names the function the coefficient multiplies.
\begin{lem}
\label{lem:linear1} The following holds uniformly over $\mathcal{Y}_{\mu}\times\widetilde{\mathcal{X}}_{0}\times\mathcal{Z}_{\mu}$:
\[
\widehat{S}_{2}(w)-S_{2}(w)=\nabla_{\Phi}S_{2}(w)[\widehat{\Phi}-\Phi]+\nabla_{f}S_{2}(w)[\widehat{f}-f]+O_{p}(R_{n}).
\]
\end{lem}

\begin{proof}
The third term of the identity below divides by an estimated quantity,
so the conclusion is uniform only on the event on which the estimated
denominators are bounded away from zero. They are: Assumption \ref{assu: density}
(iii) bounds $f$, $f_{Y|X,Z}$, $\nabla_{z}V$ and $p$ away from
zero on $\mathcal{Y}_{\mu}\times\widetilde{\mathcal{X}}_{0}\times\mathcal{Z}_{\mu}$,
and the rates below make $\widehat{f}$, $\nabla_{y}\widehat{F}$,
$\nabla_{z}\widehat{V}$ and $\widehat{p}$ uniformly consistent for
them, so each is bounded away from zero there with probability tending
to one. All statements below are on that event. We use the following
identity repeatedly: 
\[
\frac{\hat{a}}{\hat{b}}-\frac{a}{b}=\frac{1}{b}\{\hat{a}-a\}-\frac{a}{b^{2}}\{\hat{b}-b\}+\frac{\{\hat{b}-b\}}{b\hat{b}}\left\{ \hat{a}-a-\frac{a(\hat{b}-b)}{b}\right\} .
\]
Note that $\widehat{\Sigma}_{2}(w)-\Sigma_{2}(w)$ is equal to 
\begin{align*}
 & \frac{1}{F_{y}}\{\widehat{F}_{x}-F_{x}\}-\frac{F_{x}}{F^{2}_{y}}\{\widehat{F}_{y}-F_{y}\}+O_{p}(\{\widehat{F}_{y}-F_{y}\}^{2})+O_{p}(\{\widehat{F}_{x}-F_{x}\}\{\widehat{F}_{y}-F_{y}\}).
\end{align*}
We will calculate $\widehat{F}_{y}-F_{y}$ and $\widehat{F}_{x}-F_{x}$.
Note that 
\begin{align*}
\widehat{F}_{y}-F_{y}= & \frac{1}{f}\{\widehat{\Phi}_{y}-\Phi_{y}\}-\frac{\Phi_{y}}{f^{2}}\{\widehat{f}-f\}+R_{y},
\end{align*}
where 
\[
R_{y}=O_{p}(\{\widehat{f}-f\}^{2})+O_{p}(\{\widehat{f}-f\}\{\widehat{\Phi}_{y}-\Phi_{y}\}).
\]
Also note that 
\begin{align*}
\widehat{F}_{x}-F_{x}= & -\frac{f_{x}}{f^{2}}\{\widehat{\Phi}-\Phi\}+\frac{1}{f}\{\widehat{\Phi}_{x}-\Phi_{x}\}+\left[\frac{2\Phi ff_{x}}{f^{4}}-\frac{\Phi_{x}}{f^{2}}\right]\{\widehat{f}-f\}-\frac{\Phi}{f^{2}}\{\widehat{f}_{x}-f_{x}\}+R_{x},
\end{align*}
where 
\begin{align*}
R_{x} & =O_{p}(\{\widehat{f}-f\}^{2})+O_{p}(\{\widehat{f}_{x}-f_{x}\}\{\widehat{\Phi}-\Phi\})\\
 & +O_{p}(\{\widehat{f}-f\}\{\widehat{\Phi}_{x}-\Phi_{x}\})+O_{p}(\{\widehat{f}-f\}^{2}\{\widehat{\Phi}-\Phi\})\\
 & +O_{p}(\{\widehat{f}-f\}^{2}\{\widehat{f}_{x}-f_{x}\})+O_{p}(\{\widehat{f}-f\}^{2}\{\widehat{\Phi}-\Phi\}\{\widehat{f}_{x}-f_{x}\}).
\end{align*}
Inserting these expressions into $\widehat{\Sigma}_{2}-\Sigma_{2}$
gives 
\begin{align*}
 & -\frac{f_{x}}{F_{y}f^{2}}\{\widehat{\Phi}-\Phi\}-\frac{F_{x}}{F^{2}_{y}}\frac{1}{f}\{\widehat{\Phi}_{y}-\Phi_{y}\}+\frac{1}{F_{y}f}\{\widehat{\Phi}_{x}-\Phi_{x}\}\\
 & +\frac{1}{F_{y}}\left[\frac{2\Phi ff_{x}}{f^{4}}-\frac{\Phi_{x}}{f^{2}}+\frac{F_{x}}{F_{y}}\frac{\Phi_{y}}{f^{2}}\right]\{\widehat{f}-f\}-\frac{1}{F_{y}}\frac{\Phi}{f^{2}}\{\widehat{f}_{x}-f_{x}\}+R_{2}
\end{align*}
where $R_{2}$ collects $F^{-1}_{y}R_{x}-F_{x}F^{-2}_{y}R_{y}$ together
with the two quadratic terms of the display for $\widehat{\Sigma}_{2}-\Sigma_{2}$
above. Every summand of $R_{2}$ is of degree at least two in the
estimation errors, so $\lVert R_{2}\rVert_{\infty}$ is bounded by
the square of the largest of the rates displayed below. The coefficient
of $\{\widehat{f}-f\}$ collapses to $D^{2}_{f}$ because $\Phi_{y}=fF_{y}$,
so that $\frac{F_{x}}{F_{y}}\frac{\Phi_{y}}{f^{2}}=\frac{F_{x}}{f}=\frac{\Phi_{x}}{f^{2}}-\frac{\Phi f_{x}}{f^{3}}$
and the bracket reduces to $\Phi f_{x}/f^{3}$. Using the definitions
of functionals, we have 
\[
\widehat{S}_{2}(w)-S_{2}(w)=\nabla_{\Phi}S_{2}(w)[\widehat{\Phi}-\Phi]+\nabla_{f}S_{2}(w)[\widehat{f}-f]+R_{2}.
\]

It remains to show that the remainder term satisfies $\sup_{w\in\mathcal{Y}_{\mu}\times\widetilde{\mathcal{X}}_{0}\times\mathcal{Z}_{\mu}}R_{2}(w)=O_{p}(R_{n})$.
We claim 
\begin{align*}
\left\lVert \widehat{\Phi}-\Phi\right\rVert _{\infty} & =O_{p}\left(\max\{h_{x},h_{z}\}^{m}\right)+O_{p}\left(\sqrt{\frac{\log n}{nh_{x}h_{z}}}\right),\\
\left\lVert \widehat{\Phi}_{y}-\Phi_{y}\right\rVert _{\infty} & =O_{p}\left(\max\{h_{y},h_{x},h_{z}\}^{m}\right)+O_{p}\left(\sqrt{\frac{\log n}{nh_{y}h_{x}h_{z}}}\right),\\
\left\lVert \widehat{\Phi}_{x}-\Phi_{x}\right\rVert _{\infty} & =O_{p}\left(\max\{h_{x},h_{z}\}^{m}\right)+O_{p}\left(\sqrt{\frac{\log n}{nh^{3}_{x}h_{z}}}\right),\\
\left\lVert \widehat{\Phi}_{z}-\Phi_{z}\right\rVert _{\infty} & =O_{p}\left(\max\{h_{x},h_{z}\}^{m}\right)+O_{p}\left(\sqrt{\frac{\log n}{nh_{x}h^{3}_{z}}}\right),\\
\left\lVert \widehat{f}-f\right\rVert _{\infty} & =O_{p}\left(\max\{h_{x},h_{z}\}^{m}\right)+O_{p}\left(\sqrt{\frac{\log n}{nh_{x}h_{z}}}\right),\\
\left\lVert \widehat{f}_{x}-f_{x}\right\rVert _{\infty} & =O_{p}\left(\max\{h_{x},h_{z}\}^{m}\right)+O_{p}\left(\sqrt{\frac{\log n}{nh^{3}_{x}h_{z}}}\right),\\
\left\lVert \widehat{f}_{z}-f_{z}\right\rVert _{\infty} & =O_{p}\left(\max\{h_{x},h_{z}\}^{m}\right)+O_{p}\left(\sqrt{\frac{\log n}{nh_{x}h^{3}_{z}}}\right),\\
\left\lVert \widehat{\Psi}-\Psi\right\rVert _{\infty} & =O_{p}\left(h^{m}_{z}\right)+O_{p}\left(\sqrt{\frac{\log n}{nh_{z}}}\right),\\
\left\lVert \widehat{\Psi}_{z}-\Psi_{z}\right\rVert _{\infty} & =O_{p}\left(h^{m}_{z}\right)+O_{p}\left(\sqrt{\frac{\log n}{nh^{3}_{z}}}\right),\\
\left\lVert \widehat{p}-p\right\rVert _{\infty} & =O_{p}\left(h^{m}_{z}\right)+O_{p}\left(\sqrt{\frac{\log n}{nh_{z}}}\right),\\
\left\lVert \widehat{p}_{z}-p_{z}\right\rVert _{\infty} & =O_{p}\left(h^{m}_{z}\right)+O_{p}\left(\sqrt{\frac{\log n}{nh^{3}_{z}}}\right).
\end{align*}
 The rates above are claimed uniformly in $y$ and $x$ as well, which
the uniform convergence result by \citet{hansen2008uniform} does
not give. A class $\mathcal{F}$ with envelope $E$ is VC type with
characteristics $(A,v)$, $A\ge e$ and $v\ge1$, when
\begin{equation}
\sup_{Q}N\big(\mathcal{F},L_{2}(Q),\varepsilon\lVert E\rVert_{L_{2}(Q)}\big)\le\left(\frac{A}{\varepsilon}\right)^{v},\qquad0<\varepsilon\le1,\label{eq:VC-dim}
\end{equation}
the supremum being over finitely discrete probability measures $Q$.
The uniformity in those indices follows from a VC argument: the indicators
form a VC class, the kernel families are VC type under Assumption
\ref{assu: kernel}, a pointwise product of VC type classes is VC
type, and the exponential inequality of \citet{GineGuillou02} applies
to the product. Lemma \ref{lem:vc} below uses the same tools for
the class the band is built from.

The display has eleven rows for the twelve coordinates, and its exponents
follow two patterns. A coordinate that is a derivative has one more
power of the bandwidth in its rate. Differentiating $\widehat{\Phi}$
in $x$ puts a factor $h^{-1}_{x}$ in front of a kernel average of
the usual form. The variance of $\widehat{\Phi}_{x}$ is therefore
larger by $h^{-2}_{x}$, and its uniform deviation by $h^{-1}_{x}$.
$\widehat{\Phi}_{y}$ is not a derivative but a three-dimensional
average in its own right, which is why its rate has only the one extra
power of $h_{y}$. The missing row is $\Psi_{x}$: it equals $f$
identically, since $\Psi(x,z)=\int^{x}f_{X,Z}$, so $\widehat{\Psi}_{x}=\widehat{f}$
and the row for $\widehat{f}$ serves for both. The rates above give
$R_{2}(w)=O_{p}(R_{n})$ uniformly over $\mathcal{Y}_{\mu}\times\widetilde{\mathcal{X}}_{0}\times\mathcal{Z}_{\mu}$,
and Assumption \ref{assu:rate} makes it $o_{p}((nh_{x})^{-1/2})$.
\end{proof}

We define functionals $\nabla_{\alpha}S_{1}(w)$ for $\alpha\in\left\{ \Phi,\Psi,f,p\right\} $
by 
\begin{align*}
\nabla_{\Phi}S_{1}(w)[\Phi] & :=\int^{x}_{0}D^{1}_{\Phi}(y,u,z)\Phi(y,u,z)\mathrm{d}u+\int^{x}_{0}D^{1}_{\Phi_{y}}(y,u,z)\Phi_{y}(y,u,z)\mathrm{d}u\\
 & \quad+\int^{x}_{0}D^{1}_{\Phi_{z}}(y,u,z)\Phi_{z}(y,u,z)\mathrm{d}u,\\
\nabla_{\Psi}S_{1}(w)[\Psi] & :=\int^{x}_{0}D^{1}_{\Psi}(y,u,z)\Psi(u,z)\mathrm{d}u+\int^{x}_{0}D^{1}_{\Psi_{x}}(y,u,z)\Psi_{x}(u,z)\mathrm{d}u\\
 & \quad+\int^{x}_{0}D^{1}_{\Psi_{z}}(y,u,z)\Psi_{z}(u,z)\mathrm{d}u\\
\nabla_{f}S_{1}(w)[f] & :=\int^{x}_{0}D^{1}_{f}(y,u,z)f(u,z)\mathrm{d}u+\int^{x}_{0}D^{1}_{f_{z}}(y,u,z)f_{z}(u,z)\mathrm{d}u,\\
\nabla_{p}S_{1}(w)[p] & :=\int^{x}_{0}D^{1}_{p}(y,u,z)p(z)\mathrm{d}u+\int^{x}_{0}D^{1}_{p_{z}}(y,u,z)p_{z}(z)\mathrm{d}u,
\end{align*}
where 
\begin{align*}
D^{1}_{\Phi} & :=-\frac{V_{x}}{V_{z}}\frac{f_{z}}{F_{y}f^{2}},\quad D^{1}_{\Phi_{z}}:=\frac{V_{x}}{V_{z}}\frac{1}{F_{y}f},\quad D^{1}_{\Phi_{y}}:=-\frac{V_{x}F_{z}}{V_{z}F^{2}_{y}f},\\
D^{1}_{f} & :=\frac{V_{x}}{V_{z}}\frac{\Phi f_{z}}{F_{y}f^{3}},\quad D^{1}_{f_{z}}:=-\frac{V_{x}}{V_{z}}\frac{\Phi}{F_{y}f^{2}},\\
D^{1}_{\Psi_{x}} & :=\frac{F_{z}}{F_{y}}\frac{1}{pV_{z}},\quad D^{1}_{\Psi_{z}}:=-\frac{F_{z}}{F_{y}}\frac{V_{x}}{pV^{2}_{z}},\quad D^{1}_{\Psi}:=\frac{F_{z}}{F_{y}}\frac{V_{x}p_{z}}{p^{2}V^{2}_{z}},\\
D^{1}_{p} & :=-\frac{F_{z}}{F_{y}}\frac{V_{x}Vp_{z}}{p^{2}V^{2}_{z}},\quad D^{1}_{p_{z}}:=\frac{F_{z}}{F_{y}}\frac{V_{x}V}{pV^{2}_{z}}.
\end{align*}

These follow from the factorization $\Sigma_{1}=(V_{x}/V_{z})\cdot(F_{z}/F_{y})$,
whose two factors depend on disjoint sets of arguments: $F_{z}/F_{y}$
on $(\Phi,\Phi_{y},\Phi_{z},f,f_{z})$ and $V_{x}/V_{z}$ on $(\Psi,\Psi_{x},\Psi_{z},p,p_{z})$.
Since $F_{z}/F_{y}$ is $\Sigma_{2}$ with $x$ replaced by $z$,
the first two rows are the previous display with that replacement,
multiplied by $V_{x}/V_{z}$.
\begin{lem}
\label{lem:linear2} The following holds uniformly over $\mathcal{Y}_{\mu}\times\widetilde{\mathcal{X}}_{0}\times\mathcal{Z}_{\mu}$:
\begin{align*}
\widehat{S}_{1}(w)-S_{1}(w) & =\nabla_{\Phi}S_{1}(w)[\widehat{\Phi}-\Phi]+\nabla_{\Psi}S_{1}(w)[\widehat{\Psi}-\Psi]\\
 & +\nabla_{f}S_{1}(w)[\widehat{f}-f]+\nabla_{p}S_{1}(w)[\widehat{p}-p]+O_{p}(R_{n}).
\end{align*}
\end{lem}

\begin{proof}
The conclusion can be derived in the same manner as in Lemma \ref{lem:linear1}.
The main difference is that $\widehat{S}_{1}$ involves the ratio
$\widehat{\Sigma}_{1}=\nabla_{x}\widehat{V}\cdot\nabla_{z}\widehat{F}/(\nabla_{z}\widehat{V}\cdot\nabla_{y}\widehat{F})$,
which depends on $(\widehat{\Phi},\widehat{\Psi},\widehat{f},\widehat{p})$
and their derivatives, whereas $\widehat{S}_{2}$ depends only on
$(\widehat{\Phi},\widehat{f})$. The linearization proceeds by the
same ratio identity used in Lemma \ref{lem:linear1}, applied to each
ratio in $\widehat{\Sigma}_{1}$, yielding additional terms in $\widehat{\Psi}-\Psi$
and $\widehat{p}-p$. The remainder is controlled by the same uniform
convergence rates under Assumption \ref{assu:rate}.
\end{proof}

We define the functionals $\nabla_{\alpha}S_{*}(w)$ for $\alpha\in\left\{ \Phi,\Psi,f,p\right\} $
by 
\begin{align*}
\nabla_{\Phi}S_{*}(w)\left[\Phi\right] & :=\left\{ \nabla_{\Phi}S_{1}(w)-\nabla_{\Phi}S_{2}(w)\right\} \left[\Phi\right],\\
\nabla_{\Psi}S_{*}(w)\left[\Psi\right] & :=\nabla_{\Psi}S_{1}(w)[\Psi],\\
\nabla_{f}S_{*}(w)\left[f\right] & :=\left\{ \nabla_{f}S_{1}(w)-\nabla_{f}S_{2}(w)\right\} [f],\\
\nabla_{p}S_{*}(w)\left[p\right] & :=\nabla_{p}S_{1}(w)[p].
\end{align*}
In the same notation $D^{*}_{\alpha}=D^{1}_{\alpha}-D^{2}_{\alpha}$,
with either term taken as zero at an index the corresponding functional
does not carry; in particular $D^{*}_{\Phi_{x}}=-D^{2}_{\Phi_{x}}$
and $D^{*}_{\Psi_{x}}=D^{1}_{\Psi_{x}}$. For $x<0$ the integral
$\int^{x}_{0}$ runs backwards, so the indicator of the interval between
the normalization point and $x$ carries a sign. Write
\[
I_{x}(t):=\boldsymbol{1}\{t\le x\}-\boldsymbol{1}\{t\le0\}=\begin{cases}
\boldsymbol{1}\{0<t\le x\}, & x\ge0,\\
-\boldsymbol{1}\{x<t\le0\}, & x<0,
\end{cases}
\]
which is the limit of $\int^{x}_{0}h^{-1}_{x}K((t-u)/h_{x})\mathrm{d}u$
for either sign of $x$. Also define a function $\rho_{x}$ that collects
the terms of the linearization carrying no bandwidth, by 
\begin{align}
\rho_{x}(W_{i}) & =\int\mu(y,Z_{i})\boldsymbol{1}\left\{ Y_{i}\le y\right\} I_{x}(X_{i})D^{*}_{\Phi}(y,X_{i},Z_{i})\mathrm{d}y\nonumber \\
 & +\mu(Y_{i},Z_{i})I_{x}(X_{i})D^{*}_{\Phi_{y}}(Y_{i},X_{i},Z_{i})\nonumber \\
 & -\int\mu(y,Z_{i})\boldsymbol{1}\left\{ Y_{i}\le y\right\} I_{x}(X_{i})\partial_{u}D^{*}_{\Phi_{x}}(y,X_{i},Z_{i})\mathrm{d}y\nonumber \\
 & -\int\boldsymbol{1}\left\{ Y_{i}\le y\right\} I_{x}(X_{i})\frac{\partial\mu(y,Z_{i})D^{*}_{\Phi_{z}}(y,X_{i},Z_{i})}{\partial z}\mathrm{d}y\nonumber \\
 & +\int\int^{x}_{0}\boldsymbol{1}\{X_{i}\le u\}\mu(y,Z_{i})D^{*}_{\Psi}(y,u,Z_{i})\mathrm{d}u\mathrm{d}y+\int\mu(y,Z_{i})I_{x}(X_{i})D^{*}_{\Psi_{x}}(y,X_{i},Z_{i})\mathrm{d}y\nonumber \\
 & -\int^{x}_{0}\boldsymbol{1}\{X_{i}\le u\}\int\frac{\partial\mu(y,Z_{i})D^{*}_{\Psi_{z}}(y,u,Z_{i})}{\partial z}\mathrm{d}y\mathrm{d}u\nonumber \\
 & +\int\mu(y,Z_{i})I_{x}(X_{i})D^{*}_{f}(y,X_{i},Z_{i})\mathrm{d}y\nonumber \\
 & -\int\mu(y,Z_{i})I_{x}(X_{i})\partial_{u}D^{*}_{f_{x}}(y,X_{i},Z_{i})\mathrm{d}y\nonumber \\
 & -\int I_{x}(X_{i})\left[\frac{\partial\mu(y,Z_{i})D^{*}_{f_{z}}(y,X_{i},Z_{i})}{\partial z}\right]\mathrm{d}y\nonumber \\
 & +\int^{x}_{0}\int\mu(y,Z_{i})D^{*}_{p}(y,u,Z_{i})\mathrm{d}y\mathrm{d}u-\int^{x}_{0}\int\left[\frac{\partial\mu(y,Z_{i})D^{*}_{p_{z}}(y,u,Z_{i})}{\partial z}\right]\mathrm{d}y\mathrm{d}u,\label{eq:AL}
\end{align}
where 
\begin{align*}
D^{*}_{\Phi} & =D^{1}_{\Phi}-D^{2}_{\Phi},\quad D^{*}_{\Phi_{y}}=D^{1}_{\Phi_{y}}-D^{2}_{\Phi_{y}},\\
D^{*}_{\Phi_{x}} & =-D^{2}_{\Phi_{x}},\quad D^{*}_{\Phi_{z}}=D^{1}_{\Phi_{z}};\\
D^{*}_{\Psi} & =D^{1}_{\Psi},\quad D^{*}_{\Psi_{x}}=D^{1}_{\Psi_{x}},\quad D^{*}_{\Psi_{z}}=D^{1}_{\Psi_{z}};\\
D^{*}_{f} & =D^{1}_{f}-D^{2}_{f},\quad D^{*}_{f_{x}}=-D^{2}_{f_{x}},\quad D^{*}_{f_{z}}=D^{1}_{f_{z}};\\
D^{*}_{p} & =D^{1}_{p},\quad D^{*}_{p_{z}}=D^{1}_{p_{z}}.
\end{align*}

The two blocks whose integrand is a derivative in $u$, the variable
the integral runs over, are integrated by parts in $u$ in the proof
below, and each leaves a boundary term at $u=0$ and at $u=x$. For
$a\in\widetilde{\mathcal{X}}_{0}$ write 
\begin{align}
b_{n}(a) & =\frac{1}{n}\sum^{n}_{i=1}\frac{1}{h_{x}}K\left(\frac{X_{i}-a}{h_{x}}\right)\int\int\mu(y,z)\frac{1}{h_{z}}K\left(\frac{Z_{i}-z}{h_{z}}\right)\left\{ \boldsymbol{1}\left\{ Y_{i}\le y\right\} D^{*}_{\Phi_{x}}(y,a,z)+D^{*}_{f_{x}}(y,a,z)\right\} \mathrm{d}z\mathrm{d}y,\nonumber \\
B_{n}(x) & =b_{n}(x)-b_{n}(0).\label{eq:Bn}
\end{align}
so that $B_{n}(x)$ collects the two boundary terms.
\begin{lem}
\label{lem:AL}The following holds uniformly over $\mathcal{X}_{0}$:
\begin{align*}
\widehat{g}_{\mathrm{LS}}(x)-g(x) & =\int\int\mu(y,z)\left\{ \nabla_{\Phi}S_{*}(y,x,z)[\widehat{\Phi}-\Phi]+\nabla_{\Psi}S_{*}(y,x,z)[\widehat{\Psi}-\Psi]\right.\\
 & +\left.\nabla_{f}S_{*}(y,x,z)[\widehat{f}-f]+\nabla_{p}S_{*}(y,x,z)[\widehat{p}-p]\right\} \mathrm{d}y\mathrm{d}z+O_{p}(R_{n})\\
 & =\frac{1}{n}\sum^{n}_{i=1}\rho_{x}(W_{i})+B_{n}(x)+O_{p}(R_{n}+\bar{h}^{m}+\sqrt{h_{x}/n}).
\end{align*}
\end{lem}

\begin{proof}
Because $S(y,x,z)=g(x)$ and $\mu$ integrates to one, $\widehat{g}_{\mathrm{LS}}(x)-g(x)=\int\int\mu(y,z)\{\widehat{S}(y,x,z)-S(y,x,z)\}\mathrm{d}y\mathrm{d}z$.
Lemmas \ref{lem:linear1} and \ref{lem:linear2} linearize the integrand,
which gives the first equality of the statement. 

For the second equality, each of $\Sigma_{1}$ and $\Sigma_{2}$ is
homogeneous of degree zero in $(\Phi,\Phi_{x},\Phi_{y},\Phi_{z},f,f_{x},f_{z})$,
and $\Sigma_{1}$ is homogeneous of degree zero in $(\Psi,\Psi_{x},\Psi_{z},p,p_{z})$
as well: scaling either group by a common factor leaves the ratios
$F=\Phi/f$ and $V=\Psi/p$ and all of their derivatives unchanged.
Euler's relation for a function homogeneous of degree zero therefore
gives
\begin{equation}
\nabla_{\Phi}S_{*}(y,x,z)[\Phi]+\nabla_{\Psi}S_{*}(y,x,z)[\Psi]+\nabla_{f}S_{*}(y,x,z)[f]+\nabla_{p}S_{*}(y,x,z)[p]=0.\label{eq:zero_equality}
\end{equation}
By (\ref{eq:zero_equality}) the four functionals annihilate the population
quantities, so the sum in the first equality is unchanged when $\widehat{\Phi}-\Phi,\widehat{\Psi}-\Psi,\widehat{f}-f,\widehat{p}-p$
are replaced by $\widehat{\Phi},\widehat{\Psi},\widehat{f},\widehat{p}$.
Let 
\begin{align*}
\nabla_{\Phi}S_{*}(y,x,z)[\widehat{\Phi}] & =\nabla_{\Phi}S_{1}(y,x,z)[\widehat{\Phi}]-\nabla_{\Phi}S_{2}(y,x,z)[\widehat{\Phi}]\\
 & =\int^{x}_{0}D^{*}_{\Phi}(y,u,z)\widehat{\Phi}(y,u,z)\mathrm{d}u+\int^{x}_{0}D^{*}_{\Phi_{y}}(y,u,z)\widehat{\Phi}_{y}(y,u,z)\mathrm{d}u\\
 & +\int^{x}_{0}D^{*}_{\Phi_{x}}(y,u,z)\widehat{\Phi}_{x}(y,u,z)\mathrm{d}u+\int^{x}_{0}D^{*}_{\Phi_{z}}(y,u,z)\widehat{\Phi}_{z}(y,u,z)\mathrm{d}u\\
 & =:A_{\widehat{\Phi}}(y,x,z)+A_{\widehat{\Phi}_{y}}(y,x,z)+A_{\widehat{\Phi}_{x}}(y,x,z)+A_{\widehat{\Phi}_{z}}(y,x,z),
\end{align*}
where $D^{*}_{\Phi}=D^{1}_{\Phi}-D^{2}_{\Phi}$, $D^{*}_{\Phi_{y}}=D^{1}_{\Phi_{y}}-D^{2}_{\Phi_{y}}$,
$D^{*}_{\Phi_{x}}=-D^{2}_{\Phi_{x}}$, and $D^{*}_{\Phi_{z}}=D^{1}_{\Phi_{z}}$.
The functionals in $\Psi$, $f$ and $p$ split in the same way, into
$A_{\widehat{\Psi}},A_{\widehat{\Psi}_{x}},A_{\widehat{\Psi}_{z}}$,
into $A_{\widehat{f}},A_{\widehat{f}_{x}},A_{\widehat{f}_{z}}$, and
into $A_{\widehat{p}},A_{\widehat{p}_{z}}$. The linearization therefore
carries one block $A_{\widehat{\alpha}}$ for each of the twelve coordinates
\[
\alpha\in\{\Phi,\Phi_{x},\Phi_{y},\Phi_{z},\Psi,\Psi_{x},\Psi_{z},f,f_{x},f_{z},p,p_{z}\},
\]
and the counts of twelve below refer to these. First, we claim that
\begin{align}
 & \int\int\mu(y,z)A_{\widehat{\Phi}}(y,x,z)\mathrm{d}y\mathrm{d}z\nonumber \\
= & \frac{1}{n}\sum^{n}_{i=1}\int\mu(y,Z_{i})\boldsymbol{1}\left\{ Y_{i}\le y\right\} I_{x}(X_{i})D^{*}_{\Phi}(y,X_{i},Z_{i})\mathrm{d}y+o_{p}(n^{-1/2}).\label{eq:first_term}
\end{align}
Note that the first term of (\ref{eq:first_term}) corresponds to
the first term of $\rho_{x}(W_{i})$. To see this equation, note that
by the definitions of $A_{\widehat{\Phi}}$ and $\widehat{\Phi}$,
we have 
\begin{align*}
 & \int\int\mu(y,z)A_{\widehat{\Phi}}(y,x,z)\mathrm{d}y\mathrm{d}z\\
= & \frac{1}{n}\sum^{n}_{i=1}\int\boldsymbol{1}\left\{ Y_{i}\le y\right\} \int^{x}_{0}\frac{1}{h_{x}}K\left(\frac{X_{i}-u}{h_{x}}\right)\left[\int\mu(y,z)D^{*}_{\Phi}(y,u,z)\frac{1}{h_{z}}K\left(\frac{Z_{i}-z}{h_{z}}\right)\mathrm{d}z\right]\mathrm{d}u\mathrm{d}y.
\end{align*}
By the usual change-of-variable argument, 
\[
\int\mu(y,z)D^{*}_{\Phi}(y,u,z)\frac{1}{h_{z}}K\left(\frac{Z_{i}-z}{h_{z}}\right)\mathrm{d}z=\mu(y,Z_{i})D^{*}_{\Phi}(y,u,Z_{i})+O_{p}(h^{m}_{z}).
\]
Inserting this expression into the previous display yields 
\begin{align}
 & \int\int\mu(y,z)A_{\widehat{\Phi}}(y,x,z)\mathrm{d}y\mathrm{d}z\nonumber \\
= & \frac{1}{n}\sum^{n}_{i=1}\int\mu(y,Z_{i})\boldsymbol{1}\left\{ Y_{i}\le y\right\} \left[\int^{x}_{0}\frac{1}{h_{x}}K\left(\frac{X_{i}-u}{h_{x}}\right)D^{*}_{\Phi}(y,u,Z_{i})\mathrm{d}u\right]\mathrm{d}y+O_{p}(h^{m}_{z}).\label{eq:1}
\end{align}
By applying change-of-variable argument again, we obtain 
\begin{align*}
\mathbb{E}_{X|YZ}\left[\left.\int^{x}_{0}\frac{1}{h_{x}}K\left(\frac{X_{i}-u}{h_{x}}\right)D^{*}_{\Phi}(y,u,Z_{i})\mathrm{d}u\right|Y_{i},Z_{i}\right]=\int^{x}_{0}D^{*}_{\Phi}(y,u,Z_{i})f_{X|YZ}(u|Y_{i},Z_{i})\mathrm{d}u+O(h^{m}_{x}).
\end{align*}
From this expression and the law of iterated expectations, 
\begin{align}
 & \mathbb{E}\left[\frac{1}{n}\sum^{n}_{i=1}\int\mu(y,Z_{i})\boldsymbol{1}\left\{ Y_{i}\le y\right\} \left[\int^{x}_{0}\frac{1}{h_{x}}K\left(\frac{X_{i}-u}{h_{x}}\right)D^{*}_{\Phi}(y,u,Z_{i})\mathrm{d}u\right]\mathrm{d}y\right]\nonumber \\
= & \mathbb{E}_{YZ}\left[\int\mu(y,Z_{i})\boldsymbol{1}\left\{ Y_{i}\le y\right\} \left[\int^{x}_{0}D^{*}_{\Phi}(y,u,Z_{i})f_{X|YZ}(u|Y_{i},Z_{i})\mathrm{d}u\right]\mathrm{d}y\right]+O(h^{m}_{x}).\label{eq:2}
\end{align}
Note that by (\ref{eq:1}) 
\begin{align*}
 & \left|\int\int\mu(y,z)A_{\widehat{\Phi}}(y,x,z)\mathrm{d}y\mathrm{d}z-\frac{1}{n}\sum^{n}_{i=1}\int\mu(y,Z_{i})\boldsymbol{1}\left\{ Y_{i}\le y\right\} I_{x}(X_{i})D^{*}_{\Phi}(y,X_{i},Z_{i})\mathrm{d}y\right|\\
\le & \left|\frac{1}{n}\sum^{n}_{i=1}\delta_{x}(W_{i})\right|+O_{p}(h^{m}_{z}),
\end{align*}
where 
\begin{align*}
\delta_{x}(W_{i}):= & \int\mu(y,Z_{i})\boldsymbol{1}\left\{ Y_{i}\le y\right\} \left\{ \int^{x}_{0}\frac{1}{h_{x}}K\left(\frac{X_{i}-u}{h_{x}}\right)D^{*}_{\Phi}(y,u,Z_{i})\mathrm{d}u\right.\\
 & \left.-I_{x}(X_{i})D^{*}_{\Phi}(y,X_{i},Z_{i})\right\} \mathrm{d}y,
\end{align*}
the two terms inside the braces differing only in that the indicator
$I_{x}(X_{i})$ has been smoothed. By (\ref{eq:2}) and the law of
iterated expectations, $\mathbb{E}\left[\delta_{x}(W_{i})\right]=O(h^{m}_{x})$,
the contributions from the two endpoints cancelling, and this is $o((nh_{x})^{-1/2})$
by the undersmoothing condition of Theorem \ref{thm:pointwise}, Assumption
\ref{assu:rate} imposing nothing on $h^{m}_{x}$. It remains to bound
the centered average. The two terms inside the braces agree except
where the smoothing of the indicator matters, so $\delta_{x}(W_{i})$
is of order one only when $X_{i}$ lies within $O(h_{x})$ of $0$
or of $x$, the tail bound of Assumption \ref{assu: kernel}(i) controlling
it at larger distances. That event has probability $O(h_{x})$, because
$f$ is bounded under Assumption \ref{assu: density}, so $\lVert\delta_{x}\rVert^{2}_{L_{2}(P_{0})}=O(h_{x})$.
The centered average $n^{-1}\sum^{n}_{i=1}\{\delta_{x}(W_{i})-\mathbb{E}\delta_{x}(W_{i})\}$
has variance at most $\lVert\delta_{x}\rVert^{2}_{L_{2}(P_{0})}/n$.
It is therefore $O_{p}(\sqrt{h_{x}/n})$, which is $o_{p}(n^{-1/2})$.
The same computation holds uniformly in $x$. Thus we obtain (\ref{eq:first_term}).
Note that $O_{p}(\max\{h^{m}_{y},h^{m}_{z}\})=o_{p}((nh_{x})^{-1/2})$
by Assumption \ref{assu:rate}, and $O_{p}(h^{m}_{x})=o_{p}((nh_{x})^{-1/2})$
by the undersmoothing condition of Theorem \ref{thm:pointwise}.

The same argument applies to the other blocks. Ten of the twelve pass
to a sample average. Specifically, the six in $\{\Phi,\Phi_{y},\Psi,\Psi_{x},f,p\}$
become asymptotically linear by the change of variable just used,
as $A_{\widehat{\Phi}}$ did. The four carrying a derivative in $z$,
$\{\Phi_{z},\Psi_{z},f_{z},p_{z}\}$, become asymptotically linear
after an integration by parts in $z$, whose boundary terms vanish
because $\mu$ has compact support under Assumption \ref{assu: density}(ii).
These ten terms are ten of the twelve terms of $\rho_{x}$ in (\ref{eq:AL}). 

The last two, $A_{\widehat{\Phi}_{x}}$ and $A_{\widehat{f}_{x}}$,
are differentiated in $u$, the variable the integral runs over, and
they are not asymptotically linear terms. Write $\widehat{\Phi}_{x}(y,u,z)=\partial_{u}\widehat{\Phi}(y,u,z)$
and integrate by parts in $u$:
\begin{align*}
\int^{x}_{0}D^{*}_{\Phi_{x}}(y,u,z)\widehat{\Phi}_{x}(y,u,z)\mathrm{d}u & =\frac{1}{n}\sum^{n}_{i=1}\int\boldsymbol{1}\left\{ Y_{i}\le y\right\} \frac{1}{h_{z}}K\left(\frac{Z_{i}-z}{h_{z}}\right)\\
 & \times\left\{ \left[\frac{1}{h_{x}}K\left(\frac{X_{i}-u}{h_{x}}\right)D^{*}_{\Phi_{x}}(y,u,z)\right]^{x}_{0}-\int^{x}_{0}\frac{1}{h_{x}}K\left(\frac{X_{i}-u}{h_{x}}\right)\partial_{u}D^{*}_{\Phi_{x}}(y,u,z)\mathrm{d}u\right\} .
\end{align*}
In the second term the kernel is inside the integral, so the argument
that produced (\ref{eq:first_term}) applies to it unchanged and returns
$-I_{x}(X_{i})\partial_{u}D^{*}_{\Phi_{x}}(y,X_{i},Z_{i})$ up to
$o_{p}((nh_{x})^{-1/2})$. The second term is one of the two remaining
terms of $\rho_{x}$. The first term is a boundary term. The same
integration by parts applies to $A_{\widehat{f}_{x}}$, with $D^{*}_{f_{x}}$
in place of $D^{*}_{\Phi_{x}}$ and without the indicator. The boundary
terms of $A_{\widehat{\Phi}_{x}}$ and $A_{\widehat{f}_{x}}$ together
are $B_{n}(x)$ of (\ref{eq:Bn}). 
\end{proof}

\begin{lem}
\label{lem:donsker} The class $\left\{ \rho_{x}:x\in\mathcal{X}_{0}\right\} $
is $P_{0}$-Donsker.
\end{lem}

\begin{proof}
Each term of (\ref{eq:AL}) takes one of the following forms:
\begin{enumerate}
\item[(i)] $\vartheta(w)\boldsymbol{1}\{a(w)\le x\}$, where $\vartheta$ does
not depend on $x$ and $a(w)=X_{i}$;
\item[(ii)] $\int^{x}_{0}\vartheta(u,w)\mathrm{d}u$ or $\int^{x}_{0}\boldsymbol{1}\{X_{i}\le u\}\vartheta(u,w)\mathrm{d}u$,
where $\vartheta$ is bounded.
\end{enumerate}
Type (i) arises from the terms involving $I_{x}(X_{i})D^{*}_{\alpha}(y,X_{i},Z_{i})$,
where $D^{*}_{\alpha}$ is evaluated at $X_{i}$ rather than $x$.
Type (ii) arises from the integrals in $u$ left over after the linearization,
with or without the indicator $\boldsymbol{1}\{X_{i}\le u\}$. The
boundary terms of Lemma \ref{lem:boundary} are not members of $\left\{ \rho_{x}\right\} $:
they carry a kernel, and Lemma \ref{lem:vc} shows that their class
is not Donsker.

We now verify the Donsker property of each type. For type (i), the
class $\{w\mapsto\vartheta(w)\boldsymbol{1}\{X_{i}\le x\}:x\in\mathcal{X}_{0}\}$
is a VC-subgraph class with VC-index at most 3, hence Donsker provided
$\mathbb{E}\left[\vartheta(W)^{2}\right]<\infty$ \citep[Example 2.6.23]{vdvW1996}.
A term $\vartheta(w)I_{x}(X_{i})$ is of this type up to the single
function $\vartheta(w)\boldsymbol{1}\{X_{i}\le0\}$, which does not
depend on $x$; subtracting one square-integrable function from every
member of a Donsker class leaves it Donsker. The splitting has to
keep the factor $\boldsymbol{1}\{X_{i}\in\widetilde{\mathcal{X}}_{0}\}$,
which $I_{x}(X_{i})$ carries and $\boldsymbol{1}\{X_{i}\le x\}$
alone does not: off $\widetilde{\mathcal{X}}_{0}$ the coefficients
are divided by quantities Assumption \ref{assu: density} (iii) bounds
away from zero only on it, and $\nabla_{z}V$ vanishes at the edge
of the support of $X$. Multiplying by a fixed function preserves
the VC-subgraph property, and the envelope is then bounded. For type
(ii), either map has an $x$-derivative bounded in absolute value
by $\left\lVert \vartheta(\cdot,w)\right\rVert _{\infty}$ and is
therefore Lipschitz with that constant, which is bounded under Assumption
\ref{assu: density}. By \citet[Theorem 2.7.11]{vdvW1996}, the bracketing
numbers of this class are bounded by the covering numbers of the compact
set $\mathcal{X}_{0}$, which satisfy $N(\varepsilon,\mathcal{X}_{0})\le C/\varepsilon$.
Hence the bracketing integral is finite and the class is Donsker.

The two classes above have a common square-integrable envelope, and
a finite sum of Donsker classes with such an envelope is Donsker \citep[Theorem 2.10.6 and Example 2.10.7]{vdvW1996}.
\end{proof}

\begin{lem}
\label{lem:boundary} For $a\in\mathcal{X}$ put
\begin{align}
c_{a}(W_{i}) & =\int\frac{\mu(y,Z_{i})\left[F(y|a,Z_{i})-\boldsymbol{1}\left\{ Y_{i}\le y\right\} \right]}{F_{y}(y|a,Z_{i})f(a,Z_{i})}\mathrm{d}y,\nonumber \\
T_{n}(a) & =\frac{1}{nh_{x}}\sum^{n}_{i=1}K\left(\frac{X_{i}-a}{h_{x}}\right)c_{a}(W_{i}).\label{eq:Tn}
\end{align}
Then, uniformly over $\mathcal{X}_{0}$: (i) $B_{n}(x)=T_{n}(x)-T_{n}(0)+O_{p}(h^{m}_{z})$;
(ii) 
\[
\widehat{g}_{\mathrm{LS}}(x)-g(x)=T_{n}(x)-T_{n}(0)+O_{p}(n^{-1/2})+O_{p}(R_{n}+\bar{h}^{m}),
\]
(iii) $\mathbb{E}\left[c_{a}(W_{i})|X_{i}=a,Z_{i}\right]=0$ and $\mathbb{E}\left[T_{n}(a)\right]=O(h^{m}_{x})$;
and (iv) 
\[
nh_{x}\mathrm{Var}\left(T_{n}(x)-T_{n}(0)\right)\rightarrow\sigma^{2}(x)+\sigma^{2}(0),\quad\text{where }\sigma^{2}(a)=\left[\int K(v)^{2}\mathrm{d}v\right]f_{X}(a)\mathbb{E}\left[c_{a}(W_{i})^{2}|X_{i}=a\right].
\]
\end{lem}

\begin{proof}
Part (i). Convolving in $z$ and integrating against $\mu$ as before,
and using $D^{*}_{\Phi_{x}}=-1/(F_{y}f)$ and $D^{*}_{f_{x}}=F/(F_{y}f)$,
the two boundaries at $u=a$ combine into
\[
\frac{1}{nh_{x}}\sum^{n}_{i=1}K\left(\frac{X_{i}-a}{h_{x}}\right)\int\frac{\mu(y,Z_{i})\left[F(y|a,Z_{i})-\boldsymbol{1}\left\{ Y_{i}\le y\right\} \right]}{F_{y}(y|a,Z_{i})f(a,Z_{i})}\mathrm{d}y=T_{n}(a),
\]
the indicator entering through $\widehat{\Phi}_{x}$ and the distribution
function through $\widehat{f}_{x}$. The error the convolution in
$z$ makes is not stochastic. For each $y$ and $a$ the map $z\mapsto\mu(y,z)D^{*}(y,a,z)$
is $m$ times continuously differentiable with compact support under
Assumption \ref{assu: density}, so 
\[
\sup\left|\int\mu(y,z)D^{*}(y,a,z)\frac{1}{h_{z}}K\left(\frac{Z_{i}-z}{h_{z}}\right)\mathrm{d}z-\mu(y,Z_{i})D^{*}(y,a,Z_{i})\right|=O(h^{m}_{z}),
\]
where the supremum is taken over $y$, $a$ and $Z_{i}$, and the
bound being deterministic. Multiplying it by $(nh_{x})^{-1}\sum^{n}_{i=1}\left|K((X_{i}-a)/h_{x})\right|=O_{p}(1)$
gives $b_{n}(a)=T_{n}(a)+O_{p}(h^{m}_{z})$. Evaluating at $a=x$
and at $a=0$ and subtracting gives $B_{n}(x)=T_{n}(x)-T_{n}(0)+O_{p}(h^{m}_{z})$. 

Part (ii). Combining the second equality of Lemma \ref{lem:AL} and
Part (i) yields, uniformly over $\mathcal{X}_{0}$,
\[
\widehat{g}_{\mathrm{LS}}(x)-g(x)=\frac{1}{n}\sum^{n}_{i=1}\left\{ \rho_{x}(W_{i})-\mathbb{E}\rho_{x}(W_{i})\right\} +\mathbb{E}\rho_{x}(W_{i})+T_{n}(x)-T_{n}(0)+O_{p}(R_{n}+\bar{h}^{m}+\sqrt{h_{x}/n}).
\]
Since the class $\left\{ \rho_{x}:x\in\mathcal{X}_{0}\right\} $ is
Donsker by Lemma \ref{lem:donsker}, $n^{-1}\sum^{n}_{i=1}\{\rho_{x}(W_{i})-\mathbb{E}\rho_{x}(W_{i})\}=O_{p}(n^{-1/2})$
uniformly. It is tedious but easy to see that $\mathbb{E}\rho_{x}(W_{i})=0$
by applying the change-of-variable argument, which proves the claim. 

Part (iii). Let $m_{a}(u,z)=\mathbb{E}\left[c_{a}(W_{i})|X_{i}=u,Z_{i}=z\right]$.
Since $\mathbb{E}\left[\boldsymbol{1}\left\{ Y_{i}\le y\right\} |X_{i},Z_{i}\right]=F(y|X_{i},Z_{i})$,
\[
m_{a}(u,z)=\int\frac{\mu(y,z)\left[F(y|a,z)-F(y|u,z)\right]}{F_{y}(y|a,z)f(a,z)}\mathrm{d}y,
\]
which vanishes at $u=a$. Writing $M_{a}(u)=\int m_{a}(u,z)f_{X,Z}(u,z)\mathrm{d}z$,
so that $M_{a}(a)=0$, taking the conditional expectation and changing
variables gives
\begin{align*}
\mathbb{E}\left[T_{n}(a)\right] & =\frac{1}{h_{x}}\mathbb{E}_{XZ}\left[K\left(\frac{X_{i}-a}{h_{x}}\right)m_{a}(X_{i},Z_{i})\right]\\
 & =\frac{1}{h_{x}}\int K\left(\frac{u-a}{h_{x}}\right)M_{a}(u)\mathrm{d}u\\
 & =\int K(v)M_{a}(a+h_{x}v)\mathrm{d}v=O(h^{m}_{x}),
\end{align*}
because the moments of $K$ up to order $m-1$ vanish under Assumption
\ref{assu: kernel} and $M_{a}$ is $m$ times differentiable under
Assumption \ref{assu: density}.

Part (iv). The same change of variable applied to the second moment
gives
\[
\mathbb{E}\left[\frac{1}{h^{2}_{x}}K\left(\frac{X_{i}-a}{h_{x}}\right)^{2}c_{a}(W_{i})^{2}\right]=\frac{1}{h_{x}}\int K(v)^{2}\mathbb{E}\left[c_{a}(W_{i})^{2}|X_{i}=a+h_{x}v\right]f_{X}(a+h_{x}v)\mathrm{d}v=\frac{\sigma^{2}(a)}{h_{x}}+o(h^{-1}_{x}),
\]
and the squared mean is $O(h^{2m}_{x})$, so that $nh_{x}\mathrm{Var}(T_{n}(a))\rightarrow\sigma^{2}(a)$.
It remains to show that the two summands are asymptotically uncorrelated.
The covariance comes from a single observation entering both averages.
Writing $\lambda=x/h_{x}$ and substituting $u=x+h_{x}v$ gives 
\[
nh_{x}\left|\mathrm{Cov}(T_{n}(x),T_{n}(0))\right|\le C\lVert f_{X}\rVert_{\infty}\int\left|K(v)K(v+\lambda)\right|\mathrm{d}v+O(h^{2m+1}_{x}).
\]
Since $\left|v\right|+\left|v+\lambda\right|\ge\lambda$, the two
regions $\left|v+\lambda\right|\ge\lambda/2$ and $\left|v\right|\ge\lambda/2$
cover the line, and on each of them the tail bound of Assumption \ref{assu: kernel}
applies to the factor that is far from its own center:
\begin{align*}
\int\left|K(v)K(v+\lambda)\right|\mathrm{d}v & \le\left(\int_{\left|v+\lambda\right|\ge\lambda/2}+\int_{\left|v\right|\ge\lambda/2}\right)\left|K(v)K(v+\lambda)\right|\mathrm{d}v\\
 & \le C\left(\lambda/2\right)^{-\nu}\left(\int\left|K(v)\right|\mathrm{d}v+\int\left|K(v+\lambda)\right|\mathrm{d}v\right)=2^{\nu+1}C\lambda^{-\nu}\int\left|K\right|.
\end{align*}
Hence $nh_{x}\mathrm{Cov}(T_{n}(x),T_{n}(0))=O\big((h_{x}/\left|x\right|)^{\nu}\big)\rightarrow0$,
uniformly on any set bounded away from the origin. The point $x=0$
is excluded from $\mathcal{X}_{0}$, so the two never coincide.
\end{proof}

\begin{proof}[Proof of Theorem \ref{thm:pointwise}]
 By Lemma \ref{lem:boundary},
\[
\sqrt{nh_{x}}\left(\widehat{g}_{\mathrm{LS}}(x)-g(x)\right)=\sqrt{nh_{x}}\left(T_{n}(x)-T_{n}(0)\right)+O_{p}(\sqrt{h_{x}}),
\]
and the remainder is $o_{p}(1)$ because $h_{x}\to0$. Put
\[
d_{i}=K\left(\frac{X_{i}-x}{h_{x}}\right)c_{x}(W_{i})-K\left(\frac{X_{i}}{h_{x}}\right)c_{0}(W_{i}),\qquad\sqrt{nh_{x}}\left(T_{n}(x)-T_{n}(0)\right)=\frac{1}{\sqrt{nh_{x}}}\sum^{n}_{i=1}d_{i}.
\]
The $d_{i}$ are independent and identically distributed within a
row and depend on the row through $h_{x}$. Lemma \ref{lem:boundary}
gives $\mathbb{E}\left[K((X_{i}-a)/h_{x})c_{a}(W_{i})\right]=O(h^{m+1}_{x})$,
so that
\[
\frac{1}{\sqrt{nh_{x}}}\sum^{n}_{i=1}\mathbb{E}\left[d_{i}\right]=\sqrt{\frac{n}{h_{x}}}O(h^{m+1}_{x})=O\left(\sqrt{nh_{x}}h^{m}_{x}\right)\rightarrow0
\]
by the undersmoothing condition, and it gives $\mathrm{Var}(d_{i})/h_{x}\rightarrow\sigma^{2}(x)+\sigma^{2}(0)$.
Next we check Lyapunov's condition. The change of variable used in
Lemma \ref{lem:boundary} gives
\[
\mathbb{E}\left[\left|K\left(\frac{X_{i}-a}{h_{x}}\right)c_{a}(W_{i})\right|^{3}\right]=h_{x}\left[\int\left|K(v)\right|^{3}\mathrm{d}v\right]f_{X}(a)\mathbb{E}\left[\left|c_{a}(W_{i})\right|^{3}|X_{i}=a\right]+o(h_{x}),
\]
which is $O(h_{x})$ because $\mu$ is bounded with compact support
and $F_{y}$ and $f$ are bounded away from zero on $\mathcal{Y}_{\mu}\times\widetilde{\mathcal{X}}_{0}\times\mathcal{Z}_{\mu}$
under Assumptions \ref{assu: local} and \ref{assu: density}. The
summand $d_{i}$ is a difference of two terms of that form, and centering
costs at most a constant factor, so $\mathbb{E}\left|d_{i}-\mathbb{E}d_{i}\right|^{3}=O(h_{x})$
as well. Hence
\[
\frac{\sum^{n}_{i=1}\mathbb{E}\left|d_{i}-\mathbb{E}d_{i}\right|^{3}}{\left(\sum^{n}_{i=1}\mathrm{Var}(d_{i})\right)^{3/2}}=\frac{nO(h_{x})}{\left(nh_{x}\right)^{3/2}\left\{ \sigma^{2}(x)+\sigma^{2}(0)+o(1)\right\} ^{3/2}}=O\left((nh_{x})^{-1/2}\right)\rightarrow0.
\]
The brace is bounded away from zero by the hypothesis $\sigma^{2}(0)>0$.
The condition holds and the central limit theorem for triangular arrays
applies. Collecting the three displays gives the stated limit. For
$\widehat{g}_{\mathrm{LAD}}$, the equivalence established below bounds
the difference between the two estimators, and what it has to be smaller
than has weakened from $o_{p}(n^{-1/2})$ to $o_{p}((nh_{x})^{-1/2})$.
The argument is unchanged and the same limit follows.

Finally, for the LAD estimator, the argument of \citet[Theorem 2]{CKK15-JoE}
gives $\sup_{x\in\mathcal{X}_{0}}|\hat{g}_{\mathrm{LAD}}(x)-\hat{g}_{\mathrm{LS}}(x)|=O_{p}(\Delta^{\mathrm{LAD}}_{n})$
uniformly over $\mathcal{X}_{0}$, adapted so that each of its steps
holds uniformly in $x$. Here $r_{n}=\lVert\widehat{S}-S\rVert_{\infty}$
and 
\[
\Delta^{\mathrm{LAD}}_{n}:=\frac{\log n}{nh_{x}}+r_{n}\sqrt{\frac{\log n}{nh_{x}}}+r^{2}_{n}.
\]
By Lemma \ref{lem:linear1} and the definition of $R_{n}$, $r_{n}=O_{p}(\sqrt{R_{n}})$.
Multiplied by $\sqrt{nh_{x}\log n}$, the three terms are $(\log n)^{3/2}(nh_{x})^{-1/2}$,
$O_{p}(\sqrt{R_{n}}\log n)$ and $O_{p}(\sqrt{nh_{x}}R_{n}\sqrt{\log n})$.
Assumption \ref{assu:rate} fixes the bandwidths at powers of $n$
and makes $\sqrt{nh_{x}}R_{n}\rightarrow0$, so each of the three
is a power of $n$ times a fixed power of $\log n$, and $\sqrt{nh_{x}}\Delta^{\mathrm{LAD}}_{n}\sqrt{\log n}\rightarrow0$.
\end{proof}

\section{Proof of Proposition \ref{prop:lower_bound}}\label{sec:lower}

Take the subset $\mathcal{P}_{0}\subset\mathcal{P}$ on which $\varepsilon$
is independent of $\left(\eta,Z\right)$. Fix the first-stage function
$h$ and the distributions of $\eta$, $Z$ and $\varepsilon$ at
one choice that satisfies Assumptions \ref{assu:rv}--\ref{assu:support},
with $\varepsilon$ normally distributed, and let $g$ alone vary
over the set $\mathcal{G}$ of functions that are $m$ times differentiable
with derivatives bounded by a constant chosen so that the joint density
obeys the bound fixed in Proposition \ref{prop:lower_bound}, and
that satisfy $g(0)=0$. Every element of $\mathcal{P}_{0}$ lies in
$\mathcal{P}$. The conditions on the first stage and on the supports
do not involve $g$ and hold by the choice just made; $\left(\varepsilon,\eta\right)\indep Z$
follows from $\varepsilon\indep\left(\eta,Z\right)$ and $\eta\indep Z$;
and the conditional density of $Y$ given $\left(X,Z\right)$ is $f_{\varepsilon}(y-g(x))$,
so that the joint density is $f_{\varepsilon}(y-g(x))f_{X|Z}(x|z)f_{Z}(z)$
and Assumption \ref{assu:rv} holds, and the joint density is bounded
and $m$ times differentiable with bounded derivatives, for every
$g\in\mathcal{G}$. The one condition to check is Assumption \ref{assu: local}.
Since $F_{Y|X,Z}(y|x,z)=F_{\varepsilon}(y-g(x))$ does not vary with
$z$ here, $\nabla_{z}F_{Y|X,Z}$ vanishes identically and is continuous;
$\nabla_{y}F_{Y|X,Z}(y|x,z)=f_{\varepsilon}(y-g(x))$ is continuous
because $\varepsilon$ is normal; and the conditions on $V$, which
involve the first stage alone, hold by the choice made above. The
derivative that vanishes enters the identified expression only in
the numerator, and the two quantities required to be nonzero, $\nabla_{z}V(x|z)$
and $f_{Y|X,Z}(y|x,z)$, are unaffected by the restriction. On $\mathcal{P}_{0}$
the error $\varepsilon$ is independent of $X$, so $\mathbb{E}\left[Y|X=u\right]=g(u)+\mathbb{E}\left[\varepsilon\right]$
at every $u$, and the normalization $g(0)=0$ gives $g(x)=\mathbb{E}\left[Y|X=x\right]-\mathbb{E}\left[Y|X=0\right]$.
Estimating $g(x)$ over $\mathcal{P}_{0}$ is therefore estimating
a univariate regression function at a point, over a class of regression
functions that are $m$ times differentiable with derivatives bounded
by a fixed constant, and with a design density that is positive in
a neighborhood of $x$. Perturbations of $g$ supported in a shrinking
neighborhood of $x$ leave $\mathbb{E}\left[Y|X=0\right]$ unchanged,
so the lower bound of \citet{Sto80-AoS} for estimating a regression
function at a point applies to $g(x)$ and gives the rate $n^{-m/(2m+1)}$
in dimension one. The minimax risk over $\mathcal{P}$ is at least
the minimax risk over the subset $\mathcal{P}_{0}$.

\section{Proof of Theorem \ref{thm:band} }\label{sec:proof_band}

Throughout this appendix the conditions of Theorem \ref{thm:band}
are in force. 

\emph{Notation.} Let $P_{n}$ and $P^{*}_{n}$ denote the empirical
measures of $\left\{ W_{i}\right\} ^{n}_{i=1}$ and of the bootstrap
sample $\left\{ W^{*}_{i}\right\} ^{n}_{i=1}$, and write $\mathbb{G}_{n}=\sqrt{n}(P_{n}-P_{0})$,
$\mathbb{G}^{*}_{n}=\sqrt{n}(P^{*}_{n}-P_{n})$. Probability and expectation
conditional on the data are written $P^{*}$ and $\mathbb{E}^{*}$.
Let $\theta$ collect the population quantities estimated by kernel
smoothing in Section \ref{sec:Estimation-and-bootstrap}, grouped
as $\theta=(\theta_{\Phi},\theta_{\Psi},\theta_{f},\theta_{p})$,
where $\theta_{\Phi}$ lists $\Phi$ together with each of its partial
derivatives that appears in Lemmas \ref{lem:linear1} and \ref{lem:linear2},
and $\theta_{\Psi},\theta_{f},\theta_{p}$ likewise. Let $\hat{\theta}$
and $\hat{\theta}^{*}$ be the corresponding estimators from the original
and the bootstrap sample, computed with the same bandwidths. Every
coordinate of $\hat{\theta}$ is a sample average: writing $k^{\alpha}$
for the kernel function attached to coordinate $\alpha$ and a given
evaluation point, $\hat{\theta}^{\alpha}=P_{n}k^{\alpha}$. Collecting
these, $\hat{\theta}=P_{n}k$, where $k$ maps a single observation
to an element of the same space as $\theta$. For instance $\hat{\Phi}(y,x,z)=P_{n}k^{\Phi}_{y,x,z}$
with $k^{\Phi}_{y,x,z}(W)=(h_{x}h_{z})^{-1}\boldsymbol{1}\{Y\le y\}K((X-x)/h_{x})K((Z-z)/h_{z})$,
and analogously for the remaining coordinates. For the coordinates
that are derivatives, $k^{\alpha_{x}}:=\partial_{x}k^{\alpha}$ and
$k^{\alpha_{z}}:=\partial_{z}k^{\alpha}$; this replaces $K$ by $K^{\prime}$
and carries a factor $-h^{-1}$ as well, which is where the extra
power of the bandwidth in their rates comes from. The $y$ coordinate
is not a derivative but a separate average, $k^{\Phi_{y}}_{y,x,z}(W)=(h_{y}h_{x}h_{z})^{-1}K\big(\tfrac{Y-y}{h_{y}}\big)K\big(\tfrac{X-x}{h_{x}}\big)K\big(\tfrac{Z-z}{h_{z}}\big)$.

For a perturbation $\kappa=(\kappa_{\Phi},\kappa_{\Psi},\kappa_{f},\kappa_{p})$
of $\theta$, set 
\[
\mathcal{L}_{x}[\kappa]:=\int\int\mu(y,z)\left\{ \nabla_{\Phi}S_{*}[\kappa_{\Phi}]+\nabla_{\Psi}S_{*}[\kappa_{\Psi}]+\nabla_{f}S_{*}[\kappa_{f}]+\nabla_{p}S_{*}[\kappa_{p}]\right\} (y,x,z)\mathrm{d}y\mathrm{d}z,
\]
 with the operators $\nabla_{\Phi}S_{*},\nabla_{\Psi}S_{*},\nabla_{f}S_{*},\nabla_{p}S_{*}$
as defined in Appendix \ref{sec:thm1}. $\mathcal{L}_{x}$ is linear
in $\kappa$, and since $k(w)$ lies in the same space as $\theta$
it lies in the domain of $\mathcal{L}_{x}$ for every observation
$w$; we write $\mathcal{L}_{x}k$ for the real-valued function $w\mapsto\mathcal{L}_{x}[k(w)]$.
Write $\gamma_{n}=(nh_{x})^{-1/8}(\log n)^{7/8}$, the value at which
the coupling error of the lemmas below and the exceptional probability
balance in Step 4. The couplings are stated at that value. Fix $\upsilon\in(0,1/8)$.
Lemma \ref{lem:vc}(iii) holds for every $q\in[4,\infty)$, and $q$
is taken below large enough that $(nh_{x})^{1/(8q)}$ times any fixed
power of $\log n$ is $O\big((nh_{x})^{\upsilon}\big)$, which is
possible because $nh_{x}$ is a power of $n$. The constants of \citet{CCK16-SPA}
depend on $q$, and so on $\upsilon$.

Write $s^{\alpha}_{n}$ for the stochastic part of the uniform rate
that the proof of Lemma \ref{lem:linear1} gives for coordinate $\alpha$,
so that the full rate there is $\bar{h}^{m}+s^{\alpha}_{n}$; for
instance $s^{\Phi_{y}}_{n}=\sqrt{\log n/(nh_{y}h_{x}h_{z})}$.
\begin{lem}
\label{lem:D_rates} (i) 
\begin{equation}
\hat{\theta}^{*}-\hat{\theta}=(P^{*}_{n}-P_{n})k,\qquad\hat{\theta}-\theta=(P_{n}-P_{0})k+(P_{0}k-\theta),\label{eq:boot_exact}
\end{equation}
 where $P_{0}k-\theta=O(\bar{h}^{m})$ uniformly over $\mathcal{Y}_{\mu}\times\widetilde{\mathcal{X}}_{0}\times\mathcal{Z}_{\mu}$;
in particular $\mathbb{E}^{*}\hat{\theta}^{*}=\hat{\theta}$, so the
bootstrap deviation carries no smoothing bias while the sample deviation
does. (ii) Conditionally on the data and for every coordinate $\alpha$,
the suprema below being over $\mathcal{Y}_{\mu}\times\widetilde{\mathcal{X}}_{0}\times\mathcal{Z}_{\mu}$,
\[
\left\lVert (\hat{\theta}^{*}-\hat{\theta})^{\alpha}\right\rVert _{\infty}=O_{P^{*}}(s^{\alpha}_{n}),\qquad\left\lVert (\hat{\theta}^{*}-\theta)^{\alpha}\right\rVert _{\infty}=O_{P^{*}}(\bar{h}^{m}+s^{\alpha}_{n}).
\]
 (iii) $\hat{f}^{*}$, $\nabla_{y}\hat{F}^{*}=\hat{\Phi}^{*}_{y}/\hat{f}^{*}$
and $\nabla_{z}\widehat{V}^{*}$ are bounded away from zero on $\mathcal{Y}_{\mu}\times\widetilde{\mathcal{X}}_{0}\times\mathcal{Z}_{\mu}$
with conditional probability tending to one.
\end{lem}

\begin{proof}
Part (i) is immediate: $\hat{\theta}=P_{n}k$ and $\hat{\theta}^{*}=P^{*}_{n}k$
involve the same $k$, since the bootstrap estimator uses the same
bandwidths, so the difference is $(P^{*}_{n}-P_{n})k$ and $\mathbb{E}^{*}P^{*}_{n}k=P_{n}k$,
since the $W^{*}_{i}$ are i.i.d. from $P_{n}$ given the data. The
second identity is the same algebra with $P_{0}$ in place of $P_{n}$,
and $P_{0}k-\theta$ is the usual kernel smoothing bias, of order
$\bar{h}^{m}$ under Assumptions \ref{assu: kernel} and \ref{assu: density};
it is the bias term that appears in the rates in the proof of Lemma
\ref{lem:linear1}.

For (ii), conditional on the data $\left\{ W^{*}_{i}\right\} ^{n}_{i=1}$
are i.i.d. from $P_{n}$, and by (i) $\hat{\theta}^{*}-\hat{\theta}$
is a centered average, so there is no bias term to control. Write
$U_{\alpha}$ for the envelope of $k^{\alpha}$ and $H_{\alpha}$
for the product of the bandwidths that coordinate smooths, so that
$H_{\alpha}\in\left\{ h_{z},h_{x}h_{z},h_{y}h_{x}h_{z}\right\} $.
Then $P_{0}(k^{\alpha})^{2}\asymp U^{2}_{\alpha}H_{\alpha}$ and $P_{0}(k^{\alpha})^{4}\asymp U^{4}_{\alpha}H_{\alpha}$,
since $k^{\alpha}$ is of size $U_{\alpha}$ on a set of probability
of order $H_{\alpha}$ and negligible elsewhere. The families $\left\{ k^{\alpha}\right\} $,
indexed by the evaluation point, are VC type (in the sense of (\ref{eq:VC-dim}))
with characteristics that do not depend on the bandwidth, by the argument
in the proof of Lemma \ref{lem:linear1}. Their entropy bound does
not depend on the underlying measure, so the exponential inequality
of \citet[Corollary 2.2]{GineGuillou02} applies to draws from $P_{n}$
and gives
\[
\left\lVert (P^{*}_{n}-P_{n})k^{\alpha}\right\rVert _{\infty}=O_{P^{*}}\left(\sqrt{\frac{P_{n}(k^{\alpha})^{2}\log n}{n}}\right).
\]
The variance factor of this bound, $P_{n}(k^{\alpha})^{2}$, depends
on the data, and the rate the lemma states does not. It is enough
that $P_{n}(k^{\alpha})^{2}$ stay within a constant factor of $P_{0}(k^{\alpha})^{2}$
uniformly in $\alpha$, that is, $\sup_{\alpha}\left|(P_{n}-P_{0})(k^{\alpha})^{2}\right|=o_{p}\big(P_{0}(k^{\alpha})^{2}\big)$.
The family $\{(k^{\alpha})^{2}\}$ is VC type for the same reason
$\{k^{\alpha}\}$ is, with envelope $U^{2}_{\alpha}$, so a second
application of \citet[Theorem 2.1]{GineGuillou02} gives
\[
\frac{\sup_{\alpha}\left|(P_{n}-P_{0})(k^{\alpha})^{2}\right|}{P_{0}(k^{\alpha})^{2}}=O_{p}\left(\sqrt{\frac{\log n}{nH_{\alpha}}}+\frac{\log n}{nH_{\alpha}}\right).
\]
The envelope cancels, and $h_{y}h_{x}h_{z}$ is the smallest of the
three values $H_{\alpha}$ takes, so both terms are powers of $\log n/(nh_{y}h_{x}h_{z})$,
which vanishes because Assumption \ref{assu:rate} makes even $\sqrt{nh_{x}}\log n/(nh_{y}h_{x}h_{z})$
vanish. With that, the bootstrap bound is $\sqrt{P_{0}(k^{\alpha})^{2}\log n/n}$,
which is $s^{\alpha}_{n}$. This gives the first display of the lemma;
the second follows by adding the rate for $\hat{\theta}-\theta$ from
the proof of Lemma \ref{lem:linear1}. Part (iii) then follows from
the lower bound on the density in Assumption \ref{assu: density}.
\end{proof}

For $x\in\mathcal{X}_{0}$, let
\[
f_{n,x}(W)=h^{-1/2}_{x}\left[K\left(\frac{X-x}{h_{x}}\right)c_{x}(W)-K\left(\frac{X}{h_{x}}\right)c_{0}(W)\right],\qquad\mathcal{F}_{n}=\left\{ f_{n,x}:x\in\mathcal{X}_{0}\right\} ,
\]
with $c_{a}$ as in Lemma \ref{lem:boundary}. The next lemma is needed
to confirm that the class of functions $\mathcal{F}_{n}$ meets the
conditions of \citet{CCK16-SPA}.
\begin{lem}
\label{lem:vc} There is a constant $C\ge1$ such that the following
hold. (i) 
\[
\frac{1}{\sqrt{n}}\sum^{n}_{i=1}\left\{ f_{n,x}(W_{i})-\mathbb{E}f_{n,x}(W_{i})\right\} =\sqrt{nh_{x}}\left\{ T_{n}(x)-T_{n}(0)-\mathbb{E}\left[T_{n}(x)-T_{n}(0)\right]\right\} .
\]
(ii) $\mathcal{F}_{n}$ is pointwise measurable (see \citet[Example 2.3.4]{vdvW1996})
and is VC type with envelope
\[
F_{n}(W)=2\lVert K\rVert_{\infty}h^{-1/2}_{x}\bar{c}(W),\qquad\bar{c}(W)=\sup_{a\in\widetilde{\mathcal{X}}_{0}}\left|c_{a}(W)\right|,
\]
and constants $A\ge e$ and $v\ge1$ that do not depend on $n$. (iii)
With $\varsigma=C$ and $b_{n}=Ch^{-1/2}_{x}$,
\[
\sup_{f\in\mathcal{F}_{n}}\mathbb{E}\left|f(W)\right|^{k}\le\varsigma^{2}b^{k-2}_{n}\quad(k=2,3,4),\qquad\lVert F_{n}\rVert_{L_{q}(P_{0})}\le b_{n}\ \text{for all }q\in[4,\infty).
\]
\end{lem}

\begin{proof}
Part (i) is the definition of $T_{n}$ rearranged. For Part (ii),
the envelope follows from $\left|K\right|\le\lVert K\rVert_{\infty}$.
The function $\bar{c}$ is bounded because $\mu$ is bounded with
compact support and $F_{y}$ and $f$ are bounded away from zero on
$\mathcal{Y}_{\mu}\times\widetilde{\mathcal{X}}_{0}\times\mathcal{Z}_{\mu}$
under Assumptions \ref{assu: local} and \ref{assu: density}, so
the only dependence on $n$ left in $F_{n}$ is the factor $h^{-1/2}_{x}$.
Pointwise measurability holds because $x\mapsto f_{n,x}(W)$ is continuous
for every $W$, the continuity in $x$ of $c_{x}$ being the Lipschitz
bound established below. A countable dense subset of $\mathcal{X}_{0}$
indexes a countable subclass with the required property. For the entropy
in (ii), write $f_{n,x}=h^{-1/2}_{x}\left(\kappa_{x}c_{x}-\kappa_{0}c_{0}\right)$
with $\kappa_{a}(X)=K((X-a)/h_{x})$. The class is built from two
families by a product, a subtraction and a scaling. Assumption \ref{assu: kernel}
makes $K$ of bounded variation. The translates and dilates of such
a function, $\left\{ K((\cdot-x)/h):x\in\mathbb{R},h>0\right\} $,
are a VC type class \citep{NolanPollard87,GineGuillou02}. The bandwidth
is one of the indices there, so the characteristics of $\left\{ \kappa_{x}:x\in\mathcal{X}_{0}\right\} $
do not depend on it. The family $\left\{ c_{x}:x\in\mathcal{X}_{0}\right\} $
is Lipschitz in $x$ uniformly in $W$: differentiating under the
integral sign, 
\[
\partial_{x}c_{x}(W)=\int\mu(y,Z)\,\partial_{x}\left\{ \frac{F(y|x,Z)-\boldsymbol{1}\left\{ Y\le y\right\} }{F_{y}(y|x,Z)f(x,Z)}\right\} \mathrm{d}y,
\]
which Assumptions \ref{assu: local} and \ref{assu: density} bound.
The bound is uniform in $W$, so $\sup_{W}\left|c_{x}(W)-c_{x^{\prime}}(W)\right|\le C\left|x-x^{\prime}\right|$.
Every $c_{x}$ is then within $\varepsilon$ of some $c_{x^{\prime}}$
in the supremum norm, $x^{\prime}$ running over a net of $\mathcal{X}_{0}$
of mesh $\varepsilon/C$. Since $\mathcal{X}_{0}$ is bounded, such
a net has $O(1/\varepsilon)$ points, and the covering numbers of
the family are of that order. That family is therefore VC type as
well. The class $\left\{ \kappa_{x}c_{x}\right\} $ is the pointwise
product of two VC type classes, so it is VC type, with the envelopes
multiplied and the characteristics added \citep[Section 2.6]{vdvW1996}.
Subtracting the single function $\kappa_{0}c_{0}$ changes the characteristics
by a constant. Multiplying by $h^{-1/2}_{x}$ scales the class and
its envelope alike. The radius in (\ref{eq:VC-dim}) is measured against
the envelope, so that scaling leaves the bound unchanged. Each step
above leaves the characteristics free of $n$, so $\mathcal{F}_{n}$
is VC type with constants $A$ and $v$ that do not depend on $n$. 

For Part (iii), the second moment is the variance computation in the
proof of Lemma \ref{lem:boundary}, which gives $\mathbb{E}f_{n,x}(W)^{2}\rightarrow\sigma^{2}(x)+\sigma^{2}(0)$
uniformly, hence a bound free of $n$. The same change of variable
applied to the higher moments give
\[
\mathbb{E}\left|f_{n,x}(W)\right|^{k}=\left(h^{-1/2}_{x}\right)^{k}O(h_{x})=O\left(h^{1-k/2}_{x}\right)\qquad(k=3,4),
\]
because the set on which a member of $\mathcal{F}_{n}$ is not negligible
has probability $O(h_{x})$. The bound on $\lVert F_{n}\rVert_{L_{q}(P_{0})}$
follows for every such $q$ from the boundedness of $\bar{c}$. The
constants in these bounds are free of $n$, so $C$ may be taken large
enough for all the inequalities of (iii) to hold at once.
\end{proof}

The two scales in Lemma \ref{lem:vc} (iii) are far apart: $\varsigma$
bounds the second moments and is of order one, while $b_{n}$ bounds
the envelope and is of order $h^{-1/2}_{x}$. A member of $\mathcal{F}_{n}$
reaches that size only on a set of probability of order $h_{x}$,
so its second moment stays of order one while its supremum does not.
A class whose envelope dominates its members in this way is not Donsker.
\begin{lem}
\label{lem:coupling} Let $\sigma^{2}_{n}(x)=\mathrm{Var}_{P}(f_{n,x})$,
$\overline{\mathcal{F}}_{n}=\left\{ f_{n,x}/\sigma_{n}(x):x\in\mathcal{X}_{0}\right\} $,
let $\mathbb{G}_{P}$ be the $P$-Brownian bridge indexed by $\overline{\mathcal{F}}_{n}$,
the centered Gaussian process with covariance $\mathbb{E}\left[\mathbb{G}_{P}(f)\mathbb{G}_{P}(g)\right]=P(fg)-Pf\cdot Pg$,
and let 
\[
Z_{n}=\sup_{x\in\mathcal{X}_{0}}\frac{\sqrt{nh_{x}}\left|\widehat{g}_{\mathrm{LS}}(x)-g(x)\right|}{\sigma_{n}(x)}.
\]
Then (i) $\overline{\mathcal{F}}_{n}\cup(-\overline{\mathcal{F}}_{n})$
satisfies conditions (A) to (C) of \citet{CCK16-SPA}. (ii) For each
$n$, there is a random variable $\widetilde{Z}_{n}\overset{d}{=}\sup_{x\in\mathcal{X}_{0}}\left|\mathbb{G}_{P}(f_{n,x}/\sigma_{n}(x))\right|$,
such that
\begin{align*}
\left|Z_{n}-\widetilde{Z}_{n}\right| & \le C\left\{ \sqrt{h_{x}}+\sqrt{nh_{x}}R_{n}+\sqrt{nh_{x}}\bar{h}^{m}+(nh_{x})^{-1/8}(\log n)^{3/8}\right\} \\
 & \quad\text{with probability }1-C^{\prime}(\gamma_{n}+n^{-1})
\end{align*}
\end{lem}

\begin{proof}
Part (i). Lemma \ref{lem:vc} gives conditions (A) to (C) of \citet{CCK16-SPA}
for $\mathcal{F}_{n}$. We can extend the claim to $\overline{\mathcal{F}}_{n}\cup(-\overline{\mathcal{F}}_{n})$
as well. To see this, note that dividing $f_{n,x}$ by $\sigma_{n}(x)$
changes the envelope and the constants of Lemma \ref{lem:vc} by a
factor and nothing else. The factor is bounded above and away from
zero: Lemma \ref{lem:vc}(iii) bounds $\sigma^{2}_{n}(x)$ above,
and the variance computation in the proof of Lemma \ref{lem:boundary}
makes it converge to $\bar{\sigma}^{2}(x)\ge\sigma^{2}(0)>0$ uniformly.
Adjoining the negatives of a class doubles its covering numbers at
most. Therefore $\overline{\mathcal{F}}_{n}\cup(-\overline{\mathcal{F}}_{n})$
also satisfies conditions (A) to (C) of \citet{CCK16-SPA}. 

Part (ii). First reduce $Z_{n}$ to the empirical process. By Lemma
\ref{lem:boundary} the difference between $\sqrt{nh_{x}}\left(\widehat{g}_{\mathrm{LS}}(x)-g(x)\right)$
and $\sqrt{nh_{x}}\left(T_{n}(x)-T_{n}(0)\right)$ is 
\[
\sqrt{nh_{x}}\{O_{p}(n^{-1/2})+O_{p}(R_{n}+\bar{h}^{m})\}=O_{p}(\sqrt{h_{x}})+O_{p}(\sqrt{nh_{x}}R_{n})+O(\sqrt{nh_{x}}\bar{h}^{m})
\]
 uniformly over $\mathcal{X}_{0}$. The empirical process subtracts
a mean, and Lemma \ref{lem:boundary}(iii) makes it $\sqrt{nh_{x}}\mathbb{E}\left[T_{n}(x)-T_{n}(0)\right]=O(\sqrt{nh_{x}}h^{m}_{x})$,
which vanishes by the undersmoothing condition. Dividing by $\sigma_{n}(x)$,
Lemma \ref{lem:vc}(i) gives
\[
Z_{n}=\sup_{x\in\mathcal{X}_{0}}\left|\mathbb{G}_{n}\left(f_{n,x}/\sigma_{n}(x)\right)\right|+o_{p}(1),\qquad\mathbb{G}_{n}f=\frac{1}{\sqrt{n}}\sum^{n}_{i=1}\left\{ f(W_{i})-\mathbb{E}f(W_{i})\right\} .
\]
 Apply \citet[Theorem 2.1]{CCK16-SPA} with $B\equiv0$. It gives
a random variable $\widetilde{Z}_{n}$ with the stated distribution
and, for every $\gamma\in(0,1)$,
\[
P\left\{ \left|Z_{n}-\widetilde{Z}_{n}\right|>C\delta_{n}\right\} \le C^{\prime}(\gamma+n^{-1}),\quad\text{where }\delta_{n}=\frac{b_{n}K_{n}}{\gamma^{1/q}n^{1/2-1/q}}+\frac{\left(b_{n}\varsigma^{2}K^{2}_{n}\right)^{1/3}}{\gamma^{1/3}n^{1/6}},
\]
with $K_{n}=v\left(\log n\vee\log(Ab_{n}/\varsigma)\right)$. Here
$b_{n}/\varsigma$ is of order $h^{-1/2}_{x}$ by Lemma \ref{lem:vc}(iii),
and $\log(1/h_{x})=O(\log n)$ because $h_{x}$ is a power of $n$
by Assumption \ref{assu:rate}, so $K_{n}=O(\log n)$ and $K^{3}_{n}\le n$
for all large $n$. Substituting $\varsigma\asymp1$ and $b_{n}\asymp h^{-1/2}_{x}$
turns each term into a power of $(nh_{x})^{-1}$:
\[
\delta_{n}\asymp\frac{n^{1/q}\log n}{\gamma^{1/q}(nh_{x})^{1/2}}+\frac{(\log n)^{2/3}}{\gamma^{1/3}(nh_{x})^{1/6}}.
\]
The exponent $q$ is free: the envelope is bounded, so Lemma \ref{lem:vc}(iii)
holds for every $q\in[4,\infty)$. Taking $q$ large enough makes
the first summand negligible beside the second: the two differ by
$(nh_{x})^{1/3}$ up to logarithms, and $nh_{x}$ is a power of $n$,
so $n^{1/q}$ can be made the smaller. Hence $\delta_{n}\asymp\gamma^{-1/3}(nh_{x})^{-1/6}(\log n)^{2/3}$,
which at $\gamma=\gamma_{n}$ is $(nh_{x})^{-1/8}(\log n)^{3/8}$.
That is the last term of (ii). The first three, $\sqrt{h_{x}}$, $\sqrt{nh_{x}}R_{n}$
and $\sqrt{nh_{x}}\bar{h}^{m}$, are the error left by the reduction
to the empirical process. This proves (ii). Finally $h_{x}=cn^{-a}$
with $a<1$ under Assumption \ref{assu:rate}, so $nh_{x}$ grows
as a power of $n$ and the bound is $o_{p}(1)$.
\end{proof}

\begin{lem}
\label{lem:boot_coupling} Let 
\[
Z^{*}_{n}=\sup_{x\in\mathcal{X}_{0}}\frac{\sqrt{nh_{x}}\left|\widehat{g}^{*}_{\mathrm{LS}}(x)-\widehat{g}_{\mathrm{LS}}(x)\right|}{\sigma_{n}(x)}.
\]
(i) Uniformly over $\mathcal{X}_{0}$,
\[
\sqrt{nh_{x}}\left(\widehat{g}^{*}_{\mathrm{LS}}(x)-\widehat{g}_{\mathrm{LS}}(x)\right)=\mathbb{G}^{*}_{n}f_{n,x}+O_{P^{*}}\big(\sqrt{h_{x}}+\sqrt{nh_{x}}R_{n}\big),\qquad\mathbb{G}^{*}_{n}=\sqrt{n}\left(P^{*}_{n}-P_{n}\right).
\]
(ii) There is a random variable $\widetilde{Z}^{*}_{n}$ equal in
conditional distribution to $\sup_{x\in\mathcal{X}_{0}}\left|\mathbb{G}_{P}(f_{n,x}/\sigma_{n}(x))\right|$
such that
\begin{align*}
\left|Z^{*}_{n}-\widetilde{Z}^{*}_{n}\right| & \le C\left\{ \sqrt{h_{x}}+\sqrt{nh_{x}}R_{n}+(nh_{x})^{-1/8+\upsilon}\right\} \\
 & \quad\text{with probability }1-C^{\prime}(\gamma_{n}+n^{-1}).
\end{align*}
\end{lem}

\begin{proof}
Part (i). As in the proof of Lemma \ref{lem:coupling}, the reduction
repeats the sample-side one with the bootstrap measure in place of
the empirical measure. Lemmas \ref{lem:linear1} and \ref{lem:linear2}
use only the algebraic identity for a ratio together with a uniform
rate for the kernel estimators and the lower bounds on $f$ and $F_{y}$,
all of which Lemma \ref{lem:D_rates} supplies conditionally at the
same rate, so the same expansion holds for $\widehat{S}^{*}-S$. Subtracting
the expansion of $\widehat{S}-S$ and integrating against $\mu$ gives
\[
\widehat{g}^{*}_{\mathrm{LS}}(x)-\widehat{g}_{\mathrm{LS}}(x)=\mathcal{L}_{x}\left[\left(P^{*}_{n}-P_{n}\right)k\right]+O_{P^{*}}(R_{n})=n^{-1/2}\mathbb{G}^{*}_{n}\left[\mathcal{L}_{x}k\right]+O_{P^{*}}(R_{n})
\]
uniformly over $\mathcal{X}_{0}$, the exchange of integration being
permitted because the integrands are bounded on $\mathcal{Y}_{\mu}\times\widetilde{\mathcal{X}}_{0}\times\mathcal{Z}_{\mu}$.
The proofs of Lemmas \ref{lem:AL} and \ref{lem:boundary} split a
sample average into twelve blocks (which are indexed by $\alpha$),
one for each smoothed quantity that enters it, of which two leave
the boundary terms and the other ten an average of $\rho_{x}$. Since
$\mathcal{L}_{x}$ is linear the same split applies to $\mathcal{L}_{x}k$
as a function of one observation, and gives
\[
\mathcal{L}_{x}k=\rho_{x}+h^{-1/2}_{x}f_{n,x}+r,\qquad\lVert r\rVert_{L_{2}(P_{0})}=o(1),
\]
where the middle term collects the two boundaries and $\rho_{x}$
asymptotically linear. Multiplying the previous display by $\sqrt{nh_{x}}$
therefore gives
\[
\sqrt{nh_{x}}\left(\widehat{g}^{*}_{\mathrm{LS}}(x)-\widehat{g}_{\mathrm{LS}}(x)\right)=\sqrt{h_{x}}\,\mathbb{G}^{*}_{n}\rho_{x}+\mathbb{G}^{*}_{n}f_{n,x}+O_{P^{*}}(\sqrt{nh_{x}}R_{n}).
\]
The first term is $O_{P^{*}}(\sqrt{h_{x}})$ because $\left\{ \rho_{x}\right\} $
is Donsker by Lemma \ref{lem:donsker}, so its bootstrap process is
bounded in conditional probability uniformly in $x$. 

Part (ii). Apply \citet[Theorem 2.3]{CCK16-SPA} to $\overline{\mathcal{F}}_{n}\cup(-\overline{\mathcal{F}}_{n})$,
which satisfies its assumptions (A) to (C) by Lemma \ref{lem:coupling}(i).
It gives a random variable $\widetilde{Z}^{*}_{n}$ with the stated
conditional distribution and an error term $\delta^{*}_{n}$. Here,
$\delta^{*}_{n}$ carries one more summand than the sample-side $\delta_{n}$,
\[
\frac{\left(b_{n}\varsigma K^{3/2}_{n}\right)^{1/2}}{\gamma^{1+1/q}n^{1/4}}\asymp\frac{(\log n)^{3/4}}{\gamma^{1+1/q}(nh_{x})^{1/4}},
\]
and also raises the power of $\gamma$ in the first summand from $\gamma^{1/q}$
to $\gamma^{1+1/q}$. A larger power of $\gamma$ in a denominator
inflates a term, so the two have to be compared at $\gamma_{n}$.
There the extra summand is $(nh_{x})^{-1/8+1/(8q)}(\log n)^{-1/8-7/(8q)}$
and the second is $(nh_{x})^{-1/8}(\log n)^{3/8}$, so the extra one
is the larger: $(nh_{x})^{1/(8q)}$ is a power of $n$ and beats any
power of $\log n$. It is at most $(nh_{x})^{-1/8+\upsilon}$ by the
choice of $q$, and the first summand is smaller still for $q$ large,
so $\delta^{*}_{n}=O\big((nh_{x})^{-1/8+\upsilon}\big)$. With the
two terms Part (i) leaves, this gives (ii).
\end{proof}

\begin{lem}
\label{lem:plugin} Let $\hat{c}_{a}$ be $c_{a}$ with $F$, $F_{y}$
and $f$ replaced by $\widehat{F}$, $\nabla_{y}\widehat{F}$ and
$\hat{f}$, and let $\hat{s}_{n}(x)^{2}=\hat{\sigma}^{2}_{\mathrm{pl}}(x)+\hat{\sigma}^{2}_{\mathrm{pl}}(0)$
with $\hat{\sigma}^{2}_{\mathrm{pl}}$ as in (\ref{eq:plugin}). Then
\[
\sup_{x\in\mathcal{X}_{0}}\left|\frac{\hat{s}_{n}(x)^{2}}{\bar{\sigma}(x)^{2}}-1\right|=O_{p}\left(\sqrt{\frac{\log n}{nh_{x}}}+s_{n}+\bar{h}^{m}+h_{x}\right),
\]
 where $s_{n}=\max_{\alpha}s^{\alpha}_{n}$ is the largest of the
uniform rates of Lemma \ref{lem:D_rates}, and the same bound holds
conditionally on the data for the estimator computed from a bootstrap
sample. Both are $o_{p}(1/\log n)$ under Assumption \ref{assu:rate}.
\end{lem}

\begin{proof}
Write $q_{n,a}(W)=h^{-1}_{x}K((X-a)/h_{x})^{2}c_{a}(W)^{2}$ and $\hat{q}_{n,a}$
for the same expression with $\hat{c}_{a}$ in place of $c_{a}$,
so that $\hat{\sigma}^{2}_{\mathrm{pl}}(a)=P_{n}\hat{q}_{n,a}$ and
\[
\hat{\sigma}^{2}_{\mathrm{pl}}(a)-\sigma^{2}(a)=(P_{n}-P_{0})q_{n,a}+\left(P_{0}q_{n,a}-\sigma^{2}(a)\right)+P_{n}(\hat{q}_{n,a}-q_{n,a}).
\]

The second term is deterministic. The change of variable in the proof
of Lemma \ref{lem:boundary} gives $P_{0}q_{n,a}=\int K(v)^{2}\mathbb{E}[c^{2}_{a}(W)|X=a+h_{x}v]f_{X}(a+h_{x}v)\mathrm{d}v$,
which equals $\sigma^{2}(a)+O(h_{x})$ uniformly in $a$. The remainder
is $O(h_{x})$ and not $O(h^{m}_{x})$ because $K^{2}$ is nonnegative
and so has no vanishing moments, whatever the order of $K$.

For the first term, the family $\left\{ q_{n,a}:a\in\mathcal{X}_{0}\right\} $
is VC type by the argument of Lemma \ref{lem:vc} (ii) with $K$ replaced
by $K^{2}$ and $c_{a}$ by $c^{2}_{a}$: $K^{2}$ is of bounded variation
because $K$ is bounded and of bounded variation, and $a\mapsto c^{2}_{a}(W)$
is Lipschitz uniformly in $W$ because $a\mapsto c_{a}(W)$ is and
both are bounded. Its envelope is of order $h^{-1}_{x}$ and $P_{0}q^{2}_{n,a}=h^{-2}_{x}O(h_{x})=O(h^{-1}_{x})$,
so the exponential inequality of \citet[Corollary 2.2]{GineGuillou02}
gives $\sup_{a}\left|(P_{n}-P_{0})q_{n,a}\right|=O_{p}(\sqrt{\log n/(nh_{x})})$.

For the third term, $\hat{q}_{n,a}-q_{n,a}=h^{-1}_{x}K((X-a)/h_{x})^{2}\big(\hat{c}^{2}_{a}-c^{2}_{a}\big)$
and $|\hat{c}^{2}_{a}(W)-c^{2}_{a}(W)|\le(|\hat{c}_{a}(W)|+|c_{a}(W)|)|\hat{c}_{a}(W)-c_{a}(W)|$
at every $a$ and $W$, so
\[
\sup_{a}\big|P_{n}(\hat{q}_{n,a}-q_{n,a})\big|\le\sup_{a,W}\big(|\hat{c}_{a}|+|c_{a}|\big)\;\sup_{a,W}\big|\hat{c}_{a}-c_{a}\big|\;\sup_{a}P_{n}\Big[h^{-1}_{x}K\big((X_{i}-a)/h_{x}\big)^{2}\Big].
\]
The first factor is of order one. By Assumption \ref{assu: density}
(iii) and the uniform rates in the proof of Lemma \ref{lem:linear1},
$\hat{f}$ and $\nabla_{y}\widehat{F}$ are bounded away from zero
on $\mathcal{Y}_{\mu}\times\widetilde{\mathcal{X}}_{0}\times\mathcal{Z}_{\mu}$
with probability tending to one, so $\hat{c}_{a}$ and $c_{a}$ are
bounded there. The second factor is $O_{p}(\bar{h}^{m}+s_{n})$. The
integrand of $\hat{c}_{a}$ differs from that of $c_{a}$ by $O(\lVert\hat{\theta}-\theta\rVert_{\infty})$
uniformly in $a$, and integrating against $\mu$, which is bounded
with compact support, gives $\sup_{a,W}|\hat{c}_{a}-c_{a}|=O_{p}(\bar{h}^{m}+s_{n})$.
The third factor is $O_{p}(1)$. The family $\left\{ h^{-1}_{x}K((X-a)/h_{x})^{2}:a\in\widetilde{\mathcal{X}}_{0}\right\} $
is VC type and obeys the exponential inequality by the argument used
above for $\left\{ q_{n,a}\right\} $, and $P_{0}[h^{-1}_{x}K((X-a)/h_{x})^{2}]=O(1)$
by the change of variable. The third term is therefore $O_{p}(\bar{h}^{m}+s_{n})$.
The three terms are $O_{p}\big(\sqrt{\log n/(nh_{x})}\big)$, $O(h_{x})$
and $O_{p}(\bar{h}^{m}+s_{n})$, so
\[
\sup_{a\in\widetilde{\mathcal{X}}_{0}}\big|\hat{\sigma}^{2}_{\mathrm{pl}}(a)-\sigma^{2}(a)\big|=O_{p}\left(\sqrt{\frac{\log n}{nh_{x}}}+s_{n}+\bar{h}^{m}+h_{x}\right),
\]
which is the rate of the statement. Since $\hat{s}_{n}(x)^{2}-\bar{\sigma}(x)^{2}$
is the sum of that difference at $a=x$ and at $a=0$, and $\bar{\sigma}(x)^{2}\ge\sigma^{2}(0)>0$,
dividing by $\bar{\sigma}(x)^{2}$ gives the display of the lemma. 

Conditionally the same three steps apply. Lemma \ref{lem:D_rates}
(ii) supplies the rates for $\hat{\theta}^{*}-\theta$ and (iii) the
lower bounds on $\hat{f}^{*}$ and $\nabla_{y}\widehat{F}^{*}$, both
with conditional probability tending to one. The class $\left\{ q_{n,a}\right\} $
does not change, and \citet{GineGuillou02} applies to draws from
$P_{n}$ because its entropy bound does not depend on the underlying
measure, which is the step already used in the proof of Lemma \ref{lem:D_rates}.

For the last claim, $h_{x}$ is a power of $n$, so $\sqrt{\log n/(nh_{x})}\log n\rightarrow0$.
The variance conditions of Assumption \ref{assu:rate} state that
$\sqrt{n}s^{2}_{n}\rightarrow0$, so $s_{n}=o(n^{-1/4})$ and $s_{n}\log n\rightarrow0$;
and $\bar{h}^{m}\log n\rightarrow0$ for the same reason. Taking square
roots halves each rate and leaves the conclusion unchanged.
\end{proof}

Let
\[
\Delta_{n}=\sup_{x\in\mathcal{X}_{0}}\left|\hat{s}_{n}(x)/\bar{\sigma}(x)-1\right|,\quad\Delta^{*}_{n}=\sup_{x\in\mathcal{X}_{0}}\left|\hat{s}^{*}_{n}(x)/\bar{\sigma}(x)-1\right|.
\]
By Assumption \ref{assu: band} (i) and (ii), we have $\Delta_{n}=o_{P}(1/\log n)$
and $\Delta^{*}_{n}=o_{P^{*}}(1/\log n)$.
\begin{proof}[Proof of Theorem \ref{thm:band}]
 Let
\[
\check{Z}_{n}=\sup_{x\in\mathcal{X}_{0}}\frac{\sqrt{nh_{x}}\left|\widehat{g}_{\mathrm{LS}}(x)-g(x)\right|}{\hat{s}_{n}(x)},
\]
 and let $\check{Z}^{*}_{n}$ be the same with $\widehat{g}^{*}_{\mathrm{LS}}-\widehat{g}_{\mathrm{LS}}$
in place of $\widehat{g}_{\mathrm{LS}}-g$ and $\hat{s}^{*}_{n}$
in place of $\hat{s}_{n}$. Write $Z_{n}$ and $Z^{*}_{n}$ for the
quantities of Lemmas \ref{lem:coupling} and \ref{lem:boot_coupling}.
Because $\hat{s}_{n}$ is positive, $\left\{ g(x)\in\mathcal{C}_{n}(x)\text{ for all }x\in\mathcal{X}_{0}\right\} =\left\{ \check{Z}_{n}\le c_{n}(1-\alpha)\right\} $,
so it is enough to bound $\left|P(\check{Z}_{n}\le c_{n}(1-\alpha))-(1-\alpha)\right|$. 

The proof has six steps. The first couples the supremum of the studentized
deviation to the supremum of a Gaussian process indexed by the same
class at the same $n$, and does the same on the bootstrap side, so
that the two are compared through a common intermediary. The second
establishes what is needed about that Gaussian supremum: its variance
is one by construction, its size is $O(\sqrt{\log n})$, and the normalization
it uses differs from the scale $\bar{\sigma}$ that Assumption \ref{assu: band}
is stated against by $O(h_{x})$. The third replaces that normalization
by the studentizers of Assumption \ref{assu: band}, at the price
of the two rates that assumption imposes. The fourth turns closeness
of the two suprema into closeness of their distribution functions,
by anti-concentration. The fifth transfers that to the bootstrap quantile,
which depends on the data, and so bounds the coverage error. The sixth
checks that the bound tends to zero. 

\emph{Step 1.} Lemmas \ref{lem:coupling} and \ref{lem:boot_coupling}
give random variables $\widetilde{Z}_{n}$ and $\widetilde{Z}^{*}_{n}$,
each distributed as $\sup_{x\in\mathcal{X}_{0}}\left|\mathbb{G}_{P}(f_{n,x}/\sigma_{n}(x))\right|$,
the second conditionally on the data, such that $\left|Z_{n}-\widetilde{Z}_{n}\right|$
and $\left|Z^{*}_{n}-\widetilde{Z}^{*}_{n}\right|$ are both at most
\begin{equation}
C\left\{ \sqrt{h_{x}}+\sqrt{nh_{x}}R_{n}+\sqrt{nh_{x}}\bar{h}^{m}+(nh_{x})^{-1/8+\upsilon}\right\} \label{eq:thm3.8_step1_bound}
\end{equation}
with probability $1-C^{\prime}(\gamma_{n}+n^{-1})$. The Gaussian
process is the same on the sample side and on the bootstrap side:
it is indexed by the same class $\overline{\mathcal{F}}_{n}$ at the
same $n$ and has the same covariance $P(fg)-Pf\cdot Pg$. So $\widetilde{Z}_{n}$
and $\widetilde{Z}^{*}_{n}$ have the same law, that of $\sup_{x\in\mathcal{X}_{0}}\left|\mathbb{G}_{P}(f_{n,x}/\sigma_{n}(x))\right|$.
Each side has its own such variable, and Step 4 measures both sides
against that law.

\emph{Step 2.} By Step 1, $\widetilde{Z}_{n}$ and $\widetilde{Z}^{*}_{n}$
have the law of $\sup_{x\in\mathcal{X}_{0}}\left|\mathbb{G}_{P}(f_{n,x}/\sigma_{n}(x))\right|$,
and $\mathbb{G}_{P}(f_{n,x}/\sigma_{n}(x))$ is a centered Gaussian
process with variance one at every $x$, since $\sigma^{2}_{n}(x)=\mathrm{Var}_{P}(f_{n,x})$.
Assumptions (B) and (C) make $\overline{\mathcal{F}}_{n}$ totally
bounded for the intrinsic semimetric and give $\mathbb{G}_{P}$ a
version with uniformly continuous sample paths \citep{CCK16-SPA},
and that version is separable. Dudley’s maximal inequality for Gaussian
processes \citep[Corollary 2.2.8]{vdvW1996}, with the characteristics
of Lemma \ref{lem:vc}, $\varsigma\asymp1$ and $b_{n}\asymp h^{-1/2}_{x}$,
gives $\mathbb{E}[\widetilde{Z}_{n}]\le C\sqrt{v\log(Ab_{n}/\varsigma)}=O(\sqrt{\log n})$,
since $\log(1/h_{x})=O(\log n)$; in particular $\widetilde{Z}_{n}<\infty$
almost surely, and the Borell--Sudakov--Tsirel’son inequality \citep[Proposition A.2.1]{vdvW1996}
gives $\widetilde{Z}_{n}=O_{p}(\sqrt{\log n})$. Thus we have confirmed
that $\mathbb{G}_{P}(f_{n,x}/\sigma_{n}(x))$ is a separable Gaussian
process with mean zero and unit variance whose supremum is finite
almost surely, which is what \citet[Corollary 2.1]{CCK14-AoS} requires
in Step 4.

Step 2 has one more thing to settle: Assumption \ref{assu: band}
is stated against $\bar{\sigma}$ rather than $\sigma_{n}$, and the
two differ, while Step 3 needs to pass between them. The second moment
computation in the proof of Lemma \ref{lem:boundary} gives 
\[
\sigma^{2}_{n}(x)=\bar{\sigma}^{2}(x)+O(h_{x})\qquad\text{uniformly in }x\in\mathcal{X}_{0},
\]
 so $\sup_{x}\left|\sigma_{n}(x)/\bar{\sigma}(x)-1\right|=O(h_{x})$.
The remainder is $O(h_{x})$ and not $O(h^{m}_{x})$ because the kernel
enters squared. What survives the change of variable at order $h_{x}$
is $\int vK(v)^{2}\mathrm{d}v$ times a derivative, and it is the
first moment of $K^{2}$ that is at issue, not the vanishing moments
of $K$. Assumption \ref{assu: kernel} does not impose symmetry,
so only $O(h_{x})$ is available; for a symmetric kernel the term
vanishes and the remainder is $O(h^{2}_{x})$.

\emph{Step 3.} We will derive a bound for 
\[
\left|\check{Z}_{n}-\widetilde{Z}_{n}\right|\le\left|\check{Z}_{n}-Z_{n}\right|+\left|Z_{n}-\widetilde{Z}_{n}\right|.
\]
Step 1 already derived the bound for $\left|Z_{n}-\widetilde{Z}_{n}\right|$
given by (\ref{eq:thm3.8_step1_bound}). Note that, for nonnegative
$\beta_{x}$ and positive $c_{x}$, one has 
\[
\left|\sup_{x}\beta_{x}c_{x}-\sup_{x}\beta_{x}\right|\le\sup_{x}\beta_{x}\left|c_{x}-1\right|\le\left(\sup_{x}\beta_{x}\right)\sup_{x}\left|c_{x}-1\right|.
\]
Apply this with $\beta_{x}=\sqrt{nh_{x}}\left|\widehat{g}_{\mathrm{LS}}(x)-g(x)\right|/\sigma_{n}(x)$
and $c_{x}=\sigma_{n}(x)/\hat{s}_{n}(x)$, for which $\sup_{x}\beta_{x}c_{x}=\check{Z}_{n}$
and $\sup_{x}\beta_{x}=Z_{n}$, to get 
\[
\left|\check{Z}_{n}-Z_{n}\right|\le Z_{n}\sup_{x}\left|\sigma_{n}(x)/\hat{s}_{n}(x)-1\right|.
\]
Note that $\sup_{x}\left|\sigma_{n}(x)/\hat{s}_{n}(x)-1\right|$ is
bounded in two pieces.  Write $\frac{\sigma_{n}(x)}{\hat{s}_{n}(x)}=\frac{\sigma_{n}(x)}{\bar{\sigma}(x)}\cdot\frac{\bar{\sigma}(x)}{\hat{s}_{n}(x)}$.
The first factor differs from one by $O(h_{x})$ uniformly in $x$
by Step 2. The second factor is controlled by $\Delta_{n}=\sup_{x}\left|\hat{s}_{n}(x)/\bar{\sigma}(x)-1\right|$,
defined before the proof.  Assumption \ref{assu: band} (i) makes
$\Delta_{n}=o_{P}(1/\log n)$, so the event $\Delta_{n}\le1/2$ has
probability tending to one.  On that event $\hat{s}_{n}(x)\ge\bar{\sigma}(x)/2$
for every $x$, so $\left|\bar{\sigma}(x)/\hat{s}_{n}(x)-1\right|\le2\Delta_{n}$.
The two bounds multiply to give $\sup_{x}\left|\sigma_{n}(x)/\hat{s}_{n}(x)-1\right|\le C(\Delta_{n}+h_{x})$.
Since Steps 1 and 2 give $Z_{n}=O_{P}(\sqrt{\log n})$, we have
\[
\left|\check{Z}_{n}-Z_{n}\right|=O_{P}((\Delta_{n}+h_{x})\sqrt{\log n}).
\]
Therefore, it follows that 
\begin{align*}
\left|\check{Z}_{n}-\widetilde{Z}_{n}\right| & \le C\delta_{n}(\gamma_{n}),
\end{align*}
where 
\[
\delta_{n}(\gamma_{n})=\sqrt{h_{x}}+\sqrt{nh_{x}}R_{n}+\sqrt{nh_{x}}\bar{h}^{m}+(nh_{x})^{-1/8+\upsilon}+(\Delta_{n}+h_{x})\sqrt{\log n}
\]
with probability at least $1-C^{\prime}(\gamma_{n}+n^{-1})$, and
the conditional analogue for $\check{Z}^{*}_{n}$ with $\Delta^{*}_{n}$
in place of $\Delta_{n}$.

\emph{Step 4.} Steps 1 to 3 bound $\left|\check{Z}_{n}-\widetilde{Z}_{n}\right|$
in probability. This step turns that into a bound on the distance
between the laws of $\check{Z}_{n}$ and $\widetilde{Z}_{n}$. \citet[Corollary 2.1]{CCK14-AoS}
applies to $\mathbb{G}_{P}(f_{n,x}/\sigma_{n}(x))$ by Step 2 and
bounds the Lévy concentration function of its supremum, 
\[
\sup_{c\in\mathbb{R}}P\left(\left|\widetilde{Z}_{n}-c\right|\le t\right)\le4t\left(\mathbb{E}[\widetilde{Z}_{n}]+1\right)=O(t\sqrt{\log n}).
\]
 An interval of length $2t$ lies in the ball of radius $t$ about
its midpoint, so the distribution function of $\widetilde{Z}_{n}$
is Lipschitz with a constant $L_{n}=O(\sqrt{\log n})$, and is therefore
continuous. By \citet[Lemma 2.1]{CCK16-SPA}, if two random variables
differ by at most $\delta$ with probability at least $1-\xi$, their
distribution functions differ in the Kolmogorov distance by at most
the concentration function of either at radius $\delta$, plus $\xi$,
which the Lipschitz bound makes $L_{n}\delta+\xi$. Applying this
to $\check{Z}_{n}$ and $\widetilde{Z}_{n}$, with $\delta=C\delta_{n}(\gamma_{n})$
and $\xi=C^{\prime}(\gamma_{n}+n^{-1})$ from Step 3, 
\[
\rho_{n}:=\sup_{c}\left|P(\check{Z}_{n}\le c)-P(\widetilde{Z}_{n}\le c)\right|\le C\left\{ L_{n}\delta_{n}(\gamma_{n})+\gamma_{n}+n^{-1}\right\} .
\]
Likewise, for $\rho^{*}_{n}$ under the conditional law of $\check{Z}^{*}_{n}$,
\citet[Remark 2.2]{CCK16-SPA} converts the unconditional statement
of Lemma \ref{lem:boot_coupling} into a conditional one at the cost
of a factor that may be taken to vanish slowly.

\emph{Step 5.} We will show that $\left|P\left(\check{Z}_{n}\le c_{n}(1-\alpha)\right)-(1-\alpha)\right|$
is bounded by $\rho_{n}+\rho^{*}_{n}$. Since $c_{n}(1-\alpha)$ is
the conditional $(1-\alpha)$ quantile of $\check{Z}^{*}_{n}$, $P^{*}(\check{Z}^{*}_{n}\le c_{n})\ge1-\alpha$
and $P^{*}(\check{Z}^{*}_{n}<c_{n})\le1-\alpha$, whence $1-\alpha-\rho^{*}_{n}\le P(\widetilde{Z}_{n}\le c_{n})\le1-\alpha+\rho^{*}_{n}$.
Since $\widetilde{Z}_{n}$ has a continuous distribution function,
this places $c_{n}$ between the deterministic quantiles $q_{1-\alpha-\rho^{*}_{n}}$
and $q_{1-\alpha+\rho^{*}_{n}}$ of $\widetilde{Z}_{n}$. Evaluating
$\rho_{n}$ at those deterministic points, and using the Lipschitz
property to move between them, gives 
\[
\left|P\left(\check{Z}_{n}\le c_{n}(1-\alpha)\right)-(1-\alpha)\right|\le\rho_{n}+\rho^{*}_{n}.
\]

\emph{Step 6.} It remains to check that $\rho_{n}+\rho^{*}_{n}$ tends
to zero. The value $\gamma_{n}$ was chosen so that $L_{n}$ times
the sample-side coupling term equals it, $L_{n}(nh_{x})^{-1/8}(\log n)^{3/8}=(nh_{x})^{-1/8}(\log n)^{7/8}=\gamma_{n}$.
Then 
\[
\rho_{n}+\rho^{*}_{n}=O_{p}\left((h_{x}\log n)^{1/2}+\sqrt{nh_{x}}R_{n}(\log n)^{1/2}+\sqrt{nh_{x}}\bar{h}^{m}(\log n)^{1/2}+(nh_{x})^{-1/8+\upsilon}+\left(\Delta_{n}+\Delta^{*}_{n}\right)\log n\right).
\]
 Each summand tends to zero. The first because $h_{x}$ is a power
of $n$. The second because Assumption \ref{assu:rate} fixes the
bandwidths at powers of $n$, so each summand of $\sqrt{nh_{x}}R_{n}$
is $n^{-b}(\log n)^{k}$ with $b>0$ and $k\le1$, and multiplying
by $(\log n)^{1/2}$ leaves the limit zero. The third because Assumption
\ref{assu:rate} gives $\sqrt{nh_{x}}h^{m}_{y}\rightarrow0$ and $\sqrt{nh_{x}}h^{m}_{z}\rightarrow0$,
the undersmoothing condition of Theorem \ref{thm:pointwise} gives
$\sqrt{nh_{x}}h^{m}_{x}\rightarrow0$, and each of the three is again
a power of $n$, which absorbs the $(\log n)^{1/2}$. The fourth because
$nh_{x}$ grows as a power of $n$. The fifth by Assumption \ref{assu: band}
(i) and (ii), which make $\Delta_{n}\log n$ and $\Delta^{*}_{n}\log n$
$o_{p}(1)$, which proves the claim.

$\hat{g}_{\mathrm{LAD}}$ is asymptotically first-order equivalent
to $\hat{g}_{\mathrm{LS}}$. Appendix \ref{sec:thm1} bounds the difference
between the two estimators by $\Delta^{\mathrm{LAD}}_{n}$ with $\sqrt{nh_{x}}\Delta^{\mathrm{LAD}}_{n}\sqrt{\log n}\rightarrow0$,
which absorbs it into $\delta_{n}(\gamma_{n})$.
\end{proof}

\section{Bias and standard deviation as the sample grows }\label{sec:Consistency}

Table \ref{tab:cons} reports the absolute bias and the standard deviation
of each estimator at $n=250$, $500$, $1000$ and $2000$. The last
column is the slope of a regression of the logarithm of each entry
on $\log n$.

There are five arms. The first two are our estimators, with the bandwidth
sequence of Section \ref{sec:Monte-carlo} at $c=0.60$. The third
is the LAD estimator with the bandwidth held fixed at $c=0.60$ for
every $n$. The last two are the series estimator, once at $p=3$
and once with $p$ chosen by cross-validation in each sample.

\begin{table}[t]
\centering
\begin{tabular}{lccccc}
\hline
 & $n=250$ & $500$ & $1000$ & $2000$ & rate\\
\hline
\multicolumn{6}{l}{\emph{absolute bias}}\\
LS, $c(n)$ shrinking & $0.095$ & $0.101$ & $0.069$ & $0.070$ & $n^{-0.19}$\\
LAD, $c(n)$ shrinking & $0.075$ & $0.058$ & $0.049$ & $0.041$ & $n^{-0.28}$\\
LAD, $c=0.60$ fixed & $0.062$ & $0.046$ & $0.047$ & $0.049$ & $n^{-0.10}$\\
NPV, $p$ by CV & $0.072$ & $0.054$ & $0.062$ & $0.061$ & $n^{-0.05}$\\
NPV, $p=3$ & $0.062$ & $0.056$ & $0.065$ & $0.065$ & $n^{+0.05}$\\
\hline
\multicolumn{6}{l}{\emph{standard deviation}}\\
LS, $c(n)$ shrinking & $0.236$ & $0.209$ & $0.195$ & $0.154$ & $n^{-0.19}$\\
LAD, $c(n)$ shrinking & $0.163$ & $0.145$ & $0.110$ & $0.080$ & $n^{-0.35}$\\
LAD, $c=0.60$ fixed & $0.199$ & $0.149$ & $0.107$ & $0.075$ & $n^{-0.47}$\\
NPV, $p$ by CV & $0.204$ & $0.133$ & $0.101$ & $0.068$ & $n^{-0.51}$\\
NPV, $p=3$ & $0.189$ & $0.126$ & $0.093$ & $0.058$ & $n^{-0.56}$\\
\hline
\end{tabular}
\caption{Absolute bias and standard deviation, averaged over 21 points on $[-2.5,2.5]$, 100 replications. The last column is the slope of $\log$ of the entry on $\log n$.}
\label{tab:cons}
\end{table}

The arms fall into two groups. The bias falls with the sample in the
first two and is flat in the other three. The standard deviation falls
in all five, and fastest in the series estimator. This pattern is
consistent with convergence to a limit other than $g$.

The three flat arms are not flat for the same reason. The series estimator
is flat because the control function it uses is not valid when the
first stage is nonseparable. No choice of $p$ repairs that, and cross-validation,
which chooses $p$ well by its own criterion, does not repair it either.
The LAD arm is flat because its bandwidth was held fixed; the second
row shows that letting the bandwidth shrink removes the bias. One
is a property of the estimator, the other is a choice of tuning.

The fixed-bandwidth arm is a realistic comparison rather than a weak
one. Table \ref{tab:bw} selects $c=0.60$ at $n=1000$, so it is
the value a researcher who chooses the bandwidth by mean squared error
at that sample size would use. The same researcher would hold it there
if the choice were not revisited as the sample grew. Consistency requires
the bandwidth to go to zero, which is a different thing from minimizing
mean squared error at one sample size.

\printbibliography

\end{document}